\documentclass[a4paper,twocolumn,11pt,accepted=2025-12-16]{quantumarticle}
\pdfoutput=1
\usepackage{amsmath,amssymb,amsfonts,enumitem}
\usepackage{amsthm}
\usepackage{graphicx}
\usepackage{textcomp}
\usepackage[dvipsnames]{xcolor}
\usepackage{braket}
\usepackage{mathtools}
\usepackage{bm}
\usepackage[ruled,vlined]{algorithm2e}
\usepackage{algcompatible}

\usepackage{pifont}

\usepackage[normalem]{ulem}
\usepackage[colorlinks=true,citecolor=blue,linkcolor=magenta]{hyperref}
\usepackage{graphicx}
\usepackage{tikz}

\usepackage{tabularray} 
\usepackage[dvipsnames]{xcolor}

\newtheorem{theorem}{Theorem}  
\newtheorem{proposition}[theorem]{Proposition}  
\newtheorem{lemma}[theorem]{Lemma}

\newtheorem{example}{Example}[section]

\newcommand{\mg}[1]{{\color{orange}}}

\definecolor{bhblue}{HTML}{0096FF}

\definecolor{tableblue}{HTML}{ccffff}

\newcommand\mtiny[1]{\mbox{\tiny\ensuremath{#1}}}
\def\sn{s^{{\mtiny{(N)}}}_k}
\usepackage{minitoc}
\usepackage{tocloft}

\begin{document}

\doparttoc 
\faketableofcontents 
\part{}

\title{Double-bracket algorithm for quantum signal processing without post-selection
}
\newcommand{\EPFL}{Institute of Physics, \'{E}cole Polytechnique F\'{e}d\'{e}rale de
Lausanne (EPFL), Lausanne, Switzerland}
\newcommand{\KEIO}{Quantum Computing Center, Keio University, Hiyoshi 3-14-1, Kohoku-ku, Yokohama 223-8522, Japan}
\author{Yudai Suzuki}
\affiliation{\EPFL}\affiliation{\KEIO}
\author{Bi Hong Tiang}
\affiliation{School of Physical and Mathematical Sciences, Nanyang Technological University, 637371, Singapore}
\author{Jeongrak Son}
\affiliation{School of Physical and Mathematical Sciences, Nanyang Technological University, 637371, Singapore}
\author{Nelly H. Y. Ng }
\affiliation{School of Physical and Mathematical Sciences, Nanyang Technological University, 637371, Singapore}
\affiliation{Centre for Quantum Technologies, Nanyang Technological University, 637371, Singapore}
\author{Zo\"{e} Holmes}
\affiliation{\EPFL}
\author{Marek Gluza}
\email{marekludwik.gluza@ntu.edu.sg}
\affiliation{School of Physical and Mathematical Sciences, Nanyang Technological University, 637371, Singapore}

\begin{abstract}
Quantum Signal Processing (QSP), a framework for implementing matrix-valued polynomials, is a fundamental primitive in various quantum algorithms.
Despite its versatility, a potentially underappreciated challenge is that all systematic protocols for implementing QSP rely on post-selection.
This can impose prohibitive costs for tasks when amplitude amplification cannot sufficiently improve the success probability. For example, in the context of ground-state preparation, this occurs when using a too poor initial state.
In this work, we introduce a new formula for implementing QSP transformations of Hermitian matrices, which requires neither auxiliary qubits nor post-selection.
Rather, using approximation to the exact unitary synthesis, we leverage the theory of the double-bracket quantum algorithms to provide a new quantum algorithm for QSP, termed Double-Bracket QSP (DB-QSP).
The algorithm requires the energy and energetic variance of the state to be measured at each step and has a recursive structure, which leads to circuit depths that can grow super exponentially with the degree of the polynomial.
With these strengths and caveats in mind, DB-QSP should be viewed as complementing the established QSP toolkit. In particular, DB-QSP can deterministically implement low-degree polynomials to ``warm start" QSP methods involving post-selection. 
\end{abstract}
\maketitle
\section{Introduction}

The efficient implementation of matrix-valued functions plays a central role in the design of modern quantum algorithms~\cite{equiangular}. 
That is, the essence of many quantum algorithms boils down to constructing a polynomial function $ p(H)$ of a given Hermitian matrix $H$ and applying it to an input state $\ket{\Psi}$ to obtain a normalized state
\begin{align}
\label{ph}
    \ket{\Psi'}  =
\frac{ p(H)\ket\Psi}{ \|p(H)\ket\Psi\|}\ 
\end{align}
with $\|\ket{\psi}\|=\sqrt{\braket{\psi|\psi}}$ for any vector $\ket{\psi}$.
For example, real and imaginary time evolution correspond to the transformations  $p(H)\approx\exp(iH t)$ and $p(H)\approx\exp(- \tau H)$, while matrix inversion implements the transformation $p(H)\approx H^{-1}$. Quantum Signal Processing (QSP) is an algorithmic framework for realizing such polynomial transformations on quantum computers.
QSP has enabled the development of advanced quantum algorithms for solving linear systems of equations~\cite{martyn2021grand,gilyen2019quantum,harrow2009quantum}, Hamiltonian simulation~\cite{Low2019hamiltonian,low2017optimal,low2017quantum}, and ground state preparation~\cite{ge2019faster,lin2020near}.

Despite its versatility, an underappreciated challenge in QSP is the cost of post-selection~\cite{zhang2025quantum}.
For QSP implementation methods such as qubitization~\cite{Low2019hamiltonian} and Linear Combination of Unitaries (LCU)~\cite{gui2008duality,Childs2012LCU,chakraborty2024implementing} to be practically viable, the success probability for post-selection must be sufficiently high to avoid excessive resource overhead.
While amplitude amplification techniques can improve success probabilities~\cite{MR1947332,berry2014exponential}, they may be insufficient when the success probability is exponentially small in the number of qubits~\cite{sze2025hamiltonian}.
For instance, ground-state preparation algorithms with nearly optimal resource scaling may still incur exponential costs if the initial overlap between the input state and ground state is exponentially small~\cite{sze2025hamiltonian,lin2020near}.

In this work, we propose a new QSP implementation that eliminates the need for auxiliary qubits and post-selection.
Since the state after normalization in Eq.~\eqref{ph} is a genuine quantum state, there must exist a unitary operator $U_\Psi$ such that $U_\Psi\ket \Psi=p(H)\ket\Psi/\|p(H)\ket\Psi\|$.
This work identifies how to systematically perform the unitary synthesis of $U_\Psi$.
The key insight is that any linear polynomial, $aI+bH$  with real coefficients $a,b\in\mathbb{R}$ and the identity operator $I$, can be \textit{exactly} represented by a unitary operator in form of
\begin{align} \label{eq:u_psi_for_linear_poly}
 \frac{(aI +bH)\ket{\Psi}}{\|(aI+bH)\ket\Psi\|}   =e^{s[\ket{\Psi}\bra{\Psi}, H]}\ket\Psi\ ,
\end{align} 
where $s$ is determined by the energy mean and variance.
Using this building block, we prove that a recursion involving unitaries in Eq.~\eqref{eq:u_psi_for_linear_poly} together with state-dependent reflection gates can realize \textit{arbitrary} polynomial functions (see Fig.~\ref{fig:summary}). 
We then utilize the recently-established theory of Double-Bracket Quantum Algorithms (DBQA)~\cite{gluza2024doublebracket,gluza_DB_QITE_2024} to derive a unitary synthesis that can be compiled into primitive gates using standard quantum computing methods.
This leads to a new quantum algorithm which we call the Double-Bracket QSP 
(DB-QSP).


\begin{figure*}[t]
\centering
\begin{tikzpicture}
\definecolor{lightgray}{HTML}{F4F4F4}
\definecolor{littlelightgray}{HTML}{ecececff}
\definecolor{pale}{HTML}{7ca3d4ff}
\definecolor{lightred}{HTML}{d8a2a2}
\node[text width=7cm] at (-2.2,3.0){(a) Original QSP with Post-Selection~\cite{equiangular}};
\node[anchor=center] (russell) at (-1.7,0.3)
{\centering\includegraphics[width=0.4\textwidth]{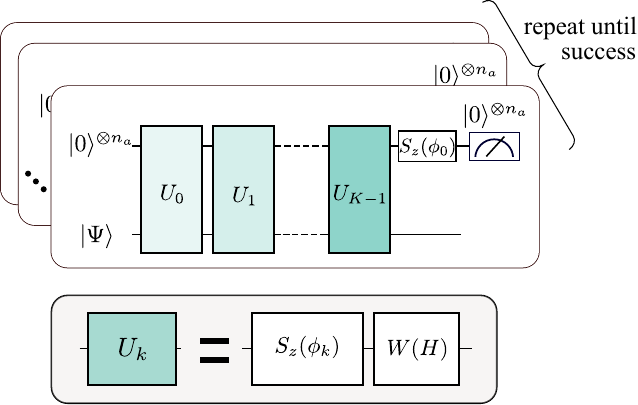}};
\node[text width=9cm] at (7.0
,3.0){(b) Our Proposal without Post-Selection (Thm.~\ref{thm complex QSP})};
    \node[anchor=center] (russell) at (6.7,0.1){\centering\includegraphics[width=0.4\textwidth]{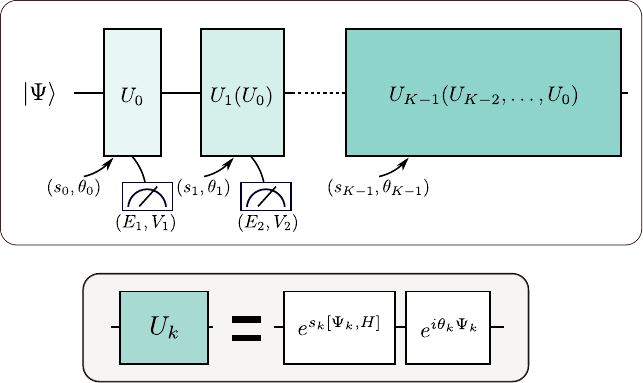}};
\end{tikzpicture}
\caption{\textbf{Quantum Signal Processing (QSP) without auxiliary qubits and post-selection.}
We introduce a new formula for implementing QSP of Hermitian matrices (Thm.~\ref{thm complex QSP}).
(a) To realize a degree-$K$ polynomial of a Hermitian matrix $H$, original QSP performs measurement on auxiliary qubits so that the desired transformation is realized, as shown in Eq.~\eqref{eq:qsp_conv_uni_qubitize}.
(b) In contrast, our formula does not require auxiliary qubits and accordingly the post-selection. Instead, we recursively apply the state-dependent unitary operators $e^{i\theta_k \Psi_{k}}e^{s_{k}[\Psi_{k},H]}$ with $\ket{\Psi_{k+1}}=e^{i\theta_{k} \Psi_{k}}e^{s_{k}[\Psi_{k},H]}\ket{\Psi_{k}}$, resulting in the circuit depth that grow significantly in the degree of polynomials $K$.
Furthermore, to determine the time duration $s_{k}$ and phase $\theta_{k}$, energy $E_{k}=\braket{\Psi_{k}|H|\Psi_{k}}$ and variance in energy $V_{k}=\braket{\Psi_{k}|H^2|\Psi_{k}}-E_{k}^2$ must be known at each step.}
\label{fig:summary}
\end{figure*}


The advantages of DB-QSP come with two challenges.
First, its recursive form means the depth of circuit required to converge with arbitrary precision grows super exponentially with the degree of the target polynomial functions.
However, low-degree approximation techniques~\cite{gilyen2019quantum,martyn2021grand} can be applied to keep the circuits depths efficient in certain cases. 
The second limitation is the need to estimate the energy and variance in energy of the state at each iteration in order to compute the step size used in the circuit at the next iteration.
However, when the degree of polynomials scales logarithmically in the inverse of the desired precision, the corresponding sampling overhead should only be polynomial.

DB-QSP can be used both as a standalone method and as a tool in conjunction with other QSP methods~\cite{Low2019hamiltonian, Childs2012LCU,chakraborty2024implementing}. In particular, it can be viewed as a (partial) alternative when the post-selection overhead of other QSP methods are prohibitively large. 
Namely, DB-QSP provides a deterministic approach to drive a state closer to a target state, such as an approximate ground state, regardless of the quality of the initial state. 
Conversely, in conventional QSP methods~\cite{Low2019hamiltonian, Childs2012LCU, chakraborty2024implementing}, a low post-selection success probability could prevent systematic improvements.
Thus DB-QSP can provide a warm-starting procedure, i.e., a means of preparing approximate initial states, for existing methods.

\section{Preliminaries} \label{sec:preliminary}
\subsection{Overview of Quantum Signal Processing (QSP)}

QSP is a framework for systematically constructing matrix-valued functions on quantum computers. The goal of QSP is to perform degree-$K$ polynomial transformation $p(H)$ of a Hermitian matrix $H$ to a $n$-qubit input state $\ket{\Psi}$ up to normalization (Eq.~\eqref{ph}). Sometimes, the implementation methodology proposed in Ref.~\cite{equiangular} itself is referred to as ``QSP". 
However, Eq.~\eqref{ph} can be achieved also via alternative techniques, e.g., Linear Combination of Unitaries (LCU)~\cite{Childs2012LCU,chakraborty2024implementing}; see App.~\ref{app:overview_qsp} for a detailed overview.
In this manuscript, we use ``QSP” to refer to the concept of implementing the polynomial functions, and distinguish it from the methodology in Ref.~\cite{equiangular,Low2019hamiltonian} by referring to the latter as ``qubitization”.

Qubitization uses a circuit $U_{Q}$ comprised of two types of operators: \textit{signal operators} ${W}$ and \textit{signal processing operators} ${S}(\phi)$, where the phase $\phi$ is drawn from a set $\{\phi_k\}$.
The desired polynomial transformation is obtained by performing a measurement in the so-called \textit{signal basis}.
Concretely, given the signal operator $W(H)$ of a Hermitian matrix $H$ with $\|H\|\le 1$ and the signal processing operator $S_{z}(\phi)$, there exists a sequence of QSP phase $\{\phi_{k}\}$ such that the following circuit
\begin{equation} \label{eq:qsp_conv_uni_qubitize}
U_{Q} = S_{z}(\phi_{0}) \prod_{k=1}^{K} W(H)S_{z}(\phi_{k})\ ,
\end{equation}
followed by measurement in the basis $M=\{\ket{+},\ket{-}\}$ can realize a degree-$K$ real polynomial $p(H)$.
The signal operator $W(H)$ can be constructed using block-encoding~\cite{chakraborty_et_al:LIPIcs.ICALP.2019.33}, which embeds a Hermitian matrix $H$ into the top-left block of a larger unitary matrix as 
\begin{equation*} \label{eq: be_simplified_}
    W(H) = \begin{bmatrix}
        H & i\sqrt{I-H^2} \\
        i\sqrt{I-H^2} & H \\
\end{bmatrix}.
\end{equation*}
The signal processing operator $S_{z}(\phi)$,
\begin{equation*}
    S_{z}(\phi) = e^{i\phi Z} = \begin{bmatrix}
        e^{i\phi} & 0 \\
        0 & e^{-i\phi} \\
\end{bmatrix},
\end{equation*}
then acts on an auxiliary qubit.
We provide details of the achievable functions via this technique in App.~\ref{app:overview_qsp}.

QSP has led to asymptotically optimal Hamiltonian simulation algorithms~\cite{Low2019hamiltonian} and a near-optimal method for ground-state preparation~\cite{lin2020near}.
Furthermore, it serves as a fundamental tool for constructing primitive quantum algorithms that exhibit quantum advantages~\cite{martyn2021grand,gilyen2019quantum}.
Therefore, its efficiency in implementing linear algebraic operations and its role as a key building block for quantum algorithms have made QSP a subject of significant interest.

\subsection{The Role of Post-Selection in Existing QSP Methods}

Despite their versatility, existing QSP implementations face several challenges such as difficulty in finding angles~\cite{alase2025quantum} and demanding implementation costs for block-encodings~\cite{kikuchi2023realization}.
As shown above, qubitization performs the measurement in the signal basis to post-select for the desired transformation.
When this post-selection in qubitization is unsuccessful, it is possible to simply repeat the experiment until a successful implementation eventually appears.
Amplitude amplification techniques~\cite{MR1947332} can often enhance success probabilities. For instance, Hamiltonian simulation benefits from this combination of techniques~\cite{berry2014exponential}.
However, in some cases, the success probability for QSP could be exponentially small in the number of qubits~\cite{sze2025hamiltonian}. 
For example, the successful probability of ground-state preparation can be prohibitively small if not initialized with a sufficiently good input state. 

We illustrate the issue using an example of general qubitization.
Given an input state $\ket\Psi$, the number of auxiliary qubits $n_{a}$ and $\alpha\in\mathbb{R}$, applying $U_{Q}$ in Eq.~\eqref{eq:qsp_conv_uni_qubitize} to $\ket\Psi$ yields
\begin{align}
\ket{0}^{\otimes n_{a}} \otimes p(H/\alpha)\ket{\Psi} + \ket{\text{garbage}^\perp},
\end{align}
where $\ket{\text{garbage}^\perp}$ is an orthogonal state, i.e., $\ket{\text{garbage}^\perp} \perp \ket{0}^{\otimes n_{a}} \otimes \frac{H}{\alpha}\ket{\Psi}$.
The probability of projecting onto $\ket{0}^{\otimes n_{a}}$ is given by 
\begin{equation}
    p_\text{succ} = \|p(H/\alpha)\ket{\Psi}\|^2 \, , 
\end{equation}
which can be exponentially small. 
For example, in the case of Imaginary-Time Evolution (ITE), where $p(H) \approx e^{-\tau H}$, the success probability scales with the overlap of the initial state and the corresponding thermal state, which can decay exponentially~\cite{chan2023simulating,sze2025hamiltonian,silva2023fragmented}.
More generally, this dependence on state fidelity persists across various scenarios.
For instance, the block-encoding query complexity for nearly-optimal ground-state preparation algorithm in Ref.~\cite{lin2020near} scales as $O(\alpha/\gamma)$, where $\gamma = |\langle \lambda_0|\Psi\rangle|^2$ is the fidelity of the input state $\ket{\Psi}$ with the ground state $\ket{\lambda_0}$ of $H$.
The query scaling $O(\alpha/\gamma)$ corresponds to the inverse success probability and thus requires repeated trials for obtaining a successful outcome.
This indicates that the success of the block-encoding depends on the input state. 
Additionally, since the number of queries to the block-encoding unitary scales with the degree of polynomials as shown in Eq.~\eqref{eq:qsp_conv_uni_qubitize}, the degree $K$ needs to be sufficiently low to ensure successful post-selection each time.
For more details and a discussion of similar challenges when using LCU for implementing QSP, see App.~\ref{app:overview_qsp}.

\section{Main Results}
\label{sec main results}

\subsection{Overview of Analysis} \label{sec:overview_main_result}

In this section, we present an algorithm for QSP that requires neither auxiliary qubits nor post-selection. Our key insight, captured in Lem.~\ref{thm zero overhead QSP unitary single step} in Sec.~\ref{sec:unitary_synthesis_linear_polynomials}, is that there exists a unitary that \textit{exactly} implements the normalized action of the linear polynomial $H-\alpha I$ on an input state $\ket{\Psi}$ for any real $\alpha$. We then show how repeated applications of this circuit to apply the linear polynomial can be used to implement any polynomial with real roots. 

Sec.~\ref{sec:unitary_synthesis_arbitrary_polynomial} tackles the extension to polynomials with complex roots. 
This leads to our main result, Thm.~\ref{thm complex QSP}, which demonstrates that interleaving the unitary sequence $U_{\Psi}$ from Lem.~\ref{thm zero overhead QSP unitary single step} with state-dependent reflection gates enables the realization of arbitrary polynomials. 

Sec.~\ref{sec:DB-QSP_implementation} introduces a method to implement the unitary sequence in Thm.~\ref{thm complex QSP} called Double-Bracket QSP (DB-QSP), which performs general QSP without post-selection. 
Namely, we show that the recently-developed DBQA framework provides a means to efficiently implement the exponentials of commutators that appear in Thm.~\ref{thm complex QSP}. 
Leveraging DBQA, we formulate DB-QSP outlined in Alg.~\ref{alg DBQSP}. 
We analyze the errors introduced by this implementation compared to the idealized scenario in Thm.~\ref{thm complex QSP} and show that circuit depths of DB-QSP scale super exponentially with the degree of the polynomial to be implemented. 

The DB-QSP algorithm (Alg.~\ref{alg DBQSP}) also requires the energy and energy variance of the state at each iteration to be estimated in order to compute the step size for the next iteration. On quantum hardware, statistical noise is inevitable due to the finite number of measurement shots. In Sec.~\ref{sec:error_analysis}, we analyze how this noise affects the accuracy of the constructed state.

To further examine the practical implications of these challenges, Sec.~\ref{sec:application} investigates the impact of circuit depth on applicability. Since the required depth depends on the polynomial degree, DB-QSP is limited to low-degree polynomials. We identify approximate ground-state preparation as a use case where DB-QSP can be practically useful.

Finally, Sec.~\ref{sec:hybrid strategy} discusses a hybrid strategy that integrates DB-QSP with existing methods such as variational quantum algorithms, quantum dynamic programming, qubitization and LCU. 
The circuit depth scaling of DB-QSP suggests that available experimental resources may be insufficient for certain tasks. 
However, even qubitization with amplitude amplification sometimes demands exponential costs. In such difficult cases, combining qubitization or LCUs with DB-QSP could reduce resource requirements.

\subsection{Main Tool: Unitary Synthesis for Polynomials with Real Roots without Post-Selection} \label{sec:unitary_synthesis_linear_polynomials}

In Sec.~\ref{sec:preliminary}, we reviewed a QSP implementation relying on post-selection.
An alternative is to find a unitary $U_\Psi$ satisfying
\begin{align}\label{eq:unitary_synthesis}
    U_\Psi\ket \Psi= \frac{p(H)\ket\Psi}{\|p(H)\ket\Psi\|}\ .
\end{align}
The following Lemma constructs a new tool that provides an explicit and exact construction of $U_\Psi$ through an exponential of a specific commutator for linear polynomials.
For simplicity, we hereafter use $\Psi$ as a shorthand for the density matrix representation of a pure state, i.e., $\Psi = \ket\Psi\bra\Psi$.

\begin{lemma}[Unitary synthesis for linear polynomials without post-selection]
\label{thm zero overhead QSP unitary single step} 
Suppose $p(H)=H-\alpha I$ is any linear polynomial of a Hermitian matrix $H$ with $\alpha\in\mathbb{R}$.
Given an input state $\ket{\Psi} $ with energy mean $E_\Psi = \bra\Psi H\ket\Psi$ and variance $V_{\Psi}=\bra\Psi H^2\ket\Psi -E_\Psi^2$, the unitary synthesis for $p(H)$ in Eq.~\eqref{eq:unitary_synthesis} can be achieved by 
    \begin{align}
    \label{DB bridge}
        U_\Psi = e^{s_\Psi[\Psi,H]},
    \end{align}
    with 
    \begin{equation} \label{eq:duration_monic_linear}
        s_\Psi = \frac{-1}{\sqrt{V_\Psi}}\arccos\left(\frac{E_{\Psi}-\alpha}{\sqrt{V_{\Psi}+\left(E_{\Psi}-\alpha\right)^2}}\right)\ .
    \end{equation}
\end{lemma}

A rigorous proof of Lem.~\ref{thm zero overhead QSP unitary single step} is provided in App.~\ref{app lemma and thm proofs}. 
Here, we present a proof sketch to clarify the derivation of the unitary operator in Eq.~\eqref{DB bridge}. First, we can see immediately that $U_\Psi$ is indeed unitary as claimed because the commutator $[\Psi,H]$ in its exponent is anti-Hermitian, i.e., $[\Psi,H] = -( [\Psi,H])^{\dagger}$.

Next, we derive Eq.~\eqref{DB bridge}, which establishes the equivalence between $e^{sW_{H}}$ with $W_{H}=[\Psi,H]$ and a linear polynomial applied to $\ket\Psi$ for $s\in\mathbb{R}$.
By definition, the unitary operator can be expressed as 
$e^{s W_H} = \sum_{k=0}^\infty\frac{s^k}{k!} W_H^k$ using all powers of~$W_H$.
However, when acting on $\ket{\Psi}$, we get
\begin{align}    \label{eq: first power of commutator}
W_H \ket{\Psi} = - (H-E_\Psi I)\ket{\Psi}\ ,
    \end{align}
    while for the second power
 \begin{align}
        W_H^2 \ket{\Psi} 
        &=  -(\bra\Psi H^2\ket\Psi - E_\Psi^2)\ket\Psi  =- V_{\Psi}\ket{\Psi}   \ .\label{state idempotence}
    \end{align}
 This shows that the square of $W_H$ leaves $\ket \Psi$ unchanged up to a rescaling prefactor. 
Thus, by substituting Eqs.~\eqref{eq: first power of commutator},~\eqref{state idempotence} into the series expansion, the resulting state can be simplified to  
\begin{align}
     \label{eq: aI+bH main_}
     e^{s_\Psi W_H}\ket\Psi &= \left(a(s_\Psi)I + b(s_\Psi)H \right)\ket{\Psi},
\end{align}
with real-valued coefficients $a(s_{\Psi}),b(s_{\Psi})$ corresponding to any duration $s_{\Psi}\in \mathbb R$ given by
\begin{align}
    a(s_{\Psi}) &= \frac{ E_\Psi }{\sqrt{V_{\Psi}}}\sin\left(s_{\Psi} \sqrt{V_{\Psi}}\right) + \cos\left(s_{\Psi} \sqrt{V_{\Psi}}\right), 
\\
        b(s_{\Psi}) &= - \frac{1}{\sqrt{V_{\Psi}}} \sin\left(s_{\Psi} \sqrt{V_{\Psi}}\right)\ .
    \end{align}
Here, the derivation exploits the Taylor series of trigonometric functions.
Finally, by solving the equations $a(s_\Psi)=-\alpha/\|p(H)\ket{\Psi}\|$ and $b(s_\Psi)=1/\|p(H)\ket{\Psi}\|$, we obtain Eq.~\eqref{eq:duration_monic_linear}, the time duration $s_{\Psi}$ to realize Eq.~\eqref{eq:unitary_synthesis} for any linear polynomial.

Lem.~\ref{thm zero overhead QSP unitary single step} indicates that there exists a duration $s$ such that the exponential of the commutator $e^{sW_{H}}$ with $W_{H}=[\Psi,H]$ can realize any linear real polynomial.
Importantly, the duration $s$ can be found by precise measurements of the energy and variance of the state $\Psi$.
\medskip

Higher order polynomials can then be realised by repeated applications of Lem.~\ref{thm zero overhead QSP unitary single step}.
The fundamental theorem of algebra shows that a polynomial of degree $K$ with real roots can be represented as $p(H) = a_K \prod_{k=1}^K(H-\alpha_k I)$ with $\alpha_{k}\in \mathbb{R}$.
This implies that such polynomials can be obtained by implementing Eq.~\eqref{DB bridge} with the corresponding factors $K$ times,
\begin{equation}
\label{eq dbi stacked}
    \frac{p(H)\ket{\Psi_0}}{ \|p(H)\ket{\Psi_0}\|} = \prod_{k=0}^{K-1} e^{s_{k}[\Psi_{k},H]} \ket{\Psi_{0}}\ ,
\end{equation}
where we start with an input state $\ket{\Psi_{0}}$ and define $\ket{\Psi_{k+1}}=e^{s_{k}[\Psi_{k},H]}\ket{\Psi_{k}}$ using $s_{k}$ in Eq.~\eqref{eq:duration_monic_linear}.

We stress that  Eq.~\eqref{eq dbi stacked} only implements functions with real roots. Nonetheless, many functions, such as Chebyshev polynomials, have only real roots. Hence Eq.~\eqref{eq dbi stacked} can be used for applications including approximations of ITE; see App.~\ref{sec:application_app} for the detail. 
However, Eq.~\eqref{eq dbi stacked} alone cannot construct \textit{arbitrary} polynomial functions, as the roots can be complex in general. We will now proceed to discuss how to extend Eq.~\eqref{eq dbi stacked} to implement polynomials with complex roots. 

\subsection{Main Result: Unitary Synthesis for Arbitrary Polynomials without Post-Selection} \label{sec:unitary_synthesis_arbitrary_polynomial}

In this section we show how Eq.~\eqref{eq dbi stacked} can be generalized to implement any \textit{arbitrary} polynomial of the form
\begin{align}\label{eq:pH_}
     p(H) = a_K \prod_{k=1}^K(H-z_k I)\ ,
 \end{align}
where the roots can be complex, i.e., $z_{k}\in\mathbb{C}$.
A core idea is that introducing a state-dependent reflection gate $e^{i\theta_{\Psi}\Psi}$ right after $U_{\Psi}$ in Eq.~\eqref{DB bridge} can realize any complex number $z$. 
That is, for any $z\in\mathbb{C}$, we obtain
\begin{align}
 \frac{(H-zI)\ket{\Psi}}{\|(H-zI)\ket\Psi\|}   =e^{i\theta_{\Psi}\Psi}e^{s_{\Psi}[\Psi, H]}\ket\Psi.
\end{align}
Using this technique, we derive a unitary synthesis formula for QSP without the need for the auxiliary qubits and post-selection, which is the main result of this work. 
The proof is provided in App.~\ref{app lemma and thm proofs}.

\begin{theorem}[Unitary synthesis for QSP without post-selection]
\label{thm complex QSP}
Consider an input state $\ket{\Psi_0}$ and any polynomial $p(H)$ of degree $K$ for a given Hermitian matrix $H$ in the form of Eq.~\eqref{eq:pH_}.
Given energy mean $E_k = \bra{\Psi_{k}} H\ket{\Psi_{k}}$ and variance $V_{k}=\bra{\Psi_{k}} H^2\ket{\Psi_{k}} -E_k^2$, the unitary synthesis in Eq.~\eqref{eq:unitary_synthesis} can be achieved by 
\begin{equation} \label{eq:qsp_wo_ps}
    \frac{p(H)\ket{\Psi_0}}{\|p(H)\ket{\Psi_0}\|}=\prod_{k=0}^{K-1} e^{i \theta_k \Psi_{k}}  e^{s_{k}[\Psi_{k},H]}\ket{\Psi_0},
\end{equation}
with $s_k = \frac{-1}{\sqrt{V_k}}\arccos\left(\frac{|E_{k}-z_{k}|}{\sqrt{V_{k}+|E_{k}-z_{k}|^2}}\right)$
and $\theta_k = \arg\left(\frac{E_k-z_k}{|E_k-z_k|}\right).$
Here, we recursively define the state $\ket{\Psi_{k}}$ by
\begin{align}
 \label{Thm 2 psik}
 \ket{\Psi_{k+1}}=e^{i \theta_k \Psi_{k}}  e^{s_{k}[\Psi_{k},H]}\ket{\Psi_{k}}\ .
\end{align}
\end{theorem}

Thm.~\ref{thm complex QSP} establishes a recursive method for constructing any QSP polynomial through a sequence of unitary operators.
Next, we explicitly demonstrate how this formulation can be implemented as a quantum algorithm.

\subsection{Implementation: Double-Bracket QSP algorithm (DB-QSP)} \label{sec:DB-QSP_implementation}

Building upon Thm.~\ref{thm complex QSP}, we present a unitary synthesis approach termed the Double-Bracket QSP algorithm~(DB-QSP).
A key challenge in implementing Eq.~\eqref{eq:qsp_wo_ps} lies in realizing the unitary operator $e^{s_{\Psi}[\Psi,H]}$.
Here, we adopt the approach of DBQAs and utilize the group commutator formula~\cite{gluza2024doublebracket,gluza_DB_QITE_2024,robbiati2024double,xiaoyue2024strategies} given by~\cite{dawson2006solovay, gluza2024doublebracket, HOPFGC3}:
\begin{align}    e^{s_\Psi[\Psi,H]} =& \left(
e^{is_\Psi^{\mtiny{(N)}} \Psi}e^{is_\Psi^{\mtiny{(N)}} H}
e^{-is_\Psi^{\mtiny{(N)}} \Psi}e^{-is_\Psi^{\mtiny{(N)}} H}
\right)^N \nonumber \\
&  +O(s_\Psi^{3/2}/\sqrt N)\ , 
\label{infinite GC}
\end{align}
where $s_\Psi^{\mtiny{(N)}} = \sqrt{|s_\Psi|/ N}$ for $s_\Psi \le 0$.
Note that, since the range of the arccos function is $[0,\pi]$, the time duration $s_{k}$ in Thm.~\ref{thm complex QSP} always takes a non-positive value.
Based on this approximation, DB-QSP implements QSP using the Hamiltonian evolution $e^{is_\Psi^{\mtiny{(N)}} H}$ and the state-dependent reflection gates $e^{is_\Psi^{\mtiny{(N)}} \Psi}$.
Specifically, the state-dependent reflection gate is implemented using the reflection about the initial state $\ket{\Psi_{0}}$ and a unitary operator $U$ satisfying $\ket{\Psi}=U\ket{\Psi_{0}}$, i.e., 
\begin{equation} 
    e^{is_\Psi^{\mtiny{(N)}}\Psi } = Ue^{is_\Psi^{\mtiny{(N)}}\Psi_{0}}U^{\dagger}.
\end{equation}
Alg.~\ref{alg DBQSP} summarizes the procedure of DB-QSP algorithm.
\begin{algorithm}
\caption{DB-QSP}
\begin{algorithmic}[1]
\STATE \textbf{Input:} Hermitian operator $H$, initial state $\ket{\Psi_0}$, degree $K$, parameters $\{z_k\}_{k=0}^{K-1}$, number of group commutator repetitions $N$.
\STATE \textbf{Output:} State $\ket{\Psi_K} =\frac{p(H)\ket{\Psi_0}}{\|p(H)\ket{\Psi_0}\|}$.
\STATE Initialize: $\ket{\Psi} \gets \ket{\Psi_0}$.
\FOR{$k = 0$ to $K-1$}
    \STATE Compute energy moment $E_k$ and variance $V_k$ for $\ket{\Psi}$. 
    \STATE Use Thm.~\ref{thm complex QSP} to determine parameters $s_k$ and $\theta_k$ for 
    $$
      e^{i\theta_k \Psi}e^{s_k [\Psi, H]} \ket{\Psi} =\frac{(H - z_k I)\ket{\Psi}}{\|(H-z_kI)\ket{\Psi}\|} .
    $$
    \STATE Set $s_k^{\mtiny{(N)}}=\sqrt{|s_k|/N}$ and the group commutator unitary $$G_k = e^{is_k^{\mtiny{(N)}} \Psi}e^{is_k^{\mtiny{(N)}} H}e^{-is_k^{\mtiny{(N)}} \Psi}e^{-is_k^{\mtiny{(N)}} H} .$$ 
    \State Update state by applying $\ket{\Psi} \gets e^{i\theta_k \Psi} G_k^N\ket{\Psi}\ .$
\ENDFOR
\STATE \textbf{Return:} $\ket{\Psi}$.
\end{algorithmic}
\label{alg DBQSP}
\end{algorithm}

We note that DB-QSP assumes that both the state-dependent reflection gate with respect to the initial state and Hamiltonian evolution can be generated efficiently.
Nevertheless, this assumption is not particularly restrictive.
For the reflection gates, a straightforward approach is to perform the density matrix exponentiation of the initial state~\cite{lloyd2014quantum, kjaergaard2022demonstration}. Yet,
if the input state $\ket{\Psi_{0}}$ is a computational basis state, the operation reduces to a multi-qubit controlled unitary, which can be implemented efficiently with cost scaling linearly in the number of qubits~\cite{barenco1995elementary, zindorf2024efficient,zindorf2025multi}.
More concretely, when $\ket
{\Psi_{0}}=\ket{0}$, the reflection gate takes the form of 
\begin{equation}
    e^{i\theta\ket{0}\bra{0}} = I + (e^{i\theta}-1)\ket{0}\bra{0}=
\begin{bmatrix}
e^{i\theta} & 0 &  0 & 0 \\
0 & 1 & 0 &  0 \\
0 & 0 & \ddots &  0 \\
0  & 0 & 0 & 1 \\
\end{bmatrix},
\end{equation}
which corresponds to a multi-qubit controlled parameterized phase gate.
Even when $\ket{\Psi_{0}}$ is not a computational basis state, if a unitary $U$ exists such that $\ket{\Psi_{0}}$ can be efficiently prepared from a basis state, e.g., $\ket{0}$, the reflection gate can be realized as $e^{i\theta\Psi_{0} } = Ue^{i\theta\ket{0}\bra{0}}U^{\dagger}$.
Similarly, when $H$ is a local Hamiltonian, efficient compilation is feasible using established Hamiltonian simulation methods~\cite{childs2010limitations, kothari, Childs2012LCU, PhysRevX.TrotterSuzukiError, SuChilds_nearly_optimal}.
Thus, in many practical scenarios where such compilation subroutines are available, Eq.~\eqref{DB bridge} serves as a unitary synthesis method of QSP without post-selection.

A key question is how efficiently DB-QSP can implement polynomials with small error.
Eq.~\eqref{infinite GC} indicates that the approximation error is governed by $|s_\Psi|^{3/2}/\sqrt N$, suggesting that the total number of group commutator repetiotions $N$ may need to increase for higher precision.
Thus, elucidating how $N$ (or equivalently, the circuit depth) scales to achieve a fixed precision is crucial for evaluating the practicality of DB-QSP.
In the following, we analytically estimate the circuit depth needed to accurately realize a polynomial $p(H)$ of degree $K$  using DB-QSP.
Before diving into this, we begin by analyzing the potential cost for implementing one step of DB-QSP with respect to the total discretization steps $N$.

\medskip    

\paragraph*{Implementation cost for a single step of DB-QSP.}

We begin by analyzing the total number of group commutator repetitions $N$ necessary to approximate $e^{s_{\Psi}[\Psi,H]}$ to $\epsilon_{0}$-precision via the group commutator formula. That is, we compute the required $N$ such that
{\small
\begin{align}
    \|e^{s_{\Psi}[\Psi,H]} - (
e^{is_\Psi^{\mtiny{(N)}} \Psi}
e^{is_\Psi^{\mtiny{(N)}} H}
e^{-is_\Psi^{\mtiny{(N)}} \Psi}
e^{-is_\Psi^{\mtiny{(N)}} H}
)^N\| \le \epsilon_{0}.
\label{eq:D_value}
\end{align}}
From Eq.~\eqref{infinite GC}, we can immediately see that the relative size of $s_{\Psi}$ and $N$ determines the error $\epsilon_{0}$.
We further recall that from Thm.~\ref{thm complex QSP} we have
\begin{align}\label{eq:skbound}
    |s_\Psi| &= \frac 1{\sqrt{V_\Psi}}\arccos\left(\frac{|E_{\Psi}-z_{}|}{\sqrt{V_{\Psi}+|E_{\Psi}-z_{}|^2}}\right) \nonumber\\
    &\le \frac{1}{|E_{\Psi}-z_{}|},
\end{align}
where the inequality is obtained by exploiting the fact that $s_{\Psi}$ is monotonically decreasing in $V_{\Psi}$ (as shown explicitly in App.~\ref{app lemma and thm proofs}). 
Combining Eq.~\eqref{infinite GC} and Eq.~\eqref{eq:skbound}, we want $1/(|E_{\Psi}-z_{}|)^{3/2}\sqrt{N} \le \epsilon_{0}$, and so we find that there exists an $N$ such that 
\begin{equation}
    N \in \mathcal{O}\left(\frac{1}{|E_{\Psi}-z|^{3} \epsilon_{0}^2}\right) 
\end{equation}
suffices to ensure Eq.~\eqref{eq:D_value} holds.

We thus see that a large gap $|E_{\Psi}-z_{}|$ reduces the required number of steps $N$. Conversely, $N$ diverges when $V_{k}=0$ and $z_{k}=E_{k}$.
This can intuitively be understood as arising from the fact that the operation $H-E_k I$ acts as an ``annihilation operator". If $V_{k}=0$, then the state is an eigenstate and $E_{k}$ corresponds to its eigenvalue, meaning $(H-E_k I)\ket{\Psi} =0$ and so the method breaks down. We note that a similar breakdown for eigenstates was observed for a quantum algorithm for ITE using the group commutator unitary in Ref.~\cite{gluza_DB_QITE_2024}.

\medskip

\paragraph*{Circuit Depth of DB-QSP.}

We now proceed to analyze the circuit depth to realize a DB-QSP state that is $\epsilon$-close to the ideal state for a degree-$K$ polynomial.
We define the circuit depth as the number of Hamiltonian evolution gates and reflection gates to construct quantum circuits for DB-QSP.
For the analysis, consider the following state constructed by DB-QSP:
{\footnotesize
\begin{align}
\label{DBQSP DME}
    &\ket{\omega_{K}} \nonumber\\
    &=\prod_{k=0}^{K-1}   e^{ i \theta_k \omega_{k}}  \left(
e^{is_k^{\mtiny{(N)}} \omega_{k}}
e^{is_k^{\mtiny{(N)}} H}
e^{-is_k^{\mtiny{(N)}} \omega_{k}}
e^{-is_k^{\mtiny{(N)}} H}
\right)^N\ket{\omega_0}\ 
\end{align}}
where the intermediate state is recursively constructed as 
\begin{align}
 &\ket{\omega_{k+1}} \nonumber\\
 &=e^{ i \theta_k \omega_{k}}  \left(
e^{is_k^{\mtiny{(N)}} \omega_{k}}
e^{is_k^{\mtiny{(N)}} H}
e^{-is_k^{\mtiny{(N)}} \omega_{k}}
e^{-is_k^{\mtiny{(N)}} H}
\right)^N\ket{\omega_{k}}\ 
\end{align}
with $s_k^{\mtiny{(N)}} = \sqrt{|s_{k} |/ N}$. We also define the exact QSP state derived from Thm.~\ref{thm complex QSP} as 
\begin{align} \label{eq:ideal_state}
    \ket{\Psi (\bm{\theta},\bm{s})} = \prod_{k=0}^{K-1} e^{i \theta_k \Psi_{k}}  e^{s_{k}[\Psi_{k},H]}\ket{\Psi_0}\ 
\end{align}
with $\bm{\theta}=(\theta_0,\ldots,\theta_{K-1})$ and $\bm{s}=(|s_0|,\ldots,|s_{K-1}|)$. The following Theorem captures the circuit depths required to ensure that the DB-QSP state in Eq.~\eqref{DBQSP DME}, agrees with the true circuit up to $\epsilon$ precision. The key assumption here is that the parameters $(\theta_{k},s_{k})$ are known exactly. In practice, the parameters will be computed with a finite number of measurement shots, requiring an additional sampling overhead and introducing additional errors. We will address this aspect in Section~\ref{sec:error_analysis}.

\begin{theorem}[DB-QSP circuit depth] \label{thm:circuit_depth_DB-QSP}
Suppose $H$ is a Hermitian matrix whose spectral radius does not exceed unity, i.e., $\|H\|\le1$. 
Let $\zeta = \max(\bm{\theta},\bm{s})$ be the maximum value of all elements in $\bm{\theta}$ and $\bm{s}$.
Also, consider $\ket{\omega_{K}}$ given by DB-QSP from Alg.~\ref{alg DBQSP} in Eq.~\eqref{DBQSP DME} and the state $\ket{\Psi (\bm{\theta},\bm{s})}$ from Thm.~\ref{thm complex QSP} in Eq.~\eqref{eq:ideal_state} for degree-$K$ polynomials.
Then there exists a circuit depth $\mathcal{N}_{K}$ such that 
\begin{equation} \label{eq:bound_circuit_depth}
    \mathcal{N}_{K} \in \mathcal{O} \left(\left(\frac{8}{3
    }\right)^2\zeta(1+6\zeta)^{2K} /  \epsilon^2 + 3\right)^{K}
\end{equation}
suffices to ensure that $\|\ket{\Psi(\bm{\theta},\bm{s})} - \ket{\omega_{K}}\| \le \epsilon$.
\end{theorem}

To prove this, we first utilize a result proven in App.~\ref{app notions of stability} that the DB-QSP error can be bounded as 
\begin{align} \label{eq:error_bound_DB-QSP}
    \|\ket{\Psi(\bm{\theta},\bm{s})} - \ket{\omega_{K}}\| 
    \le \frac{4}{3}\zeta^{1/2}(1+6\zeta)^K/\sqrt{N}
   \ .
\end{align}
Next, we compute the circuit depth, which is defined as the total number of Hamiltonian evolution gates and the reflection gates.
Given that the state is $\ket{\omega_k} = U_k \ket{0}$, we can write the recursive unitary synthesis formula as
\begin{align}
 \label{unfolded Uk}
    U_{k+1} =& U_k e^{i\theta_{k} \ket0\bra 0}U_k^\dagger \times  G^N \times U_k\ 
\end{align}
with
\begin{align*}
G &=
U_k e^{is_k^{\mtiny{(N)}} \ket0\bra 0}
U_k^{\dagger} e^{is_k^{\mtiny{(N)}} H} \nonumber\\
&\qquad \times U_ke^{-is_k^{\mtiny{(N)}} \ket0\bra 0}
U_k^{\dagger} e^{-is_k^{{\mtiny{(N)}}}H }.
\end{align*}
This implies that each step involves $4N+3$ repetitions of the unitary operators $U_k$ at the previous step. 
Therefore, since an additional $4N+1$ gates ($2N$ gates for Hamiltonian evolution and $2N+1$ for the reflection gates on the initial state $\ket{0}\bra{0}$) are required, the circuit depth $\mathcal{N}_{k+1}$ at step $k+1$ is given by $\mathcal N_{k+1} = (4N+3) \mathcal N_k + 4N+1$.
Thus, the total circuit depth required for a polynomial of degree $K$ can be represented as
\begin{equation} \label{eq:num_queries}
    \mathcal{N}_{K}= \frac{(4N+1)((4N+3)^{K}-1)}{4N+2} \le (4N+3)^{K}
\end{equation}
Thus, by substituting Eq.~\eqref{eq:num_queries} into the right-hand side of Eq.~\eqref{eq:error_bound_DB-QSP}, Eq.~\eqref{eq:bound_circuit_depth} satisfies to ensure the $\epsilon$-precision as claimed in the theorem.

Thm.~\ref{thm:circuit_depth_DB-QSP} indicates that the circuit depth scaling can be prohibitive for high degree polynomials. 
Namely, although the depth scales polynomially in the precision $1/\epsilon$, it grows super-exponentially with the degree of the polynomials $K$.
Consequently, DB-QSP is not practically applicable to polynomials of arbitrary degrees, but should target low-degree polynomials.

 \subsection{Performance Analysis of Perturbations in Parameters} \label{sec:error_analysis}

In this section, we analyze the effect of statistical noise. 
As shown in Alg.~\ref{alg DBQSP}, DB-QSP requires the estimation of energy and variance to determine the parameters $s_{k}$ and $\theta_{k}$ at each time step.
However, due to the finite number of measurement shots in practice, precise estimation is not feasible on quantum hardware. 
Consequently, parameters deviate from their true values at each time step, with perturbations satisfying $|s_{k}-\tilde{s}_{k}|\le \delta_{s}$ and $|\theta_{k}- \tilde{\theta}_{k}|\le \delta_{\theta}$.
In other words, even if the quantum hardware performs the operations perfectly, statistical errors from the measurements lead to erroneous parameters. 

Under this setting, we provide an error bound for implementing a polynomial of degree $K$.
We introduce a noisy state to handle the erroneous parameters:
\begin{align} \label{eq:erroneous_state}
    \ket{\tilde{\Psi}_{H}(\tilde{\bm{\theta}},\tilde{\bm{s}})} = \prod_{k=1}^{K} e^{ i\tilde{\theta}_k \tilde{\Psi}_{k} }  e^{\tilde{s}_{k}[\tilde{\Psi}_{k},H]}\ket{\Psi_0},
\end{align}
where we define $\ket{\tilde{\Psi}_{k+1}} = e^{ i\tilde{\theta}_k \tilde{\Psi}_{k} }  e^{\tilde{s}_{k}[\tilde{\Psi}_{k},H]}\ket{\tilde{\Psi}_{k}}$, with $\ket{\tilde{\Psi}_0}=\ket{\Psi_0}$.
Here, we also introduce $\tilde{\bm{\theta}}=(\tilde{\theta}_0,\ldots,\tilde{\theta}_{K-1})$ and $\tilde{\bm{s}}=(|\tilde{s}_0|,\ldots,|\tilde{s}_{K-1}|)$.
Again, we assume $\|H\|\le1$. 
We can then derive the following result on the circuit error with the detailed proof provided in App.~\ref{app notions of stability}.

\begin{proposition}[Stability of Thm.~\ref{thm complex QSP} under erroneous estimation] \label{prop: para_stability}
Let $H$ be a Hermitian matrix such that $\|H\|\le1$, and assume that the estimated parameters $\tilde s_k$ and $\tilde\theta_k$ satisfy $|s_k-\tilde s_k|\le \delta_s$ and $|\theta_k -\tilde \theta_k|\le \delta_\theta$ with ideal parameters $s_{k}$ and $\theta_{k}$ for all $k$.
By setting $\zeta = \max(\bm{\theta},\bm{s})$, the perturbed state $\ket{\tilde\Psi_H(\tilde{\bm{\theta}},\tilde{\bm{s}})}$ in Eq.~\eqref{eq:erroneous_state} and the state $\ket{\Psi_H (\bm{\theta},\bm{s})}$ from Thm.~\ref{thm complex QSP} satisfies

\begin{align} \label{eq:error_accumulation_by_noisy_estimation}
    &\|\ket{\Psi_{H}(\bm{\theta},\bm{s})} - \ket{\tilde\Psi_{H}(\tilde{\bm{\theta}},\tilde{\bm{s}})}\| \nonumber\\& \quad \le \frac{1}{3\zeta}(1+6 \zeta)^K \max(\delta_s,\delta_\theta)\ .
\end{align}
\end{proposition}

Prop.~\ref{prop: para_stability} indicates that the parameter deviation $\max(\delta_s,\delta_\theta)$ scales linearly with the accumulated error and hence its suppression is critical.
On the other hand, since these parameters are nonlinear functions of energy and variance, analyzing the impact of statistical estimates from Prop.~\ref{prop: para_stability} is non-trivial.
To address this, we extend our analysis to explicitly account for statistical noise in energy and variance estimation.

The recursive structure of Eq.~\eqref{eq:ideal_state} and Eq.~\eqref{eq:erroneous_state} implies that, even in the limit of infinite measurement shots, the estimated energy and variance may still differ from their ideal values. 
This discrepancy arises because these quantities are measured on a potentially different state at each iteration. 
Specifically, the statistical estimate $\overline{E}_{k}$ ($\overline{V}_{k}$) converges to $\tilde{E}_{k}$ ($\tilde{V}_{k}$) obtained from the noisy state $\ket{\tilde{\Psi}_{k}}$, rather than the ideal values $E_{k}$ and $V_{k}$.
To account for this, we extend Prop.~\ref{prop: para_stability} and demonstrate that the statistical noise, $\delta_{E}=|\overline{E}_{k}-\tilde{E}_{k}|$ and $\delta_{V}=|\overline{V}_{k}-\tilde{V}_{k}|$, exhibits a linear dependence on the accumulated error in Eq.~\eqref{eq:error_accumulation_by_noisy_estimation}, but keeps the exponential scaling with $K$. That is, $\|\ket{\Psi_{H}(\bm{\theta},\bm{s})} - \ket{\tilde\Psi_{H}(\tilde{\bm{\theta}},\tilde{\bm{s}})}\| \le C^{K} \max(\delta_{E},\delta_{V})$ for a constant $C\ge 1$. 
See App.~\ref{app notions of stability} for further technical details.

Using the results, we also estimate the number of measurement shots needed to achieve the state $\epsilon$-close to the ideal state at $K$ step.
Without loss of generality, we express the Hermitian matrix as a weighted sum of Pauli terms, i.e., $H = \sum_{i=1}^J w_i P_i$.
Then, the number of samples $N_{E}$ ($N_{V}$) required to estimate the energy (variance) within an error $\tilde{\epsilon}$ with probability at least $1-\delta$, $\delta\in(0,1]$, scales as $N_{E}\in\mathcal{O}(J\|w\|_1^2/\tilde{\epsilon}\delta)$ ($N_{V}\in\mathcal{O}(J^2\|w\|_1^4/\tilde{\epsilon}\delta)$), where $\|w\|_1=\sum_{i=1}^{J}|w_{i}|$.
See App.~\ref{app variance estimator} for further details.
Thus, the number of measurement shots needed to achieve the $\epsilon$-precision grows exponentially with the polynomial degree $K$, since the estimation error must be sufficiently small to cancel a term that scales exponentially with $K$.
However, when $K$ scales logarithmically with $1/\epsilon$,  the required measurement shots reduce to polynomial resources in $1/\epsilon$. 

Lastly, in App.~\ref{sec: classical_computation_app} and App.~\ref{app variance estimator}, we explore two different aspects to alleviate statistical errors: (1) the potential of classical computations to reduce the impact of imprecise estimation, and (2) an analysis of the estimator for variance operators.

\subsection{Application Examples:  Ground-State Approximation and Matrix Inversions} \label{sec:application}

A significant limitation of DB-QSP is its scaling with polynomial order $K$ in Thm.~\ref{thm:circuit_depth_DB-QSP}. 
Yet, Thm.~\ref{thm:circuit_depth_DB-QSP} implies that when the degree of polynomials $K$ scales logarithmically in the inverse of the precision, $1/\epsilon$, then the circuit depths follow a quasi-polynomial scaling, i.e., $\mathcal{N}_{K}=2^{\text{poly}\log(1/\epsilon)}$.
Thus, DB-QSP is potentially applicable to polynomial functions of degree at most $K=\mathcal{O}(\log(1/\epsilon))$.
Indeed, low-degree approximations, such as those using Chebyshev polynomials, allow efficient representations of certain functions using logarithmically-small degrees.
For specific examples of such approximations, see App.~\ref{sec:application_app}.
In the following, we present two representative tasks that illustrate the utility and limitations of low-degree approximations: ground-state preparation and matrix inversion. We discuss DB-QSP's applicability to other tasks in App.~\ref{sec:application_app}.

\medskip

\paragraph*{Ground-State Approximation.}

Given a Hermitian matrix $H$, the objective here is to prepare its ground state $\ket{\lambda_{0}}$.
To achieve this, some quantum algorithms employ a variety of filtering techniques that apply suitable functions of the Hamiltonian.
In what follows, we focus on two representative examples: Imaginary-Time Evolution (ITE), which employs the exponential filter, and the approach proposed in Ref.~\cite{lin2020near}, which makes use of a sign-function filter.

As a first example of a method that can be implemented using DB-QSP, we consider Imaginary-Time Evolution (ITE), where the non-unitary operator $p(H) = e^{-\tau H} $ is applied to an initial state $\ket{\Psi_{0}}$:
\begin{align} \label{eq:ite}
    \ket{\Psi_\tau} =  \frac{e^{-\tau H}\ket{\Psi_0}}{\|e^{-\tau H}\ket{\Psi_0}\|} \ .
\end{align}
A key feature of ITE is that it guarantees convergence as long as the initial state has a nonzero overlap with the ground state.
Refs.~\cite{gluza_DB_QITE_2024,mcmahon2025equating} established that Eq.~\eqref{DB bridge} serves as a first-order approximation of ITE and further extends this result by employing group commutator iterations in Eq.~\eqref{infinite GC} to construct a unitary realization of ITE: See App.~\ref{sec:application_app} for more details.
In addition to the methods proposed in Refs.~\cite{gluza_DB_QITE_2024,mcmahon2025equating}, DB-QSP can be used to directly construct a polynomial approximation of the exponential function. Specifically, such an approximation requires a polynomial of degree $K=\mathcal{O}(\sqrt{2\tilde{\tau}\log(4/\epsilon)})$ with $\tilde{\tau}=\lceil\max(e^2 \tau, \log(2/\epsilon))\rceil$~\cite{sachdeva2014faster}. 
This implies that a polynomial approximation of ITE can be implemented with a complexity that scales quasi-polynomially in the error precision~$\epsilon$. 
However, the scaling with respect to the evolution time $\tau$ remains unfavorable.

Beyond ITE, DB-QSP can also construct alternative filtering functions. 
Ref.~\cite{lin2020near} presents a nearly optimal algorithm for ground-state preparation using QSVT. The core idea is to use a low-degree approximation of the sign function, whose degree scales logarithmically with $1/\epsilon$, i.e., $K=\mathcal{O}(\log(1/\epsilon)/\delta)$, for an input $x\in[-2,2]\backslash(-\delta,\delta)$ with $\delta>0$.
Hence, with the same approximation technique, similar filtering strategies could potentially be realized via DB-QSP with favorable scaling in the error precision $\epsilon$.

We recall that the success probability of existing QSP implementations depends on the overlap between the initial state and ground state, meaning that these methods may fail entirely if the initial state is not well-prepared~\cite{lin2020near}.
In contrast, our approach is applicable to any state, as long as there is a non-zero overlap. Therefore, even if DB-QSP cannot fully implement the desired polynomial functions due to resource constraints, it can still systematically improve the quality of the state.

\medskip

\paragraph*{Matrix Inversion.}
The goal of ``matrix inversion" is to apply $A^{-1}$ to an input state, where $A$ is a square matrix. This is a core subroutine for solving linear systems $A\ket{x}=\ket{b}$ for $\ket{x}$~\cite{harrow2009quantum, martyn2021grand,gilyen2019quantum}. As shown in App.~\ref{sec:application_app}, a polynomial of degree $K=\mathcal{O}(\kappa\log(\kappa/\epsilon))$ can approximate the inverse function $1/x$ with the precision $\epsilon$ for an input $x\in[-1,1]\backslash (-\frac{1}{\kappa},\frac{1}{\kappa})$ where $\kappa\ge1$ is the condition number of the matrix.
While our results so far apply to Hermitian matrices, we can construct a Hermitian matrix from any square matrix $A$ by the extension:
\begin{equation}
    H=\begin{bmatrix}
        0 & A \\
        A^{\dagger} & 0 \\
\end{bmatrix} , 
\end{equation}
This indicates that DB-QSP has the potential to efficiently perform matrix inversion in terms of the inverse precision $1/\epsilon$.
However, the circuit depth required for DB-QSP scales super exponentially with the condition number, a key factor in assessing the algorithm’s efficiency. 
Thus, this example also highlights a fundamental challenge for DB-QSP in certain computational tasks.

\subsection{Hybrid Strategy: DB-QSP with Existing Methods} \label{sec:hybrid strategy}

The performance analysis has highlighted that, while DB-QSP holds promise for certain tasks, its circuit depth and the requirement for precise estimation of the energy and variance pose significant challenges.
However, DB-QSP does not have to be used as a standalone approach.
By integrating it with existing methods, these limitation can be alleviated and a hybrid approach may further enhance its feasibility.

In the following, we explore how combining DB-QSP with established techniques can improve performance.
Specifically, we examine three approaches: (1) Variational Quantum Algorithms (VQA)~\cite{cerezo2021variational} and classical computation, (2) Quantum Dynamic Programming (QDP)\cite{QDP}, (3) qubitization and LCU.

\medskip 

\paragraph*{VQA \& Classical pre-computations.}
A potential strategy to circumvent the challenges in DB-QSP is to employ a preconditioner that bypasses the initial steps. In this regard, classical computational methods can serve as effective preconditioners. 
Our target operation in Eq.~\eqref{eq: aI+bH main_} consists of a weighted sum of $I$ and $H$ with appropriate coefficients.
Consequently, the feasibility of classical computation relies on the efficiency of evaluating $\braket{\Psi_{0}|H^{2k+2}|\Psi_{0}}$ for degree-$k$ polynomials.
We show that, if the initial state $\ket{\Psi_{0}}$ is sparse and a Hermitian matrix contains a limited number of Pauli terms, then classical computation is feasible.
Moreover, advanced classical techniques, e.g., tensor networks, could further improve the efficiency, see Sec.~\ref{sec: classical_computation_app} for details.

Another approach might be to first use a Variational Quantum Algorithm (VQA) where a parameterized quantum circuit $U_{\theta}$ is trained to approximate the target state.
By leveraging a VQA, a relatively shallow-depth circuit might be found to replicate the operations of a few DB-QSP steps, allowing the trained circuit to serve as a warm start for DB-QSP; i.e. $\ket{\Psi_k} \approx U_\mathbf{\theta} \ket{\Psi_0}$
for a small $k$.
However, this strategy has several challenges.
First, there are no theoretical guarantees of convergence for VQAs in practical regimes.
Secondly, as highlighted in Thm.~\ref{thm:circuit_depth_DB-QSP} and Prop.~\ref{prop: para_stability}, small errors at each step can accumulate significantly as the polynomial degree increases. 
Consequently, errors introduced by the VQA may degrade the final result.
Another issue is the barren plateau phenomenon~\cite{mcclean2018barren, cerezo2020cost}, where gradient magnitudes vanish exponentially with system size, making parameter training impractical.
Indeed, it has been suggested that VQAs themselves may need warm starting strategies~\cite{puig2024variational, mhiri2025unifying}, or else they are classically simulable~\cite{cerezo2023does, angrisani2024classically, bermejo2024quantum, lerch2024efficient}.
In such cases, the direct use of DB-QSP may be a better option.

\medskip \paragraph*{Quantum Dynamic Programming (QDP).}
We recall that the significant increase in circuit depth arises from a recursive circuit structure. Specifically, the implementation of the state-dependent reflection through $e^{is_{k}\Psi_{k}}=U_{k}e^{is_{k}\ket{0}\bra{0}}U_{k}^{\dagger}$ leads to a prohibitive number of queries to the quantum gates.
Therefore, incorporating a subroutine that reduces the implementation cost would enhance the efficiency of DB-QSP. The operation $e^{is\Psi_{k}}$ is a special case of Density-Matrix Exponentiation (DME), for which some quantum algorithms have been proposed~\cite{lloyd2014quantum, Kimmel2017DME_OP, kjaergaard2022demonstration}.
DME leverages coherent swap operations between multiple copies of $\ket{\Psi}$ to realize exponentiation. 
Note that, since the swap operations are independent of previous runtime, the circuit depth scales only polynomially.

Recently, Quantum Dynamic Programming (QDP) has been proposed to study the use of routines such as DME for speeding up quantum recursions. 
QDP is  powerful in that  utilizing memory leads to an exponential reduction in circuit depth~\cite{QDP}. 
This characteristic makes it a viable subroutine for DB-QSP.
However, due to the no-cloning theorem, QDP has the disadvantage that one must extend the width when implementing recursion steps, meaning that multiple copies of the state must be prepared~\cite{QDP}.
Hence, when combining DB-QSP with QDP, it becomes crucial to balance the trade-off between width and depth for practical feasibility.

\medskip \paragraph*{Qubitization \& LCU.}
One may envision integrating DB-QSP with QSP implementations that involve post-selection (e.g. qubitization and LCU). While these methods supplemented with amplitude amplification are highly sophisticated and can function as standalone methods, their practicality can be hindered in certain scenarios.
As discussed in Sec.~\ref{sec:preliminary}, an exponentially small success probability can limit the practicality of these methods for some tasks such as ground-state preparation.
Thus, by leveraging DB-QSP as a preconditioner, we can potentially mitigate this issue and enhance the overall feasibility of these advanced algorithms.\\

\section{Discussion}\label{sec: discussion}

Quantum signal processing (QSP) is a fundamental framework for designing  quantum algorithms. 
Existing implementation methods, such as qubitization and linear combinations of unitaries (LCU), are powerful but rely on post-selection of auxiliary qubits, which could limit their celebrated efficiency in certain cases.
In this work, we propose a unitary synthesis formula for QSP without auxiliary qubits and post-selection. 
Our method, termed DB-QSP, relies on Thm.~\ref{thm complex QSP} together with the recently established Double-Bracket Quantum Algorithm (DBQA) framework~\cite{gluza2024doublebracket,gluza_DB_QITE_2024}.
While our approach comes at the cost of circuit depth and requires precise estimation of energy and variance, it can be used to efficiently implement low-degree polynomial approximations.
We further note that our method requires fewer multi-qubit controlled gates than qubitization, as no ancillary qubits are used.
Namely, DB-QSP may be more advantageous under hardware constraints, which are likely to occur in near-term or upcoming quantum devices.
Thus our proposal broadens the range of options for implementing the QSP framework on quantum hardware, with hybrid approaches that combines both DB-QSP and prior methods looking particularly appealing.

Further investigations of the fundamental limits of our algorithm's performance could be interesting. 
Thm.~\ref{thm complex QSP} clarifies that the time duration $s$ and the angle $\theta$ are determined by the energy mean and fluctuation (i.e., variance). 
This suggests that thermodynamic quantities alone provide sufficient information to guide the implementation of the target transformation. 
Moreover, these quantities are key to the algorithm’s efficiency, as the complexity is determined by the time duration and angles defined using them.
Given that this unitary process originates from an unphysical polynomial function, the link between unphysical operations and underlying physical principles will provide an intriguing perspective.

Let us highlight a geometrical view of our algorithm.
As shown in Eq.~\eqref{eq:qsp_conv_uni_qubitize}, the core of qubitization is that an auxiliary two-level system enables the construction of specific polynomials through a sequence of unitary operators interleaved with phase gates. 
Interestingly, Thm.~ \ref{thm complex QSP} reveals that QSP implementations without auxiliary qubits exhibit a similar structure (i.e., Eq.~\eqref{eq:qsp_wo_ps}). 
This structural similarity suggests that both approaches can be analyzed from a geometrical perspective.
More specifically, through the lens of DBQA~\cite{gluza2024doublebracket, gluza_DB_QITE_2024}, it is known that the exponential of commutators, $e^{[\Psi,H]}$, approximates the steepest descent direction on the Riemannian manifold of quantum states with respect to the cost function $-\|H-\Psi\|^2_2/2$~\cite{helmke_moore_optimization, moore1994numerical, bloch1990steepest, smith1993geometric, dirr2008lie, kurniawan2012controllability, schulte2011optimal, schulte2008gradient, riemannianflowPhysRevA.107.062421}. 
While our formulation introduces additional state-dependent reflection gates, polynomials with real roots, such as Chebyshev polynomials, can be constructed without these additional gates. 
This observation suggests a promising direction for exploring a geometric interpretation of QSP, and possibly of qubitization itself, within our framework.

Finally, a natural direction for future work is to compare the runtime of the proposed approach with that of alternative methods.
In particular, recent studies have introduced Lindbladian approaches requiring only a single ancilla qubit~\cite{ding2024single}.
Such dissipative schemes potentially offer an advantages over DBQA and standard QSP, as the convergence to the ground state is possible even when the initial state has zero overlap with the ground state~\cite{ding2024single,chen2025efficient,chen2024local}.
Interestingly, it is known that Lindbladian dynamics can be formulated as a gradient flow with respect to the quantum relative entropy~\cite{mittnenzweig2017entropic}.  
This suggests that, despite being seemingly distinct approaches, DBQAs, QSP and Lindbladian simulations may have a shared foundation in gradient flow dynamics defined for appropriate cost functions and underlying Riemannian manifolds.
A detailed analysis of the role of Riemannian gradients in the procedures of Refs.~\cite{ding2024single,chen2025efficient,chen2024local} is left for future investigation.

\medskip
\textit{Acknowledgments.} Insightful discussions with Thais Silva and Andrew Wright are acknowledged. MG, JS, BHT and NN are supported by the start-up grant of the Nanyang Assistant Professorship at the Nanyang Technological University in Singapore. ZH was supported by the Sandoz Family Foundation-Monique de Meuron program for Academic Promotion. MG was also supported by the Presidential Postdoctoral Fellowship of the Nanyang Technological University in Singapore.

\bibliography{QITE_DBF_EPFL.bib} 
\bibliographystyle{quantum}

\clearpage
\onecolumn\newpage
\appendix

\part{Appendix}
\parttoc

\newtheorem{theoremA}{Theorem}[section]
\newtheorem{propositionA}[theoremA]{Proposition}
\newtheorem{lemmaA}[theoremA]{Lemma}
\newtheorem{definitionA}[theoremA]{Definition}
\newtheorem{corollaryA}[theoremA]{Corollary}
\newtheorem{remarkA}[theoremA]{Remark}

\renewcommand{\thetheoremA}{\Alph{section}.\arabic{theoremA}}

\newpage
\section{Overview of Methods for QSP Involving Post-Selection} \label{app:overview_qsp}
We start with a brief overview of Quantum Signal Processing (QSP) through its unitary synthesis method known as qubitization.
Then, we also introduce the Linear Combination of Unitaries (LCU) as a unitary synthesis technique for implementing QSP.

\subsection{Overview of QSP Using Qubitization} \label{sec:review_QSP}
QSP is a framework for systematically constructing matrix-valued functions. 
The concept of QSP originated from a series of works which aimed at characterizing the achievable polynomial functions of a scalar value embedded in a single-qubit rotation~\cite{equiangular}.
The QSP methodology introduced in Ref.~\cite{equiangular} was later extended to Hermitian matrices through a technique known as \textit{qubitization}~\cite{Low2019hamiltonian}, using the framework of block-encodings. 
Subsequently, QSP was generalized to all polynomials~\cite{motlagh2024generalized} and extended to non-square matrices through Quantum Singular Value Transformation (QSVT)~\cite{gilyen2019quantum}.
Notably, this has led to asymptotically optimal Hamiltonian simulation algorithms~\cite{Low2019hamiltonian} and a near-optimal method for ground-state preparation~\cite{lin2020near}.
Furthermore, QSP serves as a fundamental tool for constructing primitive quantum algorithms that exhibit quantum advantages~\cite{martyn2021grand,gilyen2019quantum}.
Therefore, its efficiency in implementing linear algebraic operations and its role as a key building block for quantum algorithms with potential advantages have made QSP a subject of significant interest.

Here, we focus on the qubitization technique~\cite{equiangular}. 
Specifically, following the approach in Ref.~\cite{equiangular}, we begin with a degree-$K$ polynomial of a scalar input $x\in[-1,1]$.
In the original work, a quantum circuit $U_{YLC}$ was introduced with a sequential structure comprising of two types of operators:
\textit{signal operators} ${W}$ and \textit{signal processing operators} ${S}(\phi)$, where the phase $\phi$ is drawn from a set ${\phi_k}$.
The desired polynomial transformation is then obtained by performing a measurement in the so-called \textit{signal basis}.
Concretely, it was demonstrated that there exists a sequence of QSP phase $\{\phi_{k}\}$ such that the following circuit
\begin{equation} \label{app eq:qsp_conv_uni}
U_{YLC} = S_{z}(\phi_{0}) \prod_{k=1}^{K} W(x)S_{z}(\phi_{k})
\end{equation}
with the operators \vspace{-0.15cm}
\begin{equation*}
    W(x) = e^{ixX/2} = \begin{bmatrix}
        x & i\sqrt{1-x^2} \\
        i\sqrt{1-x^2} & x \\
\end{bmatrix},\qquad
    S_{z}(\phi) = e^{i\phi Z} = \begin{bmatrix}
        e^{i\phi} & 0 \\
        0 & e^{-i\phi} \\
\end{bmatrix},
\end{equation*}
followed by measurement in the basis $M=\{\ket{+},\ket{-}\}$ can realize a degree-$K$ real polynomial $p(x)$, provided that
\begin{enumerate}
\item Degree of $p(x)$ is equal to or less than $K$,
\item $p(x)$ has a parity $K$ mod 2,
\item $\forall x\in[-1,1]$, $|p(x)|\le1$.
\end{enumerate}
This setting is referred to as the Wx convention, as the signal operator is implemented using the rotation-$x$ gate. Alternatively, the rotation-$z$ gate can be used, which is known as the Wz convention. For further details, see Ref.~\cite{martyn2021grand}.

The core idea of the synthesis approach is that single-qubit rotations can implement arbitrary polynomial transformations, provided the conditions mentioned earlier are met.
Similarly, this single-qubit-like structure used in Eq.~\eqref{app eq:qsp_conv_uni} to synthesize polynomial functions can be extended to Hermitian matrices by employing the block-encoding technique~\cite{chakraborty_et_al:LIPIcs.ICALP.2019.33}, which embeds a Hermitian matrix $ H$ in the top-left block of a larger unitary matrix.
More precisely, $U_{H}$ is called a $(\alpha,n_a,\epsilon)$ block-encoding of $H$, if it satisfies 
\begin{align}  \label{eq:be_def}
    \|H - \alpha(\langle0|^{\otimes n_a} \otimes I)U_H (|0\rangle^{\otimes n_a} \otimes I)\|\le \epsilon
\end{align}
with $\alpha,\epsilon \in \mathbb{R}^{+}$ and the number of auxiliary qubits $n_{a}$.
An example of the matrix form is given by
\begin{equation} \label{eq:be_unitary}
W( H) =  \begin{bmatrix}
         H & i\sqrt{I- H^2} \\
        i\sqrt{I- H^2} &  H
        \end{bmatrix}.
\end{equation}
By substituting the signal operator $W(x)$ in the unitary $U_{YLC}$ of Eq.~\eqref{app eq:qsp_conv_uni} with the block-encoded unitary in Eq.~\eqref{eq:be_unitary}, we can perform polynomial transformations of Hermitian matrices.
Furthermore, QSP has been extended to non-square matrices via QSVT, which enables the manipulation of singular values for broader applications in quantum linear algebra.

We note that, given an input state $\ket{\Psi}$, the state after applying the block-encoding unitary $U_{H}$ in Eq.~\eqref{eq:be_def} is expressed as
\begin{equation}
    \ket{0}^{\otimes n_{a}} \otimes \frac{H}{\alpha}\ket{\Omega} + \ket{\text{garbage}^{\perp}}.
\end{equation}
Here, $\ket{\text{garbage}^{\perp}}$ is a state orthogonal to $\ket{0}^{\otimes n_{a}} \otimes H/\alpha\ket{\Omega}$.
Due to the normalization, the probability of getting $\ket{0}^{\otimes n_{a}}$ is given by $ p_\text{succ} = \|H\ket \Omega\|^2/\alpha^2$.
By extending Eq.~\eqref{app eq:qsp_conv_uni} to controlled-unitary operations, we obtain the state
\begin{align}
\ket{0}^{\otimes n_{a}} \otimes p(H/\alpha)\ket{\Omega} + \ket{\text{garbage}^{\perp}},
\end{align}
which succeeds with probability
\begin{align}
\label{succ prob be p(H)_}
p_\text{succ} = \|p(H/\alpha)\ket{\Omega}\|^2.
\end{align}
As shown in the main text, a core insight is that this probability can be exponentially small in case of Imaginary-Time Evolution (ITE) where $p(H) \approx e^{-\tau H}$. 
In such cases, Eq.~\eqref{succ prob be p(H)_} is inversely proportional to the fidelity of the initial state with a thermal state, which can decay exponentially~\cite{chan2023simulating,sze2025hamiltonian}.
More generally, this fidelity dependence holds across different scenarios. 
For instance, the block-encoding query complexity for nearly-optimal ground-state preparation algorithm in Ref.~\cite{lin2020near} scales as $O(\alpha/\gamma)$, where $\gamma = |\langle \lambda_0|\Psi\rangle|^2$ is the fidelity of the input state $\ket{\Psi}$ with the ground state $\ket{\lambda_0}$ of $H$.
The query scaling $O(\alpha/\gamma)$ corresponds to the inverse success probability and thus requires repeated trials for obtaining a successful outcome. 
Other probabilistic methods exhibit similar sensitivity~\cite{lin2020near, ge2019faster, poulin2009preparing}; see, for example, Ref.~\cite{oumarou2025molecular} for a discussion focused on computing expectation values rather than preparing quantum states and Ref.~\cite{zhang2025quantum} for modifications of the filter functions which aim to alleviate this problem.
This indicates that the success of the block-encoding depends on the input state. 
Additionally, since the number of queries to the block-encoding unitary scales with the degree of polynomials as shown in Eq.~\eqref{app eq:qsp_conv_uni}, the degree $K$ needs to be sufficiently low to ensure successful post-selection each time.

\subsection{Overview of QSP Using Linear Combination of Unitaries (LCU)}\label{sec:review_LCU}

Another straightforward approach to implementing QSP is the Linear Combination of Unitaries (LCU) technique~\cite{gui2006general,Childs2012LCU,Childs2012LCU,kothari,chakraborty2024implementing}. 
More broadly, LCU is a fundamental method for realizing general matrix functions using unitary operations. 
The key idea is that, a given matrix $H = \sum_{j=1}^J w_j U_j$, which can be expressed as a weighted sum of unitary operators $\{U_{j}\}$, can be efficiently implemented with additional auxiliary qubits whose number grows logarithmically with the number of decomposed terms $J$ in the matrix. 
The desired transformation is then realized by measuring the auxiliary qubits, which corresponds to successfully projecting the system onto a subspace where the target operation is encoded.
In this sense, LCU serves as one way to implement the block-encoding framework in Eq.~\eqref{eq:be_unitary}; that is, LCU can be used as a subroutine of qubitization. 
However, in this section, we focus on LCU as a standalone approach for realizing QSP.

We begin by outlining the LCU technique in detail.
The framework is built upon two essential subroutines: $\text{PREP}$ and $\text{SEL}$.
The $\text{PREP}$ encodes the $J$ coefficients $\{w_{j}\}$ of the target matrix $H$ on auxiliary register states $\ket{0}_{a}=\ket{0}^{\otimes n_{a}} $ as follows:
\begin{align}
    \text{PREP}:\quad \text{PREP}\;\ket{0}_a = \sum_{j=1}^J \sqrt{\frac{w_j}{\|w\|_1}} \ket{j},
\end{align}
where $\|w\|_1 = \sum_{j=1}^J |w_j|$ is the 1-norm of the matrix $H$.
The $\text{SEL}$ subroutine applies the unitary $U_{j}$ to an input state $\ket{\Omega}$, conditioned on the control register being in state $j$.
Combining these operations, we construct the unitary $U_{\text{LCU}} = \text{PREP}^\dagger \cdot \text{SEL} \cdot \text{PREP}$, which gives
\begin{align}
U_{\text{LCU}} \ket{0}_a \otimes \ket{\Omega} = \frac{1}{\|w\|_1} \ket{0}_a \otimes H\ket{\Omega} + \ket{\text{garbage}^{\perp}}. \end{align}
If a measurement of the auxiliary register yields $\ket{0}_a$, the remaining quantum state is the normalized state given by $H\ket\Omega / \|H\ket\Omega\|$.
The probability of this successful projection is given by
\begin{align}
\label{eq lcu succ p}
    p_\text{succ} = \|H\ket \Omega\|^2/\|w\|_1^2\ .
\end{align}
This procedure extends naturally to QSP.
We exploit the fundamental theorem of algebra, which states that any univariate polynomial with complex coefficients $p(z) = \sum_{k=0}^K a_k z^k$ can be factorized in terms of its roots $z_k$ to take the form $p(z) = a_K \prod_{k=0}^K(z-z_k)$. This directly generalizes to matrix functions and we get 
\begin{align}\label{eq:pH}
     p(H) = a_K \prod_{k=0}^K(H-z_k I)\ . 
 \end{align}
Setting $H$ as a Hermitian matrix, we proceed inductively by applying a sequence of the operators $F_k = H-z_k I$ using LCU.
This results in the transformation
\begin{align}
    \ket{\Psi_{k+1}} = \frac{F_k \ket{\Psi_k}}{\|F_k \ket{\Psi_k}\|} \ .
\end{align}
Since the leading coefficient $a_K$ of the polynomial $p(H)= a_K \prod_{k=1}^K F_k$ cancels out by normalization, we obtain 
\begin{align}
    \ket{\Psi_K} = \frac{p(H)\ket{\Psi_0}}{\|p(H)\ket{\Psi_0}\|}.
\end{align}
Let us next discuss the success probability of this procedure assuming that the Hamiltonian is decomposed into Pauli operators as $H = \sum_{i=1}^J w_i P_i$.
From Eq.~\eqref{eq lcu succ p}, the success probability of post-selection for $k$ step is equal to the conditional probability given that the state $\ket{\Psi_{k-1}}$ at $(k-1)$ step is successfully generated: that is, we have
\begin{align}
\label{conditional succ prob LCU}
\text{Pr}(k \text{-th step success } | \Psi_k )= \frac{\|F_{k-1} \ket {\Psi_{k-1}}\|^2}{(|z_{k-1}| +\|w\|_1)^2}\ .
\end{align}
Thus, the success probability at $K$ step is given by
\begin{align}
\label{eq final succ prob LCU}
\text{Pr}(\text{QSP success}) &= \prod_{k=1}^K  \text{Pr}(k \text{-th step success } | ~\Psi_{k-1} ) =\frac{\|\prod_{k=1}^KF_{k-1}\ket {\Psi_0}\|^2}{\prod_{k=1}^K(|z_{k-1}|+\|w\|_1)^{2K}}.
\end{align}
Suppose that the probability in Eq.~\eqref{conditional succ prob LCU} can be bounded by $1-q$ with $q\in(0,1]$, then we have $\text{Pr}(\text{QSP success}) \le (1-q)^{2K}$, indicating an exponential hardness of successful post-selection.

To address these limitations, we turn to our proposal, DB-QSP. Unlike LCU, DB-QSP constructs deterministic unitary operations that implement the desired state transformations without requiring post-selection and auxiliary qubits. 
This approach could improve the preprocessing initialization for QSP, reducing the overall hardware runtime by eliminating the need for post-selection.

\section{Proofs of Lem.~\ref{thm zero overhead QSP unitary single step} and Thm.~\ref{thm complex QSP}}
\label{app lemma and thm proofs}

\subsection{Proof of Lem.~\ref{thm zero overhead QSP unitary single step}}

For completeness, we restate the statement from the main text.

\begin{lemmaA}[Unitary synthesis for linear polynomials without post-selection]
\label{app thm zero overhead QSP unitary single step} 
Suppose $p(H)=H-\alpha I$ is any linear polynomial of a Hermitian matrix $H$ with $\alpha\in\mathbb{R}$.
Given a state vector $\ket{\Psi} $ with energy mean $E_\Psi = \bra\Psi H\ket\Psi$ and variance $V_{\Psi}=\bra\Psi H^2\ket\Psi -E_\Psi^2$, the unitary synthesis for the action of $p(H)$ can be achieved by 
    \begin{align}
    \label{app DB bridge}
        U_\Psi = e^{s_\Psi[\Psi,H]},
    \end{align}
    with 
    \begin{equation} \label{app eq:duration_monic_linear}
        s_\Psi = \frac{-1}{\sqrt{V_\Psi}}\arccos\left(\frac{E_{\Psi}-\alpha}{\sqrt{V_{\Psi}+\left(E_{\Psi}-\alpha\right)^2}}\right).
    \end{equation}
\end{lemmaA}

\begin{proof}[Proof of Lem.~\ref{app thm zero overhead QSP unitary single step}] 
Firstly, let us verify that $U_\Psi$ is indeed a unitary operator.
For a matrix of the form $e^{W}$ to be unitary, $W$ must be anti-Hermitian, i.e., $W=-W^{\dagger}$.
Since $ [\Psi,H] = -( [\Psi,H])^{\dagger}$, the operator in Eq.~\eqref{app DB bridge} is therefore unitary.

Next, we demonstrate that the unitary operator $e^{W_{H}}$ with $W_{H}=[\Psi,H]$ can be exactly represented by a linear polynomial when applied to the input state $\ket{\Psi}$.
By definition, the unitary operator can be expressed as 
\begin{align}e^{s W_H} = \sum_{k=0}^\infty\frac{s^k}{k!} W_H^k \ .
\end{align}
Now, we observe that 
\begin{align}    \label{app eq: first power of commutator}
W_H \ket{\Psi} =E_\Psi\ket{\Psi} - H\ket{\Psi},
    \end{align}
and $$W_H^2 \ket{\Psi} =E_\Psi W_H\ket{\Psi} - W_H H\ket{\Psi}= E_\Psi^2\ket{\Psi} - E_\Psi H\ket{\Psi} -\bra\Psi H^2\ket\Psi\ket\Psi +E_\Psi H\ket{\Psi} 
    = - V_{\Psi}\ket{\Psi}. $$
This indicates that any even power of the  commutator $W_H$ acting on the state $\ket\Psi$ gives
\begin{align}\label{app eq: even power of commutator}
W_H^{2k}\ket{\Psi} = (-V_{\Psi})^k\ket{\Psi}\ .
\end{align}
Similarly, we have $
W_H^{2k+1}\ket{\Psi} = (-V_{\Psi})^k W_H\ket{\Psi}$
for the cases of odd powers.
Thus, by separating the odd and even terms, we obtain a weighted sum of  $\ket{\Psi}$ and $W_{H}\ket{\Psi}$ with coefficients expressed by sine and cosine functions as
\begin{align}
    e^{sW_H}\ket{\Psi} 
    & = \cos\left(s\sqrt{V_{\Psi}}\right)\ket{\Psi} + \frac{\sin\left(s\sqrt{V_{\Psi}}\right)}{\sqrt{V_{\Psi}}}\,W_H\ket{\Psi} \ .
\end{align}
Using Eq.~\eqref{app eq: first power of commutator}, this simplifies to $e^{sW_H}\ket\Psi = \left(a(s)I + b(s)H \right)\ket{\Psi}$, where $a(s),b(s)$ are real-valued coefficients for any $s\in \mathbb R$:
\begin{align}
    a(s) &= \frac{ E_\Psi }{\sqrt{V_{\Psi}}}\sin\left(s \sqrt{V_{\Psi}}\right) + \cos\left(s \sqrt{V_{\Psi}}\right), \qquad
        b(s) = - \frac{1}{\sqrt{V_{\Psi}}} \sin\left(s \sqrt{V_{\Psi}}\right)\ .
    \end{align}

Finally, an explicit calculation reveals that the ansatz for the duration Eq.~\eqref{app eq:duration_monic_linear} solves the equations $a(s_\Psi)=-\alpha/\|p(H)\ket{\Psi}\|$ and $b(s_\Psi)=1/\|p(H)\ket{\Psi}\|$ where we utilize the equality
\begin{equation}
    \|p(H)\ket{\Psi}\|=\sqrt{V_{\Psi}+\left(E_{\Psi}-\alpha\right)^2}\ .
\end{equation}
    
The proof is concluded by noting that this means that 
\begin{align}
    e^{s_\Psi W_H}\ket\Psi = \frac{\left(H-\alpha I) \right)\ket{\Psi}}{ \|\left(H-\alpha I) \right)\ket{\Psi}\|} \ .
\end{align}  
\end{proof}
\subsection{Useful Preliminary Results}
\label{sec synthesis details}
In this section, we derive an exact formula for implementing an exponential of commutators, $e^{s[\Psi,H]}$, without any approximation or truncation error.

\subsubsection{Effective Idempotence of Exponentials of  $[\Omega,H]$  }
\label{app LCU QITE commutators}
 
We derive an equivalent expression of the unitaries $e^{s \left[\Omega,H\right]}$ found in Eq.~\eqref{app DB bridge}, involving pure states $\Psi$ and the problem Hamiltonian $H $. 
We start with the general Taylor series of the exponential of an operator
\begin{align}e^{s\left[\Psi,H\right]} = \sum_{k=0}^\infty\frac{s^k}{k!} \left(\left[\Psi,H\right]\right)^k \ ,
    \end{align}
    where all $k$-th powers of $s\left[\Psi,H\right]$ contribute to the unitary.
In general, one may approximate this infinite series expansion by truncating it to a degree-$K$ polynomial,
\begin{align}
               e^{s\left[\Psi,H\right]} \approx \sum_{k=0}^K\frac{s^k}{k!} \left(\left[\Psi,H\right]\right)^k \ .
    \end{align}
However, the error $O(s^{K+1})$ requires additional care and investment of resources to control.
Interestingly, we prove in Prop. \ref{proposition effective idempotence} that when $\Psi$ is a pure state, an \emph{exact} polynomial representation can be obtained with 
$K=2$, rather than an approximation.

\begin{propositionA}[Effective idempotence] \label{proposition effective idempotence}
Let $\Psi = \ket \Psi\bra\Psi$ be a pure density matrix associated to state vector $\ket \Psi$ with energy fluctuation $V_{\Psi}=\bra \Psi H^2\ket \Psi-\bra \Psi H\ket \Psi^2$.
Then for any duration $s\in \mathbb R$ we have
    \begin{align}
        e^{s\left[\Psi,H\right]} = I + A(s)\left[\Psi,H\right] + B(s)\left(\left[\Psi,H\right]\right)^2
    \end{align}
    where \vspace{-0.25cm}
    \begin{align}\label{eq:AB}
        A(s) = \frac{\sin\left(s \sqrt{V_{\Psi}}\right)}{\sqrt{V_{\Psi}}} , \qquad
        B(s)=\frac{1-\cos\left(s \sqrt{V_{\Psi}}\right)}{V_{\Psi}} \ .
    \end{align}
\end{propositionA}

\begin{proof}
    We will make a technical calculation showing that
    third power of the commutator is, in fact, directly proportional to the first power of the commutator, with a scaling factor that depends on energy fluctuation:
\begin{align}
    \begin{split}
\left(\left[\Psi,H\right]\right)^3 &=  - V_{\Psi} \left[\Psi,H\right] \ .
    \end{split}
    \end{align}
We call this effective idempotence.
Indeed, in general an operator $W$ is indempotent if $W^2=1$ which implies that $e^{sW} = \cos(s)I + \sin(s) W$.
Here we have the form $A^3 = \alpha A$ with $\alpha \in \mathbb{R}$, similar to idempotence. Effective idempotence has analogous consequences for the solution to the exponential series.
It implies that the $(2k+1)$-th power and the $2k$-th power of the commutator can be written as 
    \begin{align}
        \left(\left[\Psi,H\right]\right)^{2k+1} &= \left(-V_{\Psi} \right)^{k} \left[\Psi,H\right]\ , \qquad
        \left(\left[\Psi,H\right]\right)^{2k} = \left(-V_{\Psi} \right)^{k-1} \left(\left[\Psi,H\right]\right)^2.
    \end{align}
Thus we find for the series of representation of the unitary
    \begin{align}
    \begin{split}
        e^{s\left[\Psi,H\right]} 
        &= I + \left( \sum_{k=0} (-1)^{k} \frac{s^{2k+1} \left(V_{\Psi}\left[H\right]\right)^k }{(2k+1)!}\right)\left[\Psi,H\right] + \left( \sum_{k=1} (-1)^{k-1} \frac{s^{2k} \left(V_{\Psi}\left[H\right]\right)^{k-1} }{(2k)!}\right)\left(\left[\Psi,H\right]\right)^2 \\
        &= I + A(s) \left[\Psi,H\right] + B(s)\left(\left[\Psi,H\right]\right)^2,
    \end{split}
    \end{align}
    where $A(s),B(s)$ defined in Eq.~\eqref{eq:AB}, and we have utilized the Taylor series for sine and cosine in the last equality. 
    We complete the proof by deriving the effective idempotence namely
        \begin{align}
        \left(\left[\Psi,H\right]\right)^3 &= \left( \Psi H - H\Psi \right)^3 =  \left( \Psi H\Psi H - \Psi H^2\Psi -  H\Psi H + H\Psi H\Psi  \right) \left( \Psi H - H\Psi \right)
    \end{align}
    where we have used the assumption that the quantum state is pure, i.e., $\Psi^2=\Psi$.
    Next, we switch from density matrix representation to state vector representation, i.e. we substitute back  $\Psi = \ket \Psi\bra\Psi$.
    Thus, it becomes
    \begin{align}
    \begin{split}
        \left(\left[\Psi,H\right]\right)^3
        &= \left(\braket{H}\left(\Psi H + H\Psi \right) -  H\Psi H - \braket{H^2}\Psi  \right) \left( \Psi H - H\Psi \right) \\
        &= \braket{H}\left(\Psi H \Psi H + H\Psi H\right) -  H\Psi H \Psi H - \braket{H^2}\Psi H \nonumber\\
        &\quad \; - \braket{H}\left(\Psi H^2 \Psi + H\Psi H\Psi \right) +  H\Psi H^2 \Psi + \braket{H^2}\Psi H\Psi \ ,
    \end{split}
    \end{align}
            where we introduce the notation $\braket{H}=\bra \Psi H\ket \Psi$ and $\braket{H^2}=\bra \Psi H^2\ket \Psi$ in the first line, and the second equality is a direct expansion.
            Finally, we repeat the same procedure and the result is given by
            \begin{align}
              \left(\left[\Psi,H\right]\right)^3          
        &= \braket{H}\left(\braket{H}\Psi H + H\Psi H\right) - \braket{H} H\Psi H - \braket{H^2}\Psi H \nonumber\\
        &\quad \;- \braket{H}\left(\braket{H^2}  \Psi + \braket{H} H\Psi \right) +  \braket{H^2} H\Psi  + \braket{H^2}\braket{H}\Psi \\
        &= \braket{H}^2 \Psi H -  \braket{H^2} \Psi H - \braket{H}^2 H\Psi +  \braket{H^2} H\Psi \\
        &= - \left(\braket{H^2} - \braket{H}^2\right) \left[\Psi,H\right] = - \left(V_{\Psi}\left[H\right] \right)\left[\Psi,H\right],
            \end{align}
 where we again use the pure state assumption in the first equality and the definition of $V_{\Psi}$ in the last equality.
\end{proof}

\subsubsection{Exponentials of $[\Psi, H]$ Can Express the Normalized Action of Any Real-Valued Linear Polynomial in $H$}\label{sec: exponential of commutator}

We next extend Lem.~\ref{thm zero overhead QSP unitary single step} to operators of the form $xI+yH$, where $x,y\in\mathbb{R}$ are not both zero.

\begin{lemmaA}
\label{prop any linear polynomial synthesis}
 
    Let $x,y\in \mathbb R$ and  $(x,y)\neq (0,0)$.
    Define the parameter
    \begin{align} \label{eq: app_duration_general}
        s_\Psi = -\frac{\mathrm{sgn}(y)}{\sqrt{V_{\Psi}}} \mathrm{arccos} \left(\frac{x+yE_{\Psi} }{\|(xI+y H)\ket{\Psi}\|} \right).
    \end{align} 
    Then, 
    \begin{align}
       \frac{(xI+y H)\ket{\Psi} }{\|(xI+y H)\ket{\Psi}\|} = (a(s_\Psi)I + b(s_\Psi)H)\ket \Psi \ ,
    \end{align}
  where $a(s_\Psi),b(s_\Psi)$ are real-valued coefficients given by
\begin{align}
    a(s_\Psi) &= \frac{ E_\Psi }{\sqrt{V_{\Psi}}}\sin\left(s_\Psi \sqrt{V_{\Psi}}\right) + \cos\left(s_\Psi \sqrt{V_{\Psi}}\right), \label{a(s) main}
\\
        b(s_\Psi) &= - \frac{1}{\sqrt{V_{\Psi}}} \sin\left(s_\Psi \sqrt{V_{\Psi}}\right)\ .\label{b(s) main}
    \end{align}
\end{lemmaA}

\begin{proof}
    We here consider to match the weights of $I$ and $H$ between the polynomial operation and exponentials of $[\Psi, H]$.
    Specifically, we solve the following two equations;
    \begin{align}
    \frac{x}{\|(xI+y H)\ket{\Psi}\|} &= \frac{ E_\Psi }{\sqrt{V_{\Psi}}}\sin\left(s \sqrt{V_{\Psi}}\right) + \cos\left(s \sqrt{V_{\Psi}}\right), \label{a(s)_1}
\\
        \frac{y}{\|(xI+y H)\ket{\Psi}\|} &= - \frac{1}{\sqrt{V_{\Psi}}} \sin\left(s \sqrt{V_{\Psi}}\right)\ .\label{b(s)_1}
    \end{align}
    By computing Eq.~\eqref{a(s)_1}$+E_{\Psi}\times$Eq.~\eqref{b(s)_1}, we get
    \begin{equation}
        \frac{x+yE_{\Psi}}{\|(xI+y H)\ket{\Psi}\|} =  \cos\left(s_{\Psi} \sqrt{V_{\Psi}}\right)\ .
    \end{equation}
Thus, computing the inverse of the cosine function yields the desired solution.
However, since both $s_\Psi$ and $-s_\Psi$ satisfy the equation, the sign must be determined explicitly.
Eq.~\eqref{b(s)_1} indicates that the sign of $\sin(s_{\Psi}\sqrt{V_{\Psi}})$ is given by $-\mathrm{sgn}(y)$.
This leads to the expression for the duration $s_{\Psi}$ shown in Eq.~\eqref{eq: app_duration_general}.

We lastly verify that this solution is consistent.
By substituting Eq.~\eqref{eq: app_duration_general} into Eq.~\eqref{b(s) main}, we find
\begin{align}
     b(s_\Psi) &=  \frac{\mathrm{sgn}(y)}{\sqrt{V_{\Psi}}} \sin\left( \mathrm{arccos} \left(\frac{x+yE_{\Psi} }{\|(xI+y H)\ket{\Psi}\|} \right)\right) \ .
\end{align}
Applying $\sin(x) = \sqrt{1-\cos^2(x)}$ further reveals that
\begin{align}
     b(s_\Psi) &=  \frac{\mathrm{sgn(y)}}{\sqrt{V_{\Psi}}} \sqrt{1 - \frac{(x+yE_{\Psi})^2 }{\|(xI+y H)\ket{\Psi}\|^2} }\ .
\end{align}
With the identity $\|(xI+y H)\ket{\Psi}\|^2 =  (x+yE_\Psi)^2+ y^2 V_\Psi$, this simplifies to
\begin{align}
     b(s_\Psi) &= \frac{\mathrm{sgn}(y)}{\sqrt{V_{\Psi}}}  \frac{|y| \sqrt{V_\Psi}  }{\|(xI+y H)\ket{\Psi}\|}\\
    &= \frac{y}{\|(xI+y H)\ket{\Psi}\|} \ .
\end{align}
Similarly, we use the relation  $a(s_\Psi) = -E_\Psi b(s_\Psi) +\cos\left(s_\Psi \sqrt{V_{\Psi}}\right)$, which gives
\begin{align}
   a(s_\Psi) &=  \frac{-yE_\Psi}{\|(xI+y H)\ket{\Psi}\|} +  \frac{x+yE_\Psi}{\|(xI+y H)\ket{\Psi}\|}\\
   &=    \frac{x}{\|(xI+y H)\ket{\Psi}\|} \ .
\end{align}

\end{proof}

\subsection{Proof of Thm.~\ref{thm complex QSP}}
\label{sec:proof_theorem}

We again restate Thm.~\ref{thm complex QSP} in the main text.
\begin{theoremA}[Unitary synthesis for QSP without post-selection]
\label{thm complex QSP_appendix}
Suppose an input state $\ket{\Psi_0}$ and any polynomial $p(H)$ of degree $K$ for a given Hermitian matrix $H$ in the form of Eq.~\eqref{eq:pH_}.
Given energy mean $E_k = \bra{\Psi_{k}} H\ket{\Psi_{k}}$ and variance $V_{k}=\bra{\Psi_{k}} H^2\ket{\Psi_{k}} -E_k^2$, the unitary synthesis for $p(H)$ can be achieved by 
\begin{equation} \label{eq:qsp_wo_ps_app}
    \frac{p(H)\ket{\Psi_0}}{\|p(H)\ket{\Psi_0}\|}=\prod_{k=0}^{K-1} e^{i \theta_k \Psi_{k}}  e^{s_{k}[\Psi_{k},H]}\ket{\Psi_0},
\end{equation}
with \begin{equation} \label{eq: duration_angle_def_app}
    s_k = \frac{-1}{\sqrt{V_k}}\arccos\left(\frac{|E_{k}-z_{k}|}{\sqrt{V_{k}+|E_{k}-z_{k}|^2}}\right), \qquad \text{and}\qquad     \theta_k = \arg\left(\frac{E_k-z_k}{|E_k-z_k|}\right)\ .
\end{equation}
Here, we recursively define the state $\ket{\Psi_{k}}$ as $\ket{\Psi_{k+1}}=e^{i \theta_k \Psi_{k}}  e^{s_{k}[\Psi_{k},H]}\ket{\Psi_{k}}$.
\end{theoremA}
\begin{proof}
Let $z_k$ be the roots of $p(H)$ as in Eq.~\eqref{eq:pH_}.
We iterate over the roots and at each step $k$, we will find $\theta_k$ and $s_k$ such that 
the unitary $U_k =  e^{i \theta_k \Psi_{k}}e^{s_{k}[\Psi_{k},H]}$ will implement the state
\begin{align}
   \ket{\Psi_{k+1}} =\frac{(H-z_k I) \ket{\Psi_k}}{\|(H-z_k I) \ket{\Psi_k}\|}
\end{align}
as $\ket{\Psi_{k+1}} =U_k \ket{\Psi_k}$.
Let us comment that, if we apply the $k$-th filter fragment $F_k = H-z_kI$, the normalization is given by
\begin{align}
 \|(H-z_kI)\ket{\Psi}\|
&=V_\Psi +|E_\Psi-z_k|^2\  .
\label{geneal z+H omega normalization}
\end{align}

We cannot use Lem.~\ref{thm zero overhead QSP unitary single step} directly because in general $z_k\in\mathbb C$, while polynomials with only real roots such as Chebyshev polynomials can be realized by directly applying Lem.~\ref{thm zero overhead QSP unitary single step}.
Instead, in general cases, we associate to $z_k$ the real number
\begin{equation}
\label{eq def w_k}
    u_k = E_k - |E_k-z_k|
\end{equation}
which is real and using Lem.~\ref{thm zero overhead QSP unitary single step} we set $s_k$ such that
\begin{align}
e^{s_{k}[\Psi_{k},H]}\ket{\Psi_{k}} = \frac{(H-u_kI)\ket{\Psi_{k}}}{\|(H-u_kI)\ket{\Psi_{k}}\|}\ .
\end{align}

We define $\theta_k$ to be within $[0,2\pi)$ and satisfy $e^{i\theta_k} = \frac{E_k-z_k}{|E_k-z_k|}$. We next observe that using that $\Psi_k$ is pure we have the form $e^{i\theta_k \Psi_{k}}=I+(e^{i\theta_k}-1)\Psi_{k}$, we get the following expression
\begin{align}
    \ket{\Psi_{k+1}}  &= \frac{(I+(e^{i\theta_k}-1)\Psi_{k})(H-u_kI)\ket{\Psi_{k}}}{\|(H-u_kI)\ket{\Psi_{k}}\|} = \frac{(H+e^{i\theta_k}(E_k-u_k)I-E_kI)\ket{\Psi_{k}}}{\|(H-u_kI)\ket{\Psi_{k}}\|}   \ .
\end{align}
The definitions above were such that $e^{i\theta_k}(E_k-w_k)= E_k-z_k$
which leads to a cancellation and
\begin{align}
    \ket{\Psi_{k+1}} 
    &= \frac{(H-z_kI)\ket{\Psi_{k}}}{\|(H-u_kI)\ket{\Psi_{k}}\|} \ .   
\end{align}
Here the numerator involves $z_k$ as desired but the norm is an expression involving $w_k$.
We have
\begin{align}
\|(H-u_kI)\ket{\Psi_{k}}\| &= V_k +E_k^2 -2u_kE_k+u_k^2 = V_k +(E_k-u_k)^2\ ,
\end{align}
which means that, using Eq.~\eqref{eq def w_k}, we arrive at the form in Eq.~\eqref{geneal z+H omega normalization}
\begin{align}
    \|(H-u_kI)\ket{\Psi_{k}}\| 
&= V_k +|E_k-z_k|^2\ =\|(H-z_kI)\ket{\Psi_{k}}\|\  .
\end{align}
Thus the norms match and we conclude that the unitaries implement the desired action of $F_k=H-z_kI$.
\end{proof}

We conclude this section by discussing the range of $s_{k}$, which is relevant for analyzing implementation costs and circuit depth.
As shown in Eq.~\eqref{eq: duration_angle_def_app}, the duration $s_{k}$ is given by 
\begin{equation} 
    s_k = \frac{-1}{\sqrt{V_k}}\arccos\left(\frac{|E_{k}-z_{k}|}{\sqrt{V_{k}+|E_{k}-z_{k}|^2}}\right) .
\end{equation}
First, to see if $|s_{k}|$ is a decreasing function with respect to $V_{k}$ we differentiate the two components $g(V_{k})=\sqrt{V_{k}}$ and $ f(V_{k})=\arccos\left(\frac{|E_{k}-z_{k}|}{\sqrt{V_{k}+|E_{k}-z_{k}|^2}}\right) $;
\begin{align}
    g'(V_{k}) &= \frac{1}{2\sqrt{V_{k}}}\\
    f'(V_{k}) &= -\frac{1}{\sqrt{1-\left(\frac{|E_{k}-z_{k}|}{\sqrt{V_{k}+|E_{k}-z_{k}|^2}}\right)^2}} \cdot \left(\frac{|E_{k}-z_{k}|}{\sqrt{V_{k}+|E_{k}-z_{k}|^2}}\right)'=   \frac{|E_{k}-z_{k}|}{2\sqrt{V_{k}}(V_{k}+|E_{k}-z_{k}|^2)}
\end{align}
where we used $(\arccos(y))'=-1/\sqrt{1-y^2}$.
Thus, using quotient rule, we have
\begin{align}
(|s_{k}|)’ &= \frac{f’(V_{k})g(V_{k})-g’(V_{k})f(V_{k})}{g^2(V_{k})} =\frac 1{V_k} \left(\frac12\frac{|E_{k}-z_{k}|}{V_{k}+|E_{k}-z_{k}|^2}- \frac 12 
\frac{1}{\sqrt{V_k}}\arccos\left(\frac{|E_{k}-z_{k}|}{\sqrt{V_{k}+|E_{k}-z_{k}|^2}}\right)\right). 
\end{align}
Next we define $\displaystyle x = \frac{|E_{k}-z_{k}|}{\sqrt{V_{k}+|E_{k}-z_{k}|^2}}=\cos(\alpha)$ which implies that $\displaystyle \frac 1{\sqrt{V_k}} = \frac x {|E_{k}-z_{k}|} \frac1 {\sqrt{1-x^2}}$ and so we  find
\begin{align}
(|s_{k}|)’ 
&=\frac {x}{2V_k|E_{k}-z_{k}|} \left(x-   
\frac 1 {\sqrt{1-x^2}}\arccos\left(x\right)\right)\\
&=\frac {x}{2V_k|E_{k}-z_{k}|} \left(\cos(\alpha)-   
\frac \alpha {\sin(\alpha)}\right)\\
&=\frac {x}{2V_k|E_{k}-z_{k}|} \frac {\frac12\sin(2\alpha)-   
 \alpha }{\sin(\alpha)}\\
 &=\frac {x}{4V_k|E_{k}-z_{k}|} \frac {\sin(2\alpha)-   
 2\alpha }{\sin(\alpha)} \le 0\ ,
\end{align}
where we use the fact $\sin x \leq x$ in the last line.
Then, the maximum value of $s_{k}$ arises when $V_{k}=0$.
However, since $V_{k}$ appears in the denominator, we cannot simply compute the value of $s_{k}$ at $V_{k}=0$.
Thus, we apply the L'H\^opital's rule, we get
\begin{equation}
    \lim_{V_{k}\to0} |s_{k}| = \lim_{V_{k}\to 0}\frac{f'(V_{k})}{g'(V_{k})} = \frac{|E_{k}-z_{k}|}{V_{k}+|E_{k}-z_{k}|^2} \biggl|_{V_{k}=0} = \frac{1}{|E_{k}-z_{k}|}.
\end{equation}
Thus, the duration $|s_{k}|$ is upper-bounded by $1/|E_{k}-z_{k}|$.

\newpage
 \section{Notions of Stability for Unitary Synthesis of Exact Formula in Thm.~\ref{thm complex QSP} }
\label{app notions of stability}
In this section we explore the unitary synthesis of Thm.~\ref{thm complex QSP} in more detail.
in Sec.~\ref{app sec DBQSP} we prove that discretizations using group commutator approximation can converge to the desired QSP.

We then analyze the sensitivity of the exact formula in Thm.~\ref{thm complex QSP} to perturbations in the input parameters.
We begin by studying a question similar to an existing stability result for QSP using block-encodings.
Concretely, the output of QSP synthesis using qubitization will depend on any errors in the block-encoding of the input operator $H$.
This enjoys a certain degree of stability; namely, given block-encodings of $H$ and $H'$ for $\|H\|\le 1$, $\|H'\|\le 1$, their transformed block-encodings are also close, $\|p(H)-p(H')\| \le 4K\|H-H'\|$~\cite{gilyen2019quantum}.
In Sec.~\ref{app sec hamiltonian prert}, we derive a bound in this scenario.

We then focus on the impact of imperfect parameters $\theta$ and $s$ on the performance.
In Sec.~\ref{sec QSP perturbation}, we first analyze the impact of deviations in the parameters from their ideal values.
However, this analysis alone is insufficient for practical scenarios, as statistical errors arise when estimating energy and variance from a noisy state. 
To address limitation, we extend the results to the situation where the estimated energy and variance may still differ from their ideal values even in the limit of finite measurement shots.
Sec.~\ref{sec QSP Hamiltonian stability} explores this extension, beginning with the single-step case before generalizing to arbitrary steps. These results provide insight into the statistical estimation requirements necessary for achieving a converging QSP synthesis.

Hence, we further extend the result to the case where the statistical noise happens when the energy and variance is different from the ideal situation even if we have the infinite number of measurement shots.
To address this, Sec.~\ref{sec QSP Hamiltonian stability} starts with a single-step case, followed by the arbitrary steps.
This result sheds light on the demands of statistical estimation required to obtain a converging QSP synthesis.

\begin{table*}[h!]
\begin{tblr}{
    colspec = {|p{1.5cm}|p{3.5cm}|p{1.6cm}|p{10cm}|},
}
\hline
 \textbf{Section} & \textbf{Main Focus} & \textbf{Proof} & \textbf{Final Results}  \\
\hline
Sec.~\ref{app sec DBQSP} &  Difference between the exact formula in Thm.~\ref{thm complex QSP} and DB-QSP & Prop.~\ref{prop: DB-QSP convergence}& \begin{align*}
    \|\ket{\Psi(\bm{\theta},\bm{s})} - \ket{\omega_{K}}\| \le \frac{4}{3} \sqrt{\frac{\zeta}{N}
    }(1+6\xi)^{K} \ .
\end{align*} \\

\hline
Sec.~\ref{app sec hamiltonian prert} & Stability with respect to the difference in Hamiltonians & Prop.~\ref{prop: perturb in H}& \begin{align*}
    \left\| \ket{\Psi_{\theta,s}(H)}-\ket{\Psi_{\theta,s}(\tilde H)}\right\| \le \frac 13(1+6 \zeta)^K\|H-\tilde H\|\ .
\end{align*}\\
\hline

Sec.~\ref{sec QSP perturbation} & Sensitivity to changes from the exact angles  $\theta$ and $s$ &Prop.~\ref{angle stability prop app} & \begin{align*}
    \|\ket{\Psi_{H}(\bm{\theta},\bm{s})} - \ket{\tilde\Psi_{H}(\tilde{\bm{\theta}},\tilde{\bm{s}})}\| \le\frac{\max(\delta_s,\delta_\theta)}{3\zeta}(1+6\zeta)^K \ .
\end{align*} \\

\hline
Sec.~\ref{sec QSP Hamiltonian stability} & Error in a single step caused by erroneous estimation of energy and variance & Prop.~\ref{app prop stat err propagation}& \begin{align*}
     \left\| e^{i\theta_\Psi \Psi}e^{s_\Psi[\Psi,H]}\ket{\Psi} - e^{i\overline \theta \Psi }e^{\overline s_\Psi[\Psi,H]}\ket{\Psi }\right\| 
    &\le 20\eta^4\max(\delta_{V'},\delta_{E'}) 
\end{align*}  \\
\hline
Sec.~\ref{sec QSP Hamiltonian stability} & Error in $K$ steps using the estimated QSP parametrization $(\overline\theta,\overline s)$& Prop.~\ref{eq app qsp estimatin stability}& \begin{align*}
    \|\ket{\Psi_{H}(\bm{\theta},\bm{s})} - \ket{\overline\Psi_{H}( {\bm{\overline\theta}}, {\bm{\overline s}})}\| \le(14+120\eta^4)^{K}\max(\delta_V,\delta_E) \ .
\end{align*} \\
\hline
\end{tblr}
\caption{\textbf{Summary of the results explored in this section.} Props.~\ref{prop: DB-QSP convergence} and~\ref{eq app qsp estimatin stability} are the key results, but the other derivations should be helpful in understanding their proof.
For notation, please refer to the corresponding sections. In addition, we introduce the shorthand $\delta_{E'}=|E_\Psi-E'|$ and $\delta_{V'}=|V_\Psi-V'|$ in this table.
}
\end{table*}

\subsection{Convergence of DB-QSP}
\label{app sec DBQSP}
\setcounter{theorem}{2} 
\begin{propositionA}[DB-QSP convergence]\label{prop: DB-QSP convergence}
Suppose $H$ is a Hermitian matrix whose spectral radius does not exceed unity, i.e., $\|H\|\le1$. 
Let $\zeta = \max(\bm{\theta},\bm{s})$ be the maximum value across all elements in $\bm{\theta}=(\theta_0,\ldots,\theta_{K-1})$ and $\bm{s}=(|s_0|,\ldots,|s_{K-1}|)$.
For the analysis, we define the state constructed by DB-QSP with $\sn = \sqrt{|s_k|/N}$
\begin{align}
    \ket{\omega_{K}} =\prod_{k=0}^{K-1}   e^{ i \theta_k \omega_{k}}  \left(
e^{i\sn \omega_{k}}
e^{i\sn H}
e^{-i\sn \omega_{k}}
e^{-i\sn H}
\right)^N\ket{\omega_0}\ 
\end{align}
We also define the exact QSP state derived from Thm.~\ref{thm complex QSP} 
\begin{align} 
    \ket{\Psi (\bm{\theta},\bm{s})} = \prod_{k=0}^{K-1} e^{i \theta_k \Psi_{k}}  e^{s_{k}[\Psi_{k},H]}\ket{\Psi_0}\ .
\end{align}
Then we have 
\begin{align}
    \|\ket{\Psi(\bm{\theta},\bm{s})} - \ket{\omega_{K}}\| \le \frac{4}{3} \sqrt{\frac{\zeta}{N}
    }(1+6\xi)^{K} \ .
\end{align}
\end{propositionA}
\begin{proof}
 
Let us define the intermediate QSP states as
     \begin{align}
         \ket{\Psi_k}= \prod_{k'=0}^{k-1} e^{i \theta_{k'} \Psi_{k'}}  e^{s_{k'}[\Psi_{k'},H]}\ket{\Psi_0}
    \end{align}   
and the intermediate DB-QSP states as 
\begin{align}
 \ket{\omega_{k+1}}=e^{ i \theta_k \omega_{k}}  \left(
e^{i\sn \omega_{k}}
e^{i\sn H}
e^{-i\sn \omega_{k}}
e^{-i\sn H}
\right)^N\ket{\omega_{k}}\ .
\end{align}
First, we decompose the difference between the updated QSP states and the DB-QSP states as follows:
\begin{align}
\|\ket{\Psi_{k+1}}-\ket{\omega_{k+1}}\| 
        =& \|e^{i\theta_{k}\Psi_k}e^{s_{k}[\Psi_{k},H]}\ket{\Psi_{k}}- e^{ i \theta_k \omega_{k}}  \left(
e^{i\tilde s_k \omega_{k}}
e^{i\tilde s_k H}
e^{-i\tilde s_k \omega_{k}}
e^{-i\tilde s_k H}
\right)^N\ket{\omega_k} \|\\
=& \Big\| \left( e^{i\theta_k \Psi_k}\, e^{s_k [\Psi_k,H]} \ket{\Psi_k} - e^{i\theta_k \omega_k}\, e^{s_k [\Psi_k,H]} \ket{\Psi_k} \right)\nonumber\\
&+\left( e^{i\theta_k \omega_k}\, e^{s_k [\Psi_k,H]} \ket{\Psi_k} - e^{i\theta_k \omega_k}\, e^{s_k [\omega_k,H]} \ket{\Psi_k} \right) \nonumber\\
& 
+\left( e^{i\theta_k \omega_k}\, e^{s_k [\omega_k,H]} \ket{\Psi_k} - e^{i\theta_k \omega_k}\, e^{s_k [\omega_k,H]} \ket{\omega_k} \right) \nonumber\\
&+\left( e^{i\theta_k \omega_k}\, e^{s_k [\omega_k,H]} \ket{\omega_k} - e^{i\theta_k \omega_k}\, \left(e^{i\sn \omega_k}\, e^{i\sn H}\, e^{-i\sn \omega_k}\, e^{-i\sn H}\right)^N \ket{\omega_k} \right)\Big\| \\
\leq& \|  \left( e^{i\theta_k \Psi_k} -  e^{i\theta_k \omega_k} \right)e^{s_k [\Psi_k,H]} \ket{\Psi_k} \| \nonumber\\
&+\| e^{i\theta_k \omega_k}\, \left(e^{s_k [\Psi_k,H]} -   e^{s_k [\omega_k,H]}\right) \ket{\Psi_k} \| 
 \nonumber\\
& + \|e^{i\theta_k \omega_k}\, e^{s_k [\omega_k,H]}\left(  \ket{\Psi_k} -  \ket{\omega_k} \right)\| \nonumber\\
&+\|e^{i\theta_k \omega_k}\, \left(e^{s_k [\omega_k,H]}  -  \left(e^{i\sn \omega_k}\, e^{i\sn H}\, e^{-i\sn \omega_k}\, e^{-i\sn H}\right)^N \right)\ket{\omega_k} \| \label{eq: DB-QSP convergence intermediate}
\end{align}
where we use triangle inequality to obtain the last inequality.
Next, we evaluate these terms separately.
\begin{enumerate}
    \item Using the definition of the operator norm
    \begin{align}
        \|  \left( e^{i\theta_k \Psi_k} -  e^{i\theta_k \omega_k} \right)e^{s_k [\Psi_k,H]} \ket{\Psi_k} \|&\leq  \| e^{i\theta_k \omega_k}- e^{i\theta_k \Psi_k}\|  
        \leq |\theta_k|\, \|\Psi_k- \omega_k\| 
    \end{align}
   where we utilize the inequality $\|e^{A}-e^{B}\|\le\|A-B\|$ for unitary operator and the fact $\|AB\| \leq \|A\|\|B\|$.
    Moreover, note that $ \|\Psi_k   -  \omega_k\| \le 2  \left\| \ket{\Psi_{k}}-  \ket{\omega_k} \right\|$ and thus we obtain
    \begin{align}
       \|  \left( e^{i\theta_k \Psi_k} -  e^{i\theta_k \omega_k} \right)e^{s_k [\Psi_k,H]} \ket{\Psi_k} \|\leq 2 |\theta_k| \,\| \ket{\Psi_k}- \ket{\omega_k}\| 
    \end{align}

    \item For the second term, we have
    \begin{align}
        \|e^{i\theta_k \omega_k}\left(  e^{s_k [\Psi_k,H]}   - e^{i\theta_k \omega_k}\right)\, e^{s_k [\omega_k,H]}  \ket{\Psi_k}\| &\leq \|e^{i\theta_k \omega_k}\| \cdot \| \left(e^{s_k [\Psi_k,H]} - e^{s_k [\omega_k,H]}\right)\ket{\Psi_k} \| \\
        &\leq \| e^{s_k [\Psi_k,H]} - e^{s_k [\omega_k,H]} \| \\
                &\leq |s_k| \, \| \left[\Psi_k-\omega_k,H\right]\|
    \end{align}
    where we again employ the unitary invariance and normalised state assumption in the second line; and the property  $\|e^{A}-e^{B}\|\le\|A-B\|$ in the last line.
    Next, using the bound $\|[A,B]\|\le 2\|A\|\|B\|$, it can be further simplified to
    \begin{align}
         \|e^{i\theta_k \omega_k}\left(  e^{s_k [\Psi_k,H]}   - e^{i\theta_k \omega_k}\right)\, e^{s_k [\omega_k,H]}  \ket{\Psi_k}\| &\leq 2 |s_k| \,\| \Psi_k-\omega_k\| \cdot \|H\|
    \end{align}
    Similar to the first term, employing the bound $ \|\Psi_k   -  \omega_k\| \le 2  \left\| \ket{\Psi_{k}}-  \ket{\omega_k} \right\|$ and the assumption that $\|H\|\le1$, we get
     \begin{align}
         \|e^{i\theta_k \omega_k}\left(  e^{s_k [\Psi_k,H]}   - e^{i\theta_k \omega_k}\right)\, e^{s_k [\omega_k,H]}  \ket{\Psi_k}\| &\leq 4 |s_k| \,\| \ket{\Psi_{k}}-  \ket{\omega_k}\|
     \end{align}

    \item For the third term, since $e^{i\theta_k \omega_k}\, e^{s_k [\omega_k,H]} $ is unitary operator, the third term can be simplified to \begin{align}
         \|e^{i\theta_k \omega_k}\, e^{s_k [\omega_k,H]}\left(  \ket{\Psi_k} -  \ket{\omega_k} \right)\| &\leq \|\ket{\Psi_k}-\ket{\omega_k}\|
    \end{align}
    where we use the unitary invariance property of norm.
    \item Finally, for the fourth term, it becomes
    \begin{align}
        &\|e^{i\theta_k \omega_k}\, \left(e^{s_k [\omega_k,H]}  -  \left(e^{i\sn \omega_k}\, e^{i\sn H}\, e^{-i\sn \omega_k}\, e^{-i\sn H}\right)^N \right)\ket{\omega_k} \| \nonumber\\
        &\leq \|e^{  {s}_{k}[\omega_{k}, H]} -\left(
e^{i\tilde s_k \omega_{k}}
e^{i\tilde s_k H}
e^{-i\tilde s_k \omega_{k}}
e^{-i\tilde s_k H}
\right)^N \|
    \end{align}
   Using  upper bound in Lemma.~(9) from \cite{gluza2024doublebracket} by replacing $s_k \to \sn$, we have
\begin{align}\label{eq:GCI_compilation}
  \left\|  e^{i\sn \omega_k}e^{i\sn H}    
    e^{-i\sn \omega_k}
    e^{-i\sn H}- e^{\sn[\omega_k,H]} \right\| \leq   |s_k|^{3/2} N^{-3/2} \bigg( \|[H, [H, \omega_k]]\| + \|[\omega_k, [\omega_k, H]]  \|\bigg) \ ,
\end{align}
By the definition of $\sn$ and  telescoping, we have 
\begin{align}
     &\|e^{  {s}_{k}[\omega_{k}, H]} -\left(
e^{i\sn \omega_{k}}
e^{i\sn H}
e^{-i\sn \omega_{k}}
e^{-i\sn H}
\right)^N \| \nonumber\\
&\leq |s_k|^{3/2}/ \sqrt N \times \left( \|[H, [H, \omega_k]]\| + \|[\omega_k, [\omega_k, H]]  \| \right) \\
&\leq 2|s_k|^{3/2}/ \sqrt N \times \left( \|[H, \omega_k]\|\times\|H\| + \|[\omega_k, H] \|\times\|\omega_k\|  \right)\\
&\leq 4|s_k|^{3/2}/ \sqrt N \times \left( \| \omega_k\|\times\|H\|^2 + \|H \|\times\|\omega_k\| ^2 \right)
\end{align}
where we recall the bound $\|[A,B]\|\le 2\|A\|\|B\|$ in the second and third line. Since we assume that $\|H\| \leq 1$ and $\|\omega_k=1\|$, we achieve
\begin{align}
    \|e^{  {s}_{k}[\omega_{k}, H]} -\left(
e^{i\sn \omega_{k}}
e^{i\sn H}
e^{-i\sn \omega_{k}}
e^{-i\sn H}
\right)^N \| &\leq 8|s_k|^{3/2}/ \sqrt N\ .
\end{align}
\end{enumerate}
Collecting all terms, Eq.~\eqref{eq: DB-QSP convergence intermediate} becomes
\begin{align}
       \| \ket{\Psi_{k+1}}-  \ket{\omega_{k+1}}\| &\leq (1+2|\theta_{k}|+4|s_{k}|) \| \ket{\Psi_{k}}-  \ket{\omega_{k}}\| + 8|s_k|^{3/2}/\sqrt{N} \\
    & \le (1+6\zeta) \| \ket{\Psi_{k}}-  \ket{\omega_{k}}\| + 8\zeta^{3/2}/\sqrt{N}\ .
\end{align}
where we use the definition $\zeta = \max_{k=1,\ldots,K}(\theta_k,s_k)$ to obtain last line.
Iterating this recursive bound, we get
\begin{align}
    \| \ket{\Psi_{k+1}}-  \ket{\omega_{k+1}}\| \leq \frac{8\zeta^{3/2}}{\sqrt{N}} \sum_{i=0}^k (1+6\zeta) ^i 
    &=\frac{8\zeta^{3/2}}{\sqrt{N}} \times\frac{(1+6\xi)^{k+1}-1}{(1+6\xi)-1} \\
    &\leq \frac{4}{3}\sqrt{\frac{\xi}{N}}(1+6\xi)^{k+1}\ .
\end{align}
Setting $K=k+1$, the proposition statement is justified.

\end{proof}

\subsection{Perturbation of the Hamiltonian}
\label{app sec hamiltonian prert}
Using Thm.~\ref{thm complex QSP} we define $\ket{\Psi_{\bm{\theta},\bm{s}}(H)} = \prod_{k=1}^{K} e^{i \theta \Psi_{k}}  e^{s_{k}[\Psi_{k},H]}\ket{\Psi_0}$. This definition indicates that we will hold the angles $\theta_k $ and $s_k$ fixed but consider what happens if the Hamiltonian is perturbed.

\begin{propositionA}[QSP task stability]\label{prop: perturb in H}
Suppose $H$ is a Hermitian matrix whose spectral radius does not exceed unity, i.e., $\|H\|\le1$. 
Let $\zeta = \max(\bm{\theta},\bm{s})$ be the maximum value across all elements in $\bm{\theta}=(\theta_0,\ldots,\theta_{K-1})$ and $\bm{s}=(|s_0|,\ldots,|s_{K-1}|)$.
Then, we have
\begin{align}
    \left\| \ket{\Psi_{\theta,s}(H)}-\ket{\Psi_{\theta,s}(\tilde H)}\right\| \le \frac 13(1+6 \zeta)^K\|H-\tilde H\|\ .
\end{align}
\end{propositionA}
\begin{proof}
Let us define the intermediate QSP states
     \begin{align}
         \ket{\Psi_k}= \prod_{k'=0}^{k-1} e^{i \theta_{k'} \Psi_{k'}}  e^{s_{k'}[\Psi_{k'},H]}\ket{\Psi_0}
    \end{align}   
    and analogously $\ket{\tilde \Psi_k}$ are the intermediate states of QSP with $\tilde H$.
    Thus, the difference between $\ket{\Psi_{k+1}}$ and $\ket{\tilde \Psi_{k+1}}$ is given by
    \begin{align}
           \|\ket{\Psi_{k+1}}-\ket{\tilde{\Psi}_{k+1}}\| &= \|e^{i\theta_{k}\Psi_k}e^{s_{k}[\Psi_{k},H]}\ket{\Psi_{k}}- e^{i{\theta}_{k}\tilde{\Psi}_k}e^{{s}_{k}[\tilde{\Psi}_{k},\tilde H]}\ket{\tilde\Psi_{k}}\|\ .
    \end{align}
    Next, following the same procedure in Eq.~\eqref{eq: DB-QSP convergence intermediate} from Subsec.~\ref{app sec DBQSP}, we add and subtract the term $\{e^{i\theta_{k}\Psi_k}e^{s_{k}[\Psi_{k},H]}\ket{\tilde \Psi_{k}}, \;e^{i\theta_{k}\tilde \Psi_k}e^{s_{k}[\Psi_{k},H]}\ket{\tilde \Psi_{k}}\}$ to split them into multiple norm calculations via triangle inequality.

    Consequently, the result is
    \begin{align}
    \|\ket{\Psi_{k+1}}-\ket{\tilde{\Psi}_{k+1}}\|
 &\le \| \ket{\Psi_{k}}-  \ket{\tilde\Psi_{k}}\|+
 \|e^{i\theta_{k}\Psi_k} - e^{i{\theta}_{k}\tilde\Psi_{k}}\|+\|e^{s_{k}[\Psi_{k},H]}- e^{ {s}_{k}[\tilde{\Psi}_{k},\tilde H]}\|\\
 &\le \| \ket{\Psi_{k}}-  \ket{\tilde\Psi_{k}}\|+
 |\theta_{k}|\cdot
 \|\Psi_k   -  \tilde\Psi_{k} \|+|s_k|\cdot \| [\Psi_{k},H]-[\tilde{\Psi}_{k},\tilde H] \|\ , \label{eq: perturbation in H intermediate}
\end{align}
where we recall the unitary invariance property of norm in the first inequality and we utilize the formula $\|e^{A}-e^{B}\|\le\|A-B\|$ in the second inequality.
We then simplify these three terms separately.
\begin{enumerate}
    \item For the first term $ \| \ket{\Psi_{k}}-  \ket{\tilde\Psi_{k}}\|$, it remains unchanged.

     \item For the second term, it becomes
     \begin{align}
          |\theta_{k}|\cdot
 \|\Psi_k   -  \tilde\Psi_{k} \| \leq 2 |\theta_{k}|\cdot \|\ket{\Psi_k}   -  \ket{\tilde\Psi_{k}} \| \ ,
     \end{align}
     where we use the relation  $ \|\Psi_k   -  \tilde\Psi_{k} \|\le 2 \| \ket{\Psi_{k}}-  \ket{\tilde\Psi_{k}}\|$.

\item For the third term, we rewrite it as
\begin{align}
 |s_k| \cdot  \| [\Psi_k, H] - [\tilde{\Psi}_k, \tilde{H}] \|
&= |s_k| \cdot \|\Psi_k H - H \Psi_k - \left(\tilde{\Psi}_k \tilde{H} - \tilde{H}\tilde{\Psi}_k\right) \|\\
&= |s_k| \cdot\| \Psi_k H - \Psi_k \tilde{H} + \Psi_k \tilde{H} - \tilde{\Psi}_k \tilde{H} - H \Psi_k + H \tilde{\Psi}_k - H \tilde{\Psi}_k + \tilde{H}\tilde{\Psi}_k  \|\\
&= |s_k| \cdot \|\Psi_k\,(H - \tilde{H}) + \left(\Psi_k - \tilde{\Psi}_k\right)\tilde{H} - H\left(\Psi_k - \tilde{\Psi}_k\right) - (H-\tilde{H})\,\tilde{\Psi}_k  \| \ .
\end{align}
By triangle inequality and operator norm's definition, we obtain
\begin{align}
   |s_k| \cdot \| [\Psi_{k},H]-[\tilde{\Psi}_{k},\tilde H] \| \le 2|s_k| \cdot\|\Psi_k   -  \tilde\Psi_{k} \| \cdot \|H\| + 2 |s_k| \cdot\|H-\tilde H\| \ .
\end{align}
Similarly, using  $ \|\Psi_k   -  \tilde\Psi_{k} \|\le 2 \| \ket{\Psi_{k}}-  \ket{\tilde\Psi_{k}}\|$, it becomes
\begin{align}
     \| [\Psi_{k},H]-[\tilde{\Psi}_{k},\tilde H] \| \le 4|s_k| \cdot\|\ket{\Psi_k}   -  \ket{\tilde\Psi_{k}} \| + 2 |s_k| \cdot\|H-\tilde H\| \ .
\end{align}
\end{enumerate}

Collecting all the terms, Eq.~\eqref{eq: perturbation in H intermediate} reduces to
\begin{align}
    \|\ket{\Psi_{k+1}}-\ket{\tilde{\Psi}_{k+1}}\| &\le (1+ 2|\theta_k|+4 |s_k| )\| \ket{\Psi_{k}}-  \ket{\tilde\Psi_{k}}\| +2|s_k|\cdot \|H-\tilde H\| \\
    &\le (1+6\zeta)\| \ket{\Psi_{k}}-  \ket{\tilde\Psi_{k}}\| + 2|s_k|\cdot \|H-\tilde H\| \ , 
\end{align}
where we use the definition $\zeta = \max(\bm{\theta},\bm{s})$ in the last line.
Finally, iterating this recursive bound and it yields
\begin{align}
    \|\ket{\Psi_{k+1}}-\ket{\tilde{\Psi}_{k+1}}\| &\le 
     \frac{|s_k|\cdot \|H-\tilde H\|}{3\zeta}(1+6\zeta)^{k+1}
    \le \frac{1}{3} \|H-\tilde H\|(1+6\zeta)^{k+1} \ , 
\end{align}
where we again used the definition $\zeta = \max(\bm{\theta},\bm{s})$ , i.e. $\displaystyle \frac{|s_k|}{\zeta} \leq 1$. Setting $K=k+1$, the proposition's statement is justified.
\end{proof}
 
\subsection{Perturbation of Angles}
\label{sec QSP perturbation} 
In order to study sensitivity of the parametrization in  Thm.~\ref{thm complex QSP} we define
\begin{align}
    \ket{\Psi_{H}(\bm{\theta},\bm{s})} = \prod_{k=0}^{K-1} e^{i \theta \Psi_{k}}  e^{s_{k}[\Psi_{k},H]}\ket{\Psi_0}\ .
\end{align}

In practice, we first measure the energy and variance, then compute and then compute $s_{k}$ and $\theta_{k}$ to implement the operation.
From this perspective, the time duration $s_{k}$ and phase $\theta_{k}$ vary at each step; that is, the perturbations satisfy $|s_{k}-\tilde{s}_{k}|\le \delta_{s}$ and $|\theta_{k}- \tilde{\theta}_{k}|\le \delta_{\theta}$.
In other words, even if the unitary implementation is perfect, the determined values for time duration and phase can cause errors.

Under this setting, we establish an error bound for implementing a non-unitary polynomial of degree $K$.
For simplicity, we assume the errors in time duration and phase remain constant across all steps.
In what follows, we denote the ideal state and operations as $\ket{\Psi_{k+1}}$ and $e^{i\theta_{k}\Psi_k}e^{s_{k}[\Psi_{k},H]}$, whereas erroneous counterparts are given by $e^{i\tilde{\theta}_{k}\tilde{\Psi}_k}e^{\tilde{s}_{k}[\tilde{\Psi}_{k},H]}$.
Finally, we note that no group commutator approximation is performed in this analysis.

\begin{propositionA}[QSP parametrization stability]
\label{angle stability prop app}
Let $H$ be a Hermitian matrix such that $\|H\|\le1$, and assume that the estimated parameters $\tilde s_k$ and $\tilde\theta_k$ satisfy $|s_k-\tilde s_k|\le \delta_s$ and $|\theta_k -\tilde \theta_k|\le \delta_\theta$ with ideal parameters $s_{k}$ and $\theta_{k}$ for all $k$.
By setting $\zeta = \max(\bm{\theta},\bm{s})$, the perturbed state $\ket{\tilde\Psi_H(\tilde{\bm{\theta}},\tilde{\bm{s}})}$ in Eq.~\eqref{eq:erroneous_state} and the state $\ket{\Psi_H (\bm{\theta},\bm{s})}$ from Thm.~\ref{thm complex QSP} satisfies

\begin{align} \label{app eq:error_accumulation_by_noisy_estimation}
    \|\ket{\Psi_{H}(\bm{\theta},\bm{s})} - \ket{\tilde\Psi_{H}(\tilde{\bm{\theta}},\tilde{\bm{s}})}\| \le \frac{1}{3\zeta}(1+6 \zeta)^K \max(\delta_s,\delta_\theta)\ .
\end{align}
\end{propositionA}
\begin{proof}
Let us define the intermediate QSP states
     \begin{align}
         \ket{\Psi_k}= \prod_{k'=0}^{k-1} e^{i \theta_{k'} \Psi_{k'}}  e^{s_{k'}[\Psi_{k'},H]}\ket{\Psi_0}
    \end{align}   
    and analogously $\ket{\tilde \Psi_k}$ are the intermediate states of QSP with $\tilde \theta_k$ and $\tilde s_k$.
    The difference between $\ket{\Psi_{k+1}}$ and $\ket{\tilde \Psi_{k+1}}$ is given by
    \begin{align}
        \|\ket{\Psi_{k+1}}-\ket{\tilde{\Psi}_{k+1}}\| &= \|e^{i\theta_{k}\Psi_k}e^{s_{k}[\Psi_{k},H]}\ket{\Psi_{k}}- e^{i{\tilde \theta}_{k}\tilde{\Psi}_k}e^{\tilde {s}_{k}[\tilde{\Psi}_{k},  H]}\ket{\tilde\Psi_{k}}\| \ .
    \end{align}
Again,  following the same procedure in Eq.~\eqref{eq: DB-QSP convergence intermediate} from Subsec.~\ref{app sec DBQSP}, we add and subtract the term $\{e^{i\theta_{k}\Psi_k}e^{s_{k}[\Psi_{k},H]}\ket{\tilde \Psi_{k}}, \;
e^{i \theta_{k} \tilde \Psi_k}e^{s_{k}[\Psi_{k},H]}\ket{\tilde \Psi_{k}}, \; 
e^{i \theta_{k}\tilde \Psi_k}e^{s_{k}[\tilde \Psi_{k},H]}\ket{\tilde \Psi_{k}}\},\; 
e^{i\tilde \theta_{k}\tilde \Psi_k}e^{s_{k}[\tilde \Psi_{k},H]}\ket{\tilde \Psi_{k}}
\}$ to split them into multiple norm calculation via triangle inequality.
Therefore, the result is 
\begin{align}
    \|\ket{\Psi_{k+1}}-\ket{\tilde{\Psi}_{k+1}}\| 
  &\le \| \ket{\Psi_{k}}-  \ket{\tilde\Psi_{k}}\|+
 \|e^{i\theta_{k}\Psi_k} - e^{i {\theta}_{k} \tilde\Psi_{k}}\|+\|e^{s_{k}[\Psi_{k},H]}- e^{ {s}_{k}[ \tilde {\Psi}_{k},  H]}\|
\nonumber\\
  &\quad + 
 \|e^{i\theta_{k}\tilde\Psi_k} - e^{i\tilde {\theta}_{k}\tilde\Psi_{k}}\|+\|e^{s_{k}[\tilde \Psi_{k},H]}\ket{\tilde \Psi_k} - e^{  \tilde {s}_{k}[\tilde{\Psi}_{k},  H]}\ket{\tilde \Psi_k} \| \\
  &\le \| \ket{\Psi_{k}}-  \ket{\tilde\Psi_{k}}\|\nonumber\\
  &\quad + |\theta_{k}|\cdot
 \|\Psi_k   -  \tilde\Psi_{k} \|\nonumber\\
   &\quad + |s_k|\cdot \| [\Psi_{k}-\tilde{\Psi}_{k},  H] \|  \nonumber\\
 & \quad+ |\theta_k-\tilde \theta_k| \nonumber\\
   &\quad +   \| \left(e^{{s}_k[\tilde\Psi,H]}-e^{\tilde s_k[\tilde\Psi,H]} \right)\ket{\tilde \Psi_k} \| \ , \label{eq: perturbation in angle intermediate}
\end{align}
where we use the formula $\|e^{A}-e^{B}\|\le\|A-B\|$ in the second inequality.
Next, we proceed to evaluate these terms separately.
 \begin{enumerate}
     \item For the first term $ \| \ket{\Psi_{k}}-  \ket{\tilde\Psi_{k}}\|$, it remains unchanged.

     \item For the second term, it becomes
     \begin{align}
          |\theta_{k}|\cdot
 \|\Psi_k   -  \tilde\Psi_{k} \| \leq 2 |\theta_{k}|\cdot \|\ket{\Psi_k}   -  \ket{\tilde\Psi_{k}} \| \ ,
     \end{align}
     where we use the relation  $ \|\Psi_k   -  \tilde\Psi_{k} \|\le 2 \| \ket{\Psi_{k}}-  \ket{\tilde\Psi_{k}}\|$.

     \item For the third term, it is
     \begin{align}
          |s_k|\cdot \| [\Psi_{k}-\tilde{\Psi}_{k},  H] \|  &\leq 2 |s_k|\cdot \| \Psi_{k}-\tilde{\Psi}_{k}\| \cdot \|H\|\\
          &\leq 4 |s_k|\cdot\| \ket{\Psi_{k}}-  \ket{\tilde\Psi_{k}}\| \ ,
     \end{align}
     where we use the bound $\|[A,B]\|\le 2\|A\|\|B\|$ in the first line and the relation  $ \|\Psi_k   -  \tilde\Psi_{k} \|\le 2 \| \ket{\Psi_{k}}-  \ket{\tilde\Psi_{k}}\|$.
     Note that we also exploited the assumption that  $\|H\|\le1$ in the second line.

     \item For the fourth term, we recall the definition of $\delta_\theta$, i.e. $|\theta_k -\tilde \theta_k|\le \delta_\theta$.

     \item For the fifth term, we observe that
     \begin{equation}
\begin{split}
    e^{\tilde{s}[\Psi,H]}\ket{\tilde\Psi_{k}} &= \left(\left(\frac{E_{\Omega}}{\sqrt{V_{\Omega}}}\sin(\tilde{s} \sqrt{V_{\Omega}})+ \cos(\tilde{s} \sqrt{V_{\Omega}})\right) I -\frac{1}{\sqrt{V_{\Omega}}}\sin(\tilde{s} \sqrt{V_{\Omega}}) H \right) \ket{\tilde\Psi_{k}}\\
    &= \cos(\delta_{s}\sqrt{V_{\Omega}}) \left(e^{s[\Psi,H]}\ket{\tilde\Psi_{k}}\right) + \sin(\delta_{s}\sqrt{V}_{\Omega}) \left(e^{(s+\pi/2\sqrt{V_{\Omega}})[\Psi,H]}\ket{\tilde\Psi_{k}}\right).
\end{split}
\end{equation}
Using this expression, the fifth term can be simplified to
     \begin{equation}
\begin{split}
    \|e^{\tilde{s}[\Psi,H]}\ket{\Psi}-e^{s[\Psi,H]}\ket{\Psi}\|&=\sqrt{2-2|\braket{\Psi|e^{\tilde{s}[\Psi,H]}e^{-s[\Psi,H]}|\Psi}|}\\
    &= \sqrt{2-2|\cos(\delta_{s}\sqrt{V_{\Omega}}) |}\\
    &=\sqrt{4-4|\cos^2(\delta_{s}\sqrt{V_{\Omega}/2})|} = 2|\sin(\delta_{s}\sqrt{V_{\Omega}}/2)| \le \delta_{s}\sqrt{V_{\Omega}},
\end{split}
\end{equation}
Using the fact that $\sqrt{V_{\Psi_{k}}}\le \|H\|$ and the assumption that $\|H\|\le 1$ , we have
\begin{align}
    \|e^{\tilde{s}[\Psi,H]}\ket{\Psi}-e^{s[\Psi,H]}\ket{\Psi}\| \leq \delta_s \ .
\end{align}
 \end{enumerate}
 
Collecting all the terms, Eq.~\eqref{eq: perturbation in angle intermediate} reduces to
\begin{align}
    \|\ket{\Psi_{k+1}}-\ket{\tilde{\Psi}_{k+1}}\| &\leq (1+2|\theta_{k}|+4|s_{k}|) \| \ket{\Psi_{k}}-  \ket{\tilde\Psi_{k}}\| + \delta_{\theta} + \delta_{s}\\
    &\le (1+ 6\zeta )\| \ket{\Psi_{k}}-  \ket{\tilde\Psi_{k}}\| +2\gamma. 
\end{align}
where we utilize the definition of $\gamma$ in the last line, i.e. $\gamma = \max(\delta_s,\delta_\theta)$.
Now, solving the iterative sequence, we get
\begin{equation}
\begin{split}
     \|\ket{\Psi_{k+1}}-\ket{\tilde{\Psi}_{k+1}}\| \le 2\gamma \sum_{i=0}^{k} (1+6\zeta)^{i} 
    &=2\gamma \frac{1-(1+6\zeta)^{k+1}}{1-(1+6\zeta)} \le \frac{\gamma}{3\zeta}(1+6\zeta)^{k+1}\ .
\end{split}
\end{equation}

 Setting $K=k+1$, the proposition statement is justified.
 
\end{proof}

 \subsection{Statistical Error Propagation}
\label{sec QSP Hamiltonian stability} 

In this section, we study sensitivity of the parametrization in  Thm.~\ref{thm complex QSP} to estimation errors of the energy and variance. 
More precisely, for $k=0,\ldots K-1$, we define the energy $\overline E_k$ and the variance $\overline V_k$ for states $\ket{\overline\Psi_k}$, which is recursively determined by  
\begin{align}
\label{app estimation state sequence}
    \ket{\overline \Psi_{k+1}} = e^{i \overline\theta_k \overline \Psi_{k}}  e^{\overline s_{k}[\overline \Psi_{k},H]}\ket{\overline \Psi_k}\ ,
\end{align}
with 
  \begin{equation}
\overline s_k = \frac{-1}{\sqrt{\overline V_k}}\arccos\left(\frac{\overline E_k-u}{\sqrt{\overline V_k +(\overline E_k-u)^2}}\right)
\end{equation}  
and 
\begin{align}
    \overline \theta_k = \text{arg}\left(\frac{\overline{E}_k-z_k}{|\overline{E}_k-z_k|}\right).
\end{align}
With this, the final state reads
\begin{align}
    \ket{\overline \Psi_{H}(\overline{\bm{\theta}},\overline{\bm{s}})} = \prod_{k=0}^{K-1} e^{i \overline\theta_k \overline \Psi_{k}}  e^{\overline s_{k}[\overline \Psi_{k},H]}\ket{\overline \Psi_0}\ .
\end{align}

Note that, while Prop.~\ref{angle stability prop app} characterizes the sensitivity to differences in parameters $s_{k}$ and $\theta_{k}$, its direct application to analyzing the impact of statistical estimates is non-trivial. 
To address this, we establish a lemma that circumvents this challenge by directly considering the relevant quantum states.
In the analysis, we define
    \begin{equation}
        s(E,V) = \frac{-1}{\sqrt{V}}\arccos\left(\frac{|E-z|}{\sqrt{V +|E-z|^2}}\right)
    \end{equation} 
    and
     \begin{align}
 \theta(E,V) = \text{arg}\left(\frac{ {E} -z }{| {E} -z |}\right)\ .
\end{align}
for any $E\in \mathbb R$ and $V\ge0$.

\begin{propositionA}[Statistical error propagation]
\label{app prop stat err propagation}
Suppose $H$ is a Hermitian matrix whose spectral radius does not exceed unity, i.e., $\|H\|\le1$.
Consider the linear polynomial $p(H)= H-zI$, which we implement using Thm.~\ref{thm complex QSP} for some $z\in \mathbb C$.
Let $\ket\Psi$ be a state with energy $E_\Psi$
 and variance $V_\Psi$.
Then, for any $E'\in \mathbb R$ and $V'\ge 0$, we have
\begin{align}
    \| \left(e^{i\theta(E_\Psi,V_\Psi) \Psi}e^{{s(E_\Psi,V_\Psi)}[ \Psi,H]}-e^{i\theta(E',V') \Psi }e^{s(E',V')[ \Psi,H]} \right)\ket{ \Psi} \| \le 20  \eta^4\max(|E_\Psi-E'|,|V_\Psi-V'|) \ ,
\end{align}
where $\displaystyle \eta = \max(\frac{1}{\sqrt{V_\Psi }},\frac{1}{\sqrt{V'}},\frac{1}{|E_\Psi-z|},\frac{1}{|E'-z|},1+|z|)$ is the maximal characteristic instability scale.
\end{propositionA}

\begin{proof}
For brevity, we define $\theta = \theta(E_\Psi,V_\Psi)$, and $\overline \theta= \theta(E',V')$, as well as $s = s(E,V)$ and $\overline s = s(E',V')$.

First we reduce the problem into two separate bounds
\begin{align}
\left\| e^{i\theta \Psi}e^{s[\Psi,H]}\ket{\Psi} - e^{i\overline \theta \Psi }e^{\overline s[\Psi,H]}\ket{\Psi }\right\|  
&\le     \left\| e^{i\theta \Psi}e^{s[\Psi,H]}\ket{\Psi} - e^{i\theta \Psi }e^{\overline s[\Psi,H]}\ket{\Psi }\right\| \nonumber\\
&\quad +     \left\| e^{i\theta \Psi}e^{\overline s[\Psi,H]}\ket{\Psi} - e^{i\overline \theta \Psi }e^{\overline s[\Psi,H]}\ket{\Psi }\right\|  \\
&=     \left\| e^{i\theta \Psi} - e^{i\overline \theta \Psi } \right\| + \left\|  e^{s[\Psi,H]}\ket{\Psi} -  e^{\overline s[\Psi,H]}\ket{\Psi }\right\|\ ,\label{eq theta bound}
\end{align}
where we use the unitary invariance in the second inequality.
Next, we evaluate these two terms individually.
\begin{enumerate}
    \item
    First, by utilizing the fact that $\Psi$ is a pure state, we have
\begin{align} 
    \left\| e^{i\overline \theta \Psi} - e^{i\theta \Psi } \right\| &= \|I+(e^{i\overline \theta}-1)\Psi -\left(I+(e^{i\theta}-1)\Psi\right) \| \\
    & = |e^{i\overline \theta}-e^{i\theta}| \\
    &= \left|\frac{E'-z}{|E'-z|} - \frac{E_\Psi-z}{|E_\Psi-z|}  \right|\\
    &\le |E'-z|\left|\frac{1}{|E'-z|} - \frac{1}{|E_\Psi-z|}  \right| + \frac{\left|E' - E_\Psi  \right|}{|E_\Psi-z|} \\
    & \le \eta \bigg||E_{\Psi}-z|-|E'-z|\bigg| + \eta |E' - E_\Psi| \\
    & \le 2\eta |E' - E_\Psi| \label{eq:phase_error_app},
\end{align}
where we utilize the triangle inequality for the forth line, while we use reverse triangle inequality in the last line.

    \item 
 The implementation with $\overline{s}$ results in $e^{\overline s [\Psi,H]}\ket{\Psi } =(a(\overline{s})I+b(\overline{s})H)\ket{\Psi}$ with
    \begin{align}
    a(\overline s) &= \frac{  E_\Psi }{\sqrt{  V_{\Psi}}}\sin\left(\overline s \sqrt{  V_{\Psi}}\right) + \cos\left(\overline s \sqrt{   V_{\Psi}}\right), \\
        b(\overline s) &= - \frac{1}{\sqrt{  V_{\Psi}}} \sin\left(\overline s \sqrt{  V_{\Psi}}\right)\ .
    \end{align}

    Recall that the equalities are derived in Lem.~\ref{thm zero overhead QSP unitary single step}; see Sec.~\ref{app lemma and thm proofs} for more details.
    We stress that $E_\Psi$ and $V_\Psi$ could be different from the estimated ones used for determining $\overline{s}$.
    This stands in contrast to implementing the polynomial which we wanted $e^{  s[\Psi,H]}\ket{\Psi }= (a(s)I +b(s)H)\ket\Psi$ with
    \begin{align}
    a(s) &= \frac{  E_\Psi }{\sqrt{  V_{\Psi}}}\sin\left(s \sqrt{  V_{\Psi}}\right) + \cos\left(s \sqrt{   V_{\Psi}}\right), \\
        b(s) &= - \frac{1}{\sqrt{  V_{\Psi}}} \sin\left(s \sqrt{  V_{\Psi}}\right)\ .
    \end{align}

With these expressions, we have
   \begin{align}
\left\| e^{\overline{s}[\Psi,H]}\ket{\Psi} - e^{ {s}[\Psi,H]}\ket{\Psi} \right\| 
&= \left\|  (a(\overline s)I +b(\overline s)H)\ket{\Psi } -  (a(s)I +b(s)H)\ket{\Psi } \right\|\\
&\le |a(s)-a(\overline s)| +  |b(s)-b(\overline s)|\,\|H\|\\
&\le |a(s)-a(\overline s)| +  |b(s)-b(\overline s)|\ . \label{eq: a diff and b diff}
\end{align}
In the last line, we used the spectral assumption $\|H\|\le 1$.
\begin{enumerate}
    \item We begin by bounding $|b(s)-b(\overline s)|$ because this will help with the bound for $a(s)$.
First, for ease of notation, we introduce $\alpha = \sqrt{V_\Psi} s = \arccos\left(\frac{|E_\Psi-z|}{\sqrt{V_\Psi +|E_\Psi-z|^2}}\right)$ to denote $b(s)=-\sin(\alpha)/\sqrt{V_{\Psi}}$.
Similarly, $b(\overline{s})$ with the estimated $\overline{s}$ is expressed using $\overline \alpha = \sqrt{V'} \overline s$ as
\begin{align}
        b(\overline s) &= - \frac{1}{\sqrt{  V_{\Psi}}} \sin\left( \sqrt{ V_{\Psi}/V'} \overline{\alpha}\right)\ .
\end{align}
Thus, we have
\begin{align}
        |b(s)-b(\overline s)| &=\frac 1 {\sqrt {V_\Psi}} |\sin( \sqrt{  V_{\Psi} / V'} \overline\alpha) - \sin(  \alpha)|
        \\
        &\le \frac 1 {\sqrt {V_\Psi}} |\sin( \sqrt{  V_{\Psi} / V'} \overline\alpha) - \sin(  \overline\alpha)| + \frac 1 {\sqrt {V_\Psi}} |\sin( \overline \alpha) - \sin(  \alpha)| \ .
\end{align}

\begin{enumerate}
    \item For the first part, using $\sin(a)-\sin(b) = 2\sin(\frac{a-b}2)\cos(\frac{a+b}2)$, we get
\begin{align}
     &\frac 1 {\sqrt {V_\Psi}} |\sin( \sqrt{  V_{\Psi} / V'} \overline\alpha) - \sin(  \overline\alpha)| 
\nonumber\\
&\le \frac 1 {\sqrt {V_\Psi}} \left|2\sin\left( (\sqrt{  V_{\Psi} / V'} -1) \overline\alpha/2\right) \cos \left( (\sqrt{  V_{\Psi} / V'} +1) \overline\alpha/2\right) \right|\\
&\le \frac \pi {2\sqrt {V_\Psi}} \left|\sqrt{  V_{\Psi} / V'} -1\right|  \\
&\le \frac{\pi}{2}\left| \frac 1 {\sqrt {V_\Psi}} -\frac 1 {\sqrt {V'}}\right| 
\\
&\le \frac{\pi}{2} \frac{|V_{\Psi}-V'|}{\sqrt{V_{\Psi}}\sqrt{V'}(\sqrt{V_{\Psi}}+\sqrt{V'})}\\
&\le \frac{\pi}{2} \frac{|V_{\Psi}-V'|}{V_{\Psi}\sqrt{V'}} \\
&\le 2\eta^3 |V_{\Psi}-V'|  ,
\end{align}
where we use $\cos(x)\le1$,  $\sin(x)\le x$ and $\alpha\le \pi/2$ in the second inequality.
As for the third equality, we utilize 
$$\frac{1}{\sqrt{A}}-\frac{1}{\sqrt{B}}=\frac{B-A}{\sqrt{AB}(\sqrt{A}+\sqrt{B})},$$
while we use $1/(x+y) \le 1/x$ for $x>0$ and $y>0$ in the fourth inequality.

\item For the second part, we notice that
\begin{align}
    \sin(\alpha) = \sqrt{1 - |E_\Psi-z|^2/{(V_\Psi +|E_\Psi-z|^2)}} = \sqrt{V_\Psi / {(V_\Psi +|E_\Psi-z|^2)}}.
\end{align}
with $\cos \alpha = \frac{|E_\Psi-z|}{\sqrt{V_\Psi +|E_\Psi-z|^2}}$.
Hence, we have
\begin{align}
    &\frac 1 {\sqrt {V_\Psi}} |\sin( \overline \alpha) - \sin(  \alpha)| \nonumber\\ & \le \frac 1 {\sqrt {V_\Psi}}\left|\sqrt{\frac{V_{\Psi}}{V_{\Psi}+|E_{\Psi}-z|^2}}-\sqrt{\frac{V'}{V'+|E'-z|^2}}\right| \\
    & \le \left|\frac{1}{\sqrt{V_{\Psi}+|E-z|^2}}-\frac{1}{\sqrt{V'+|E'-z|^2}}\right| +  \frac 1 {\sqrt {V_\Psi}\sqrt{V'+|E'-z|^2}}\left|\sqrt{V_{\Psi}}-\sqrt{V'}\right| \ .
\end{align}
The first component is further upper bounded by
\begin{align}
    &\left|\frac{1}{\sqrt{V_{\Psi}+|E_{\Psi}-z|^2}}-\frac{1}{\sqrt{V'+|E'-z|^2}}\right|\nonumber\\ &\le \frac{|V_{\Psi}+|E-z|^2-(V'+|E'-z|^2)|}{\sqrt{V_{\Psi}+|E-z|^2}\sqrt{V'+|E'-z|^2}(\sqrt{V_{\Psi}+|E-z|^2}+\sqrt{V'+|E'-z|^2})} \\
    &\le \frac{|V_{\Psi}-V'| + ||E_{\Psi}-z|^2-|{E'}-z|^2|}{|E_{\Psi}-z|^2|E'-z|} \\
    & \le \frac{|V_{\Psi}-V'|}{|E_{\Psi}-z|^2|E'-z|} + \frac{2(1+|z|)|E_{\Psi}-E'|}{|E_{\Psi}-z|^2|E'-z|} \label{eq: C93}\\
    & \le \eta^3 |V_{\Psi}-V'| + 2\eta^4 |E_{\Psi}-E'| \ ,
\end{align}
where we exploit $E_{\Psi}+E'\le2$, and $|E_{\Psi}-z|^2-|E'-z|^2 \le 2(1+|z|)|E_{\Psi}-E'|\le 2 \eta|E_{\Psi}-E'|$ because of the assumption $\|H\|\le1$.

Also, the second component is given by
\begin{align}
    \frac 1 {\sqrt {V_\Psi}\sqrt{V'+|E'-z|^2}}\left|\sqrt{V_{\Psi}}-\sqrt{V'}\right| &\le \frac{|V_{\Psi}-V'|}{\sqrt {V_\Psi}\sqrt{V'+|E'-z|^2}(\sqrt{V'}+\sqrt{V_{\Psi}})} \\
    & \le \frac{|V_{\Psi}-V'|}{V_\Psi|E'-z|} \\
    & \le \eta^3 |V_{\Psi}-V'|.
\end{align}
Consequently, we have 
\begin{align}
    \frac 1 {\sqrt {V_\Psi}} |\sin( \overline \alpha) - \sin(  \alpha)| & \le \eta^3 |V_{\Psi}-V'| + 2\eta^4 |E_{\Psi}-E'| + \eta^3 |V_{\Psi}-V'| \\
    &\le 2\eta^3 |V_{\Psi}-V'| + 2\eta^4 |E_{\Psi}-E'|.
\end{align}
\end{enumerate}

Therefore, the upper bound of $ |b(s)-b(\overline s)|$ is expressed as 
\begin{align}
    |b(s)-b(\overline s)| &\le 2\eta^3 |V_{\Psi}-V'| + 2\eta^4 |E_{\Psi}-E'| + 2\eta^3 |V_{\Psi}-V'| \\
    \quad&=4\eta^3 |V_{\Psi}-V'| + 2\eta^4 |E_{\Psi}-E'|. \label{eq: b diff}
\end{align}


\item  

Using similar procedure, we arrive at
\begin{align}
|a(s)-a(\overline s)| &\le E_\Psi |b(s)-b(\overline s)| + |\cos( \sqrt{  V_{\Psi} / V'} \alpha) - \cos(  \alpha)| + |\cos(   \overline  \alpha) - \cos(  \alpha)|\\
&\le |b(s)-b(\overline s)| + |2\sin( (\sqrt{  V_{\Psi} / V'} -1) \alpha/2) \sin( (\sqrt{  V_{\Psi} / V'} +1) \alpha/2)| \nonumber\\ &\quad
+ |\cos(   \overline  \alpha) - \cos(  \alpha)|\\
&\le |b(s)-b(\overline s)| + \frac{\pi}{2}\left| \frac 1 {\sqrt {V_\Psi}} -\frac 1 {\sqrt {V'}}\right|  |\cos(   \overline  \alpha) - \cos(  \alpha)|\ .
\end{align}
For the last component, recalling $\displaystyle \cos \alpha = \frac{|E_\Psi-z|}{\sqrt{V_\Psi +|E_\Psi-z|^2}}$, we have
\begin{align}
    &|\cos( \overline \alpha) - \cos(  \alpha)| \nonumber\\&\le { \left|\frac{|E_\Psi-z|} { \sqrt {V_\Psi +|E_\Psi-z|^2}} - \frac{|E'-z|} { \sqrt{V' +|E'-z|^2}}\right|}\\
    &\le |E_\Psi-z| \left|\frac{1} { \sqrt{V_\Psi +|E_\Psi-z|^2}} - \frac{1} { \sqrt{V' +|E'-z|^2}}\right| + \frac{\left||E_\Psi-z|-|E'-z|\right|}{\sqrt{V' +|E'-z|^2}} \\
    & \leq|E_\Psi-z| \left|\ \frac{|V_{\Psi}-V'|}{|E_{\Psi}-z|^2|E'-z|} + \frac{2(1+|z|)|E_{\Psi}-E'|}{|E_{\Psi}-z|^2|E'-z|}\right| +\frac{ |E_{\Psi}-E'|}{\sqrt{V'}}\\
    & \le \eta^2 |V_{\Psi}-V'| + 2\eta^3 |E_{\Psi}-E'| + \eta|E_{\Psi}-E'| \\
    &= \eta^2 |V_{\Psi}-V'| + (2\eta^3+\eta) |E_{\Psi}-E'| \ ,
\end{align}
where we use the result Eq.~\eqref{eq: C93} in the third inequality.
Hence, we obtain
\begin{align}
    |a(s)-a(\overline s)| & \le 4\eta^3 |V_{\Psi}-V'| + 2\eta^4 |E_{\Psi}-E'| + \frac{\pi}{2} \eta^3 |V_{\Psi}-V'| + \eta^2 |V_{\Psi}-V'| \nonumber\\ &\quad+ (2\eta^3+\eta) |E_{\Psi}-E'| \\
    &\le (6\eta^3 + \eta^2) |V_{\Psi}-V'|  + (2\eta^4 + 2\eta^3+\eta) |E_{\Psi}-E'|. \label{eq: a diff}
\end{align}
\end{enumerate}

By substituting Eq.~\eqref{eq: b diff} and Eq.~\eqref{eq: a diff} into Eq.~\eqref{eq: a diff and b diff}, we have
\begin{align}
    \left\| e^{\overline{s}[\Psi,H]}\ket{\Psi} - e^{ {s}[\Psi,H]}\ket{\Psi} \right\|   &\le (6\eta^3 + \eta^2) |V_{\Psi}-V'|  + (2\eta^4 + 2\eta^3+\eta) |E_{\Psi}-E'| + 4\eta^3 |V_{\Psi}-V'| \nonumber\\ &\quad+ 2\eta^4 |E_{\Psi}-E'|\\
    &= (10\eta^3 + \eta^2) |V_{\Psi}-V'|  + (4\eta^4 + 2\eta^3+\eta) |E_{\Psi}-E'| \\
    &\leq 18 \eta^4 \;\text{max}(|V_{\Psi}-V'|, |E_{\Psi}-E'|) \label{eq:angle_error_app} \ .
\end{align}
where we use $10\eta^3 + \eta^2 \le 11\eta^4 $ and $2\eta^3+\eta \le3 \eta^4$ in the last line since $\eta \ge1$ by definition.
Lastly, combining Eq.~\eqref{eq:phase_error_app} and Eq.~\eqref{eq:angle_error_app} leads to the constants claimed above.

\end{enumerate}

\end{proof}

Finally, leveraging the techniques developed thus far, we establish a bound on the deviation between the ideal state and the state affected by noisy estimations of energy and variance in terms of their statistical errors.

\begin{propositionA}[QSP estimation stability]
\label{eq app qsp estimatin stability}
Suppose $H$ is a Hermitian matrix whose spectral radius does not exceed unity, i.e., $\|H\|\le1$ and consider the polynomial $p(H)$ of degree $K$ with roots $z_{k}$ satisfying $|z_{k}|\le|z|$.
We denote the ideal parameter sequences as $\bm{\theta}=(\theta_0,\ldots,\theta_{K-1})$ and $\bm{s}=(|s_0|,\ldots,|s_{K-1}|)$, which yield the exact state $\ket{\Psi_{H}(\bm{\theta},\bm{s})}$.
Similarly, let $\bm{\overline\theta}=(\overline\theta_0,\ldots,\overline\theta_{K-1})$ and $\bm{\overline s}=(|\overline s_0|,\ldots,|\overline s_{K-1}|)$ be the parameters obtained from statistical estimates of the energy mean and variance, as used in the states defined in Eq.~\eqref{app estimation state sequence}.
Also,  we define the maximal instability scale across all states as $\displaystyle \eta = \max(\frac{1}{\sqrt{V_\Psi }},\frac{1}{\sqrt{V'}},\frac{1}{|E_\Psi-z|},\frac{1}{|E'-z|},1+|z|)$.
Then, we have
\begin{align}
    \|\ket{\Psi_{H}(\bm{\theta},\bm{s})} - \ket{\overline\Psi_{H}(\overline{\bm{\theta}},\overline{\bm{s}})}\| \le (14+120\eta^4)^{K} \cdot \max(\delta_V,\delta_E) = O(\max(\delta_V,\delta_E))\ .
\end{align}
where $\delta_E \ge |E_k -\overline E_k|$ and $\delta_V \ge |V_k -\overline V_k|$ is the statistical errors of energy and variance.
\end{propositionA}
\begin{proof}
Let us define the intermediate QSP states
     \begin{align}
         \ket{\Psi_k}= \prod_{k'=0}^{k-1} e^{i \theta_{k'} \Psi_{k'}}  e^{s_{k'}[\Psi_{k'},H]}\ket{\Psi_0}
    \end{align}   
    and analogously $\ket{\overline \Psi_k}$ will be the intermediate states of QSP with $\overline \theta_k$ and $\overline s_k$.
    Thus, the difference between $\ket{\Psi_{k+1}}$ and $\ket{\overline \Psi_{k+1}}$ is given by
    \begin{align}
        \|\ket{\Psi_{k+1}}-\ket{\overline{\Psi}_{k+1}}\| &= \|e^{i\theta_{k}\Psi_k}e^{s_{k}[\Psi_{k},H]}\ket{\Psi_{k}}- e^{i{\overline \theta}_{k}\overline{\Psi}_k}e^{\overline {s}_{k}[\overline{\Psi}_{k},  H]}\ket{\overline\Psi_{k}}\|
    \end{align}
Again,  following the same procedure in Eq.~\eqref{eq: DB-QSP convergence intermediate} from Sec.~\ref{app sec DBQSP}, we add and subtract the term $\{e^{i\theta_{k}\Psi_k}e^{s_{k}[\Psi_{k},H]}\ket{\overline \Psi_{k}}, \;
e^{i \theta_{k} \overline \Psi_k}e^{s_{k}[\Psi_{k},H]}\ket{\overline \Psi_{k}}, \; 
e^{i \theta_{k}\overline \Psi_k}e^{s_{k}[\overline \Psi_{k},H]}\ket{\overline \Psi_{k}}\},\; 
e^{i\overline \theta_{k}\overline \Psi_k}e^{s_{k}[\overline \Psi_{k},H]}\ket{\overline \Psi_{k}}
\}$ to split them into multiple norm calculation via triangle inequality.
Therefore, the result is 
\begin{align}
    \|\ket{\Psi_{k+1}}-\ket{\overline{\Psi}_{k+1}}\| 
  &\le \| \ket{\Psi_{k}}-  \ket{\overline\Psi_{k}}\|+
 \|e^{i\theta_{k}\Psi_k} - e^{i {\theta}_{k} \overline\Psi_{k}}\|+\|e^{s_{k}[\Psi_{k},H]}- e^{ {s}_{k}[ \overline {\Psi}_{k},  H]}\|
\nonumber\\
  &\quad + 
 \|e^{i\theta_{k}\overline\Psi_k} - e^{i\overline {\theta}_{k}\overline\Psi_{k}}\|+\|e^{s_{k}[\overline \Psi_{k},H]}\ket{\overline \Psi_k} - e^{  \overline {s}_{k}[\overline{\Psi}_{k},  H]}\ket{\overline \Psi_k} \| \\
  &\le \| \ket{\Psi_{k}}-  \ket{\overline\Psi_{k}}\|\nonumber\\
  &\quad + |\theta_{k}|\cdot
 \|\Psi_k   -  \overline\Psi_{k} \|\nonumber\\
   &\quad + |s_k|\cdot \| [\Psi_{k}-\overline{\Psi}_{k},  H] \|  \nonumber\\
 & \quad+ \|e^{i\theta_{k}\overline\Psi_k} - e^{i\overline {\theta}_{k}\overline\Psi_{k}}\| \nonumber\\
   &\quad +   \| \left(e^{{s}_k[\overline\Psi,H]}-e^{\overline s_k[\overline\Psi,H]} \right)\ket{\overline \Psi_k} \| 
\ . \label{eq: QSP stability intermediate}
\end{align}
\begin{enumerate}
\item For the second term, we recall that $ \|\Psi_k   -  \overline \Psi_k\|\le 2 \| \ket{\Psi_{k}}-  \ket{\overline \Psi_k}\|$. Thus, we get
\begin{align}
    |\theta_{k}|\cdot
 \|\Psi_k   -  \overline\Psi_{k} \| \leq 2\pi \cdot 2 \| \ket{\Psi_{k}}-  \ket{\overline \Psi_k}\| \leq 13\| \ket{\Psi_{k}}-  \ket{\overline \Psi_k}\| . \label{eq: QSP theta term}
\end{align}

\item For the third term, since $|s_k|\le \pi/2\sqrt{V_{k}}\le 2\eta$, we have
\begin{align}
    |s_k|\cdot \| [\Psi_{k}-\overline{\Psi}_{k},  H] \|  &\leq 2\eta \cdot 2 \|\Psi_k   -  \overline\Psi_{k} \|\cdot \|H\| \leq 8 \eta \; \| \ket{\Psi_{k}}-  \ket{\overline \Psi_k}\|  . \label{eq: QSP skterm}
\end{align}

    \item By a similar consideration as in Eq.~\eqref{eq:phase_error_app}, we arrive at
\begin{align}
    \|e^{i\theta_{k}\overline\Psi_k} - e^{i\overline {\theta}_{k}\overline\Psi_{k}}\| &\le \|e^{i\theta_{k}\overline\Psi_k} - e^{i\tilde{\theta}_{k}\overline\Psi_{k}}\| + \|e^{i\tilde{\theta}_{k}\overline\Psi_k} - e^{i\overline {\theta}_{k}\overline\Psi_{k}}\| 
    \\&\le 4\eta  \| \ket{\Psi_{k}}-  \ket{\overline\Psi_{k}}\| + 2\eta \delta_E \ , \label{eq: QSP theta diff}
\end{align}
where we used     $|E_k -\tilde E_k| \le 2\|\ket{\Psi_k}-\ket{\overline \Psi_k}\|$
to bound $|\theta_k-\tilde \theta_k|$. 

\item Here, most steps proceeded as in the Prop.~\ref{angle stability prop app}, but the last remaining term needs a separate treatment as we do not a-priori have a bound on $|s_k-\overline s_k|$.
To proceed with the bound, we use the exact expectation values $\tilde E_k$ and $\tilde V_k$ of the states $\ket{\overline\Psi_k}$  in Eq.~\eqref{app estimation state sequence} to introduce
  \begin{equation}
\tilde s_k = s(\tilde E_k,\tilde V_k) =  \frac 1{\sqrt{\tilde V_k}}\arccos\left(\frac{\tilde E_k-u}{\sqrt{\tilde V_k +(\tilde E_k-u)^2}}\right)\ .
\end{equation}  

Next, Prop.~\ref{app prop stat err propagation} is used twice.
Indeed, we use $s_{k}$, $\tilde{s}_{k}$ and $\overline{s}_{k}$ to denote different durations:
\begin{itemize}
    \item $s_k$: a time duration computed using the exact expectation values of the exact states $\ket{\Psi_k}$
    \item $\tilde s_k$: a time duration computed using exact expectation values for states $\ket{\overline\Psi_k}$
    \item $\overline s_k$: a time duration computed using the estimated expectation values of the states $\ket{\overline\Psi_k}$.
\end{itemize}

Then, we have
\begin{align}
    \left\| \left(e^{{s}_k[\overline\Psi,H]}-e^{\overline s_k[\overline\Psi,H]} \right)\ket{\overline \Psi_k} \right\| \le \left\| \left(e^{{s}_k[\overline\Psi,H]}-e^{\tilde s_k[\overline\Psi,H]} \right)\ket{\overline \Psi_k} \right\| + \left\| \left(e^{\tilde{s}_k[\overline\Psi,H]}-e^{\overline s_k[\overline\Psi,H]} \right)\ket{\overline \Psi_k} \right\| \ .
\end{align}

We remark that both terms are already computed in Eq.~\eqref{eq:angle_error_app} for the proof of Prop.~\ref{app prop stat err propagation}.
Hence, we can compute them as follows.
\begin{enumerate}
    \item The first term is equal to Eq.~\eqref{eq:angle_error_app} when $V'$ ($E'$) is replaced with $\tilde{V_{k}}$ ($\tilde{E}_{k}$), i.e.,
\begin{align}
     \left\| \left(e^{{s}_k[\overline\Psi,H]}-e^{\tilde s_k[\overline\Psi,H]} \right)\ket{\overline \Psi_k} \right\| &\le 18\eta^4 \cdot \left(\text{max}(|V_{k}-\tilde V_k|, |E_{k}-\tilde E_k|)\right) \ .
\end{align}
Thus, the remaining task is to account the differences of the exact expectation values and variances, which are given by
\begin{align}
    |E_k -\tilde E_k| &\le  2\|\ket{\Psi_k}-\ket{\overline \Psi_k}\|, \\
    |V_k -\tilde V_k| &\le | \bra{\Psi_k}H^2\ket{\Psi_k} - \bra{\overline \Psi_k}H^2\ket{\overline \Psi_k}|+|E_k^2 -\tilde E_k^2| \le 6\|\ket{\Psi_k}-\ket{\overline \Psi_k}\| .
\end{align}
Using this observation, we have
\begin{align}
     \left\| \left(e^{{s}_k[\overline\Psi,H]}-e^{\tilde s_k[\overline\Psi,H]} \right)\ket{\overline \Psi_k} \right\| &\le 18\eta^4 \cdot \left(6\|\ket{\Psi_k}-\ket{\overline \Psi_k}\|\right) =108 \eta^4 \|\ket{\Psi_k}-\ket{\overline \Psi_k}\|\ .
    \label{eq: QSP stability last term part 1}
\end{align}

\item  For the second term, the gaps in the energy and variance arise from the inaccurate estimation caused by statistical noise.
Thus we have
\begin{align}
     \left\| \left(e^{\tilde{s}_k[\overline\Psi,H]}-e^{\tilde s_k[\overline\Psi,H]} \right)\ket{\overline \Psi_k} \right\| &\le 18 \eta^4\max(\delta_V,\delta_E)\ .
     \label{eq: QSP stability last term part 2}
\end{align}

\end{enumerate}

\end{enumerate}

Overall, by substituting the above calculations into Eq.~\eqref{eq: QSP stability intermediate}, we have
\begin{align}
    \|\ket{\Psi_{k+1}}-\ket{\overline{\Psi}_{k+1}}\| 
&\leq (14+120\eta^4) \| \ket{\Psi_{k}}-  \ket{\overline\Psi_{k}}\| + 20\eta^4 \max(\delta_V,\delta_E)\ .
\end{align}

Thus, using these bounds inductively, we get
\begin{align}
    \| \ket{\Psi_{k}}-  \ket{\overline\Psi_{k}}\|&\le 20\eta^4 \max(\delta_V,\delta_E)\sum_{i=0}^{k-1} (14+120\eta^4)^{i} \\
    &=20\eta^4 \max(\delta_V,\delta_E) \frac{1-(14+120\eta^4)^{k}}{1-(14+120\eta^4)} \\
    &\le (14+120\eta^4)^{k}\max(\delta_V,\delta_E)\ .
\end{align}

By setting $k=K$, we get the claimed scaling.
\end{proof}

\section{Applications of DB-QSP} \label{sec:application_app}

In this section, we explore the potential applications of DB-QSP. 
As noted in the main text, a limitation of DB-QSP is that only low-degree polynomials is feasible.
Therefore, we first provide useful approximation techniques that can expand its applicability.
We then discuss applicability of DB-QSP for other tasks.

\subsection{Examples of Low-Degree Polynomial Approximations}
To assess the effectiveness of DB-QSP, it is crucial to understand what kind of functions can be realized using low-degree polynomials.
In QSP with post-selection, this issue is directly linked to the success probability of post-selection, as it depends on the degree of polynomials as well as an input state $\ket{\Psi_0}$. 
Importantly, several polynomial approximation techniques have been explored in the literature~\cite{martyn2021grand,gilyen2019quantum}.
To illustrate that DB-QSP can achieve $\epsilon$-precision while maintaining a logarithmically small polynomial degree, we present three representative examples below.

We first show an approximation of the sign function, which can be used for the ground-state preparation task~\cite{lin2020near}.

\begin{example}[Approximation of the sign function $\text{sgn}(x)$~\cite{lin2020near,martyn2021grand}]
Suppose $\delta>0$, $x\in\mathbb{R}$ and $\epsilon\in(0,1/2)$.
Given a degree $K=\mathcal{O}(\log(1/\epsilon)/\delta)$, there exists an odd polynomial $p(x)$, such that
\begin{itemize}
    \item for all $x\in[-2,2]$: $|p(x)|\le1$ and
    \item for all $x\in[-2,2]\backslash(-\delta,\delta)$: $|p(x)-\text{sgn}(x)|\le\epsilon$
\end{itemize}
where 
\begin{equation}
    \text{sgn}(x)=
  \begin{cases}
    1 & \text{if $x>0$,} \\
    -1                 & \text{if $x<0$,} \\
    0       & \text{if $x=0$.}
  \end{cases}
\end{equation}
\end{example}

As another polynomial function for filtering in the ground-state preparation task, we also demonstrate the approximation of trigonometric functions.
\begin{example}[Polynomial approximation of trigonometric functions by  Jacobi-Anger expansion~\cite{gilyen2019quantum,martyn2021grand}]\label{lem:jacobi-anger}
    Suppose $s\in\mathbb{R}$ and $\epsilon\in(0,\frac{1}{e})$.
    Given a degree $K=\lfloor\frac{1}{2}r\left(\frac{e|s|}{2},\frac{5}{4}\epsilon\right)\rfloor$, trigonometric functions can be approximated as follows;
    \begin{align}
        \|\cos(sx)-J_{0}(s)+2\sum_{l=1}^{K}(-1)^{l}J_{2l}(t)T_{2l}(x)\|_{[-1,1]}\le\epsilon,\\
        \|\sin(sx)-2\sum_{l=1}^{K}(-1)^{l}J_{2l+1}(t)T_{2l+1}(x)\|_{[-1,1]}\le\epsilon,
    \end{align}
    where $J_{m}(s)$ is the Bessel functions of the first kind and $T_{m}(x)$ is the Chebyshev polynomials of the first kind.
    Also, $r(t,\epsilon)$ is a function that asymptotically scales as
    \begin{equation}
        r(t,\epsilon)=\Theta\left(|t|+\frac{\log(1/\epsilon)}{\log(e+\frac{\log(1/\epsilon)}{|t|})}\right).
    \end{equation}
\end{example}
This indicates that  trigonometric functions can be approximated using a polynomials of degree $d=\lfloor\frac{1}{2}r\left(\frac{e|s|}{2},\frac{5}{4}\epsilon\right)\rfloor$ to achieve $\epsilon$-precision.

Next, we move to an polynomial approximation for matrix inversion used for, e.g., solving linear system of equations.

\begin{example}[Polynomial approximation of the inverse~\cite{childs2017quantum,gilyen2019quantum}]
Suppose $\kappa>1$, $\epsilon\in(0,\frac{1}{2})$ and $x\in[-1,1]\backslash(-\frac{1}{\kappa}),\frac{1}{\kappa}$.
Then
\begin{equation}
    f(x) = \frac{1-(1-x^2)^a}{x}
\end{equation}
is an odd function with $a=\lceil\kappa^2\log(\kappa/\epsilon) \rceil$ is $\epsilon$-close to the inverse $1/x$.
Given a degree $K=\lceil \sqrt{a\log(4a/\epsilon)}\rceil=\mathcal{O}(\kappa\log(\kappa/\epsilon))$, the odd real function
\begin{equation}
    g(x)=4\sum_{l=1}^{K}(-1)^{l}\left[\frac{\sum_{j=l+1}^{a}{\binom{2a}{a+j}}}{2^{2a}}\right]T_{2l+1}(x)
\end{equation}
is $\epsilon$-close to $f(x)$ on the domain $[-1,1]$.
\end{example}
This indicates that DB-QSP has the potential to efficiently perform matrix inversion in terms of the inverse precision $1/\epsilon$.
However, the circuit depth required for DB-QSP scales super exponentially with the condition number, a key factor in assessing the algorithm’s efficiency. 
Thus, this example also highlights a fundamental challenge for DB-QSP in certain computational tasks.

With these approximations, our approach could circumvent the exponential costs for some cases.
Furthermore, approximations for other functions have been provided, e.g., in Ref.~\cite{low2017quantum}.
This suggests that certain tasks benefiting from these polynomials can also be performed using DB-QSP.
Thus, DB-QSP remains practically viable in some cases.

\subsection{Derivation of DB-QITE Using DB-QSP}

As discussed in the main text, DB-QSP can be utilized to implement Imaginary-Time Evolution (ITE), which is a key technique for ground-state preparation.
The non-unitary operator in imaginary time evolution, $e^{-\tau H}$, can be approximated up to the first order as $p(H)=I-\tau H$.
Thus, Lem.~\ref{thm zero overhead QSP unitary single step} immediately suggests that DB-QSP can be used to approximate the imaginary time evolution.
Then, by using the group commutator
\begin{align}
    e^{s_0[\Psi,H]}  =  e^{i\sqrt{s_0}H}e^{i\sqrt{s_0}\Psi}    
    e^{-i\sqrt{s_0}H}
    e^{-i\sqrt{s_0}\Psi} +O(s_0^{3/2})\ 
\end{align}
and noticing that the last unitary has a trivial action on $\ket \Psi$, we arrive at the proposal in Ref.~\cite{gluza_DB_QITE_2024}
\begin{align}
    \ket{\omega_{k+1}} =  e^{i\sqrt{s_0}H}e^{i\sqrt{s_0}\omega_k}    
    e^{-i\sqrt{s_0}H}\ket{\omega_{k}}\ .
\end{align}
This is essentially the same as choosing $N=1$ in DB-QSP with $\theta_k=0$.

One subtlety is that DB-QSP suggests to use state-dependent scheduling $s_k$.
Ref.~\cite{gluza_DB_QITE_2024} proved that using a sufficiently small constant $s_0$ allows to converge to the ground state and in every step have a cooling rate matching imaginary time evolution.
The iterative use of Thm.~\ref{thm zero overhead QSP unitary single step} shows that this quantum algorithm can be devised based on QSP as the design approach.

\subsection{Hamiltonian Simulation}

Next, we move onto the Hamiltonian simulation task, where the goal is to implement the real-time evolution $e^{-itH}$.
To perform this task, a common assumption in QSP implementation is direct access to a subroutine which applies the input Hermitian matrix $H$ to an input state; see App.~\ref{app:overview_qsp}.
One approach to Hamiltonian simulation is to approximate the evolution operator using a Taylor series expansion 
\begin{align}
    p_\text{HS}(H) = \sum_{k=0}^K\frac{(-it)^n}{n!} H^n \ ,
\end{align}
which allows the Hamiltonian evolution to be approximated via polynomial transformations.

At first glance, this polynomial decomposition suggests that DB-QSP might also be applicable to Hamiltonian simulation.
However, DB-QSP is not designed for this task.
Since Alg.~\ref{alg DBQSP} assumes direct access to $e^{itH}$, using DB-QSP for Hamiltonian simulation would be vacuous.
This query model is also known as the Hamiltonian evolution model and has been widely used in tasks such as ground-state energy estimation with early fault-tolerant quantum computers~\cite{dong2022ground,zhang2022computing,LinQPE}.
From this perspective, the fact that DB-QSP does not target Hamiltonian simulation is not a limitation but rather an inherent feature of the query model.

\subsection{Evolution under a Polynomial Function of Hamiltonian}
Interestingly, though, DB-QSP can be used to effectively transform Hamiltonians, while working in the Hamiltonian evolution model.
In its simplest variant, we aim to implement the Hamiltonian simulation of $H^2$, the second power of the input matrix.
Strikingly, DB-QSP can be applied to this scenario by interpreting $p_\text{HS}(H^2)$ as a polynomial of doubled degree in the variable $H$, which leads to the factorization
\begin{align}
    p_\text{HS}(H^2) = a_{2K} \prod_{k=1}^{2K} (H - \sqrt{z_k} I)(H +\sqrt{z_k} I)
\end{align}
in contrast to the alternative formulation
\begin{align}
p_\text{HS}(H^2) = a_K \prod_{k=1}^{K} (H^2 - z_k I),
\end{align}
which treats $H^{2}$ as the primary variable.
More generally, for evolution under $e^{itg(H)}$ where $g$ is a polynomial, we observe that if $h = p \circ g$ is also a polynomial, then $p_\text{HS}(g(H))$ can be factorized accordingly, allowing us to proceed in an analogous manner.

The possibility to systematically use the Hamiltonian simulation $e^{-itH}$ to simulate $e^{-itH^2}$ is implied by classic results in Lie group theory~\cite{lloyd1995almost}, but an explicit construction of the type provided by DB-QSP is new to our knowledge. 
In particular, DB-QSP could provide a convergence rate and a circuit lower bound to a large class of instances of this classic question.
Finally, we remark that Thm.~\ref{thm complex QSP} is required in this case, because $p_\text{HS}(H) = I-iH-H^2/2$ has complex roots $z_\pm = \pm1-i$ for $K=2$ and $t=1$.

\subsection{Laurent Polynomials}
Another application is the Laurent polynomials, which include terms with negative powers, i.e.,
\begin{align}
    p_\text{L}(H) = \sum_{k=-K}^K a_k H^k\ .
\end{align}
While it is useful for QSP to consider ``polynomials" involving inverse powers, Thm.~\ref{thm complex QSP} does not directly provide unitary synthesis for this type.
Yet, assuming the matrix $H$ satisfies $\|I-H\|<1$, we can consider $H^{-K} \approx p_\text{INV}(H)^K$ and thus write
\begin{align}
    p_\text{DBL}(H) = p_\text{INV}(H)^K \sum_{k=0}^K a_{k-K} H^{k}\ .
\end{align}
This gives the approximation $p_\text{DBL}(H)\approx p_\text{L}(H)$ and provides an example how one can implement Laurent polynomials using DB-QITE.
The general case for Laurent polynomials with arbitrary Hermitian matrices is left for future work.

\section{Classically-Aided DB-QSP Synthesis } \label{sec: classical_computation_app}

Statistical error is unavoidable when estimating energy mean and variance on quantum hardware, because of the finite number of measurement shots. 
Error analysis in Eq.~\eqref{eq:error_accumulation_by_noisy_estimation} further suggests that this issue becomes more pronounced as the polynomial degree increases. 
This suggests the need for an approach to circumvent this challenge.
One potential solution is to leverage classical computation in initial steps.
Motivated by this, we explore conditions under which energy and variance can be efficiently computed using classical resources.

Assume that the initial state $\ket{\Psi_0}$ is expressed in a basis where only $m$ of its components are nonzero.
For instance, if the initial state is a tensor-product of zero states, i.e., $\ket{0}^{\otimes L}$ with $L$-qubits, then $m=1$.
Additionally, suppose the target Hermitian matrix is given by $H = \sum_{i=1}^J w_i P_i$ with Pauli operators $\{P_{i}\}$.
Using Eq.~\eqref{eq: aI+bH main_} together with the effect of state-dependent phase gate, the resultant state can always be written as
\begin{equation} \label{eq:aI+bH_state}
    \ket{\Psi_{k}} = \left(\prod_{j=1}^{k} (a'(s_{j})I+b(s_{j})H)\right) \ket{\Psi_{0}}
\end{equation}
with $a'(s_{j})\in\mathbb{C}$ and $b(s_{j})\in\mathbb{R}$.
Note that the coefficients $a'(s_{j})$ and $b(s_{j})$ are determined by the energy and variance of the state at $(k-1)$-th step. Our goal is to estimate the energy $O = H$ and the variance $O = (H - \braket{H})^2$.
By substituting Eq.~\eqref{eq:aI+bH_state} into the expectation value $\braket{\Psi_{k}|O|\Psi_{k}}$, we obtain 
\begin{align} 
    \braket{\Psi_{k}|H|\Psi_{k}} &= \sum_{l=1}^{2k+1} \xi_l \braket{\Psi_{0}|H^{l}|\Psi_{0}},\label{eq:energy_classcial}\\
    \braket{\Psi_{k}|(H - \braket{H})^2|\Psi_{k}} &= \sum_{l=0}^{2k+2} \xi'_l \braket{\Psi_{0}|H^{l}|\Psi_{0}}.\label{eq:variance_classcial}
\end{align}
where coefficients $\{\xi_l\}$ and $\{\xi'_l\}$ are determined by computing Eq.~\eqref{eq:aI+bH_state}. 
Implication of Eqs.~\eqref{eq:energy_classcial},~\eqref{eq:variance_classcial} is that, if we can compute $\braket{\Psi_{0}|H^{l}|\Psi_{0}}$ up to $l=2k+2$ classically, the energy and variance at $k$ step is tractable using classical computers.

With these criteria in mind, we analyze the conditions under which the classical computation of energy and variance is feasible.
First, the number of Pauli operators in $H^{2k+2}$ is at most $J^{2k+2}+1$.
Furthermore, since each Pauli operator has exactly one nonzero entry per row and column, the total number of nonzero elements that need to be stored scales as $\mathcal{O}(mJ^{2k+2})$.
To ensure classical tractability in terms of memory and computational cost, we require this scaling to remain within $\mathcal{O}(\text{poly}(n))$.
Thus, the condition on $k$ for energy and variance to be classically computable is given by
\begin{align}
\begin{split}
   & m^2J^{2k+2} \le \text{poly}(n) \\
   \Leftrightarrow\,& k \le \frac{\log\left(\text{poly}(n)\right)/2-\log(m)}{\log \left(J\right)} -1 \nonumber\\
  \Leftrightarrow\,& k = \mathcal{O}\left(\frac{\log\left(\text{poly}(n)\right)/2-\log(m)}{\log \left(J\right)}\right).
\end{split}
\end{align}
This indicates that, if $m,J=\mathcal{O}(1)$, we can classically compute up to $k=\mathcal{O}(\log(n))$.
On the other hand, computing only constant step $k$ is possible if $m,J=\mathcal{O}(\text{poly}(n))$
This clearly captures the classical difficulty: if the initial state contains many non-zero elements and the number of Pauli terms becomes prohibitively large, it becomes infeasible to compute the energy even for a single step. 
However, for instance, the Pauli terms for Ising models scale linearly in the number of qubits. 
Furthermore, some situations involve easy-to-prepare initial states like the tensor product of zero states, where  $m=1$.
Thus, this result suggests that a few steps of classical computation may be feasible in some cases. We also note that this estimation is straightforward, and advanced classical techniques could further improve the efficiency, which we will leave for future work.

\section{Unbiased Estimator of the Operator Variance for Hamiltonians } \label{app variance estimator}

In this section, we derive an unbiased estimator for the variance of an observable expressed as a weighted sum of Pauli operators.
We first describe the measurement procedure used to estimate the expectation values of individual Pauli operators and their products. 
Next, we construct a straightforward variance estimator and demonstrate its bias arising from finite-sample effects.
To address this, we derive a corrected formula that provides an unbiased estimate of the operator variance.

\subsection{Measurement Procedure}

Here, we focus on the observables $\hat{O}$ that can be decomposed as the weighted sum of Pauli operators, i.e.,
the observable can be expressed in the form 
\begin{align}
\hat{O}=\sum_{i=1}^{L} w_i P_i \ ,    
\end{align}
where $P_i\in\{I,X,Y,Z\}^{\otimes n}$ denotes the Pauli operators for $n$ qubits and $w_i$ represents its corresponding weights.

Then, the variance of the observable $\hat{O}$ is defined as 
$V = \braket{\hat{O}^2} - \braket{\hat{O}}^2$, where $\braket{\cdot }= \braket{\psi | \cdot | \psi}$ for a pure quantum state $\ket{\psi}$.
Thus, the square of the observable is given by
\begin{align}
  \hat{O}^2 = \left(\sum_{i=1}^{L} w_i P_i\right)  \left(\sum_{j=1}^{L} w_j P_j\right)  
  =\sum_{i=1}^{L} w_i^2 I + \sum_{\substack{i,j=1 \\ i \neq j}}^L w_i w_j P_{ij} \ ,
\end{align}
where we use the Pauli operators's identity $P_i^2=I$ and we introduce the notation $P_{ij}=P_i P_j$.
Using the expression of $\hat{O}$ and $\hat{O}^2$, the variance of the observable is given by
\begin{align}\label{eq: variance operator}
    V = \braket{\hat{O}^2} - \braket{\hat{O}}^2  &= \sum_{i=1}^{L} w_i^2 \braket{I} + \sum_{\substack{i,j=1 \\ i \neq j}}^L w_i w_j \braket{P_{ij}} - \left(\sum_{i=1}^{L} w_i \braket{P_i}\right)^2 \ .
\end{align}
To estimate the variance, we measure each Pauli component multiple times.
Using the measurement outcomes, we can then estimate the expectation value of each term in Eq.~\eqref{eq: variance operator}.
The procedure is:
\begin{enumerate}
    \item $\overline{\braket{I}}=1$ by the assumption of normalised state.
    \item 
    Suppose we measure $P_i$ a total of $N_i$ times, yielding outcomes (a set of measured bit-strings)
\begin{align}
    \{b^{(P_i)}_1,\, b^{(P_i)}_2,\, \dots,\, b^{(P_i)}_{N_i}\} \quad \text{with each} \;\; b_k \in \{-1, +1\} \;\; \text{ for } 1\leq k \leq N_i \ .
\end{align}
Then the estimation of a single Pauli operator $P_i$ is $\displaystyle       \overline{\langle P_i\rangle} = \frac{1}{N_i}\sum_{k=1}^{N_i} b^{(P_i)}_k$.
    \item Similarly, if we measure the product operator $P_{ij}=P_iP_j$ (for $i\neq j$) $N_{ij}$ times, we obtain the estimator
\begin{align}
    \overline{\langle P_iP_j\rangle} = \frac{1}{N_{ij}}\sum_{k=1}^{N_{ij}} b^{(P_{ij})}_k.
\end{align}
\end{enumerate}

\subsection{Biased and Unbiased Estimator of the Operator Variance}\label{sec: unbiased estimator}

\subsubsection{Biased and Unbiased Estimator}
First, we mention the definition of biased and unbiased estimator.
\begin{definitionA}[Estimator]
Let $\overline{A}$ be an estimator of the parameter $A$. The estimator is said to be:
\begin{itemize}
    \item \textbf{unbiased} if $\mathbb{E}[\overline{A}] = A$, 
    \item \textbf{biased} if $\mathbb{E}[\overline{A}] \neq A$ \ ,
\end{itemize}
where we use $\mathbb{E}$ to denote the expected value over the sampling process in this section.
\end{definitionA}

Using the statistics, a natural choice of the estimator of Eq.~\eqref{eq: variance operator} would be
\begin{align}
    \tilde{V}&= \sum_{i=1}^{L} w_i^2 + \sum_{\substack{i,j=1 \\ i\neq j}}^{L} w_iw_j\,\overline{\langle P_iP_j\rangle} - \left(\sum_{i=1}^{L} w_i\,\overline{\langle P_i\rangle}\right)^2 \\
    &=\sum_{i=1}^{L} w_i^2 + \sum_{\substack{i,j=1 \\ i\neq j}}^{L} w_iw_j \left(\frac{1}{N_{ij}}\sum_{k=1}^{N_{ij}} b^{(P_{ij})}_k\right) - \left[\sum_{i=1}^{L} w_i\left(\frac{1}{N_i}\sum_{k=1}^{N_i} b^{(P_i)}_k\right)\right]^2 \ .\label{eq: variance operator estimation}
\end{align}

However, we demonstrate a simple example where the last term of this estimator introduces bias into the estimator.

\begin{example}\label{example: biased estimator}
Let $P\in\{I,X,Y,Z\}^{\otimes n}$ be a Pauli operator and suppose that we estimate its expectation value by performing $N$ independent and identically distributed (i.i.d.) measurements, yielding outcomes $\{b_i\}_{i=1}^N$ (with $b_i\in\{-1,+1\}$).
Using our construction, the natural estimator for $\langle P \rangle^2$ is $\tilde{\langle P \rangle^2} \coloneqq \left( \frac{1}{N}\sum_{i=1}^{N} b_i \right)^2$, and its expectation value yields
\begin{align}
   \mathbb{E}\!\left[\left( \frac{1}{N}\sum_{i=1}^{N} b_i \right)^2\right]  = \frac{1}{N^2}\left( \sum_{i=1}^{N} \mathbb{E}[b_i^2] + \sum_{\substack{i,j=1 \\ i\neq j}}^{N} \mathbb{E}[b_i b_j] \right)
   &= \frac{1}{N^2}\left( \sum_{i=1}^{N} 1 + \sum_{\substack{i,j=1 \\ i\neq j}}^{N} \langle P_i\rangle^2 \right)   \\
    &= \frac{1}{N^2}\left( N + N(N-1) \langle P \rangle^2 \right) = \langle P \rangle^2 + \frac{1 - \langle P \rangle^2}{N} \ .\label{eq: biased estimator of square} 
\end{align}
where we use the i.i.d. assumption  $( \mathbb{E}[b_i b_j] = \mathbb{E}[b_i] \, \mathbb{E}[b_j] $ for $ i\neq j )$ and the relation $ b_i^2 = 1 $ in the second line.
Clearly, $\mathbb{E} \left[\overbrace{\langle P \rangle^2} \right] \neq \langle P \rangle^2$ for any finite sample size $N$, and hence it is a biased estimator by the definition.
To remove this bias, we introduce a correction factor and define the unbiased estimator as
\begin{align}\label{eq:corrected_estimator}
\overline{\langle P \rangle^2} \coloneqq \frac{N}{N-1}\left[\left(\frac{1}{N}\sum_{i=1}^{N} b_i\right)^2 - \frac{1}{N}\right] \,.
\end{align}
Directly evaluating the expectation value of new estimator of $\overline{ \langle P \rangle^2}$ yields
    \begin{align}
    \mathbb{E}\left[ \frac{N}{N-1}\left[\left(\frac{1}{N}\sum_{i=1}^{N} b_i\right)^2 - \frac{1}{N}\right] \right] 
    = \frac{N}{N-1}\left[\langle P\rangle^2 + \frac{1-\langle P\rangle^2}{N} - \frac{1}{N}\right] = \langle P\rangle^2 \ .
\end{align}
Thus, this is indeed an unbiased estimator for $\langle P\rangle^2$ as the expectation value of the estimator is consistent with the true value.
\end{example}

\begin{propositionA}[Unbiased estimator for the variance of an observable]\label{prop: unbiased estimator}
Consider an observable $\hat{O}$ which can be written as $\hat{O}=\sum_{i=1}^{L} w_i P_i $, where $P_i\in\{I,X,Y,Z\}^{\otimes n}$ denotes the Pauli operators for $n$ qubits and $w_i$ represents its corresponding weights.
The unbiased estimator for the variance of this observable is then given by
      \begin{align}
    \overline{V}
    &=\sum_{i=1}^{L} w_i^2 + \sum_{\substack{i,j=1 \\ i\neq j}}^{L} w_iw_j \left(\frac{1}{N_{ij}}\sum_{k=1}^{N_{ij}} b^{(P_{ij})}_k\right) \nonumber \\
    &\qquad - \sum_{i=1}^{L} \frac{w_i^2 N_i}{N_i-1} \left[\left(\frac{1}{N_i}\sum_{k=1}^{N_i} b^{(P_i)}_k\right)^2-\frac{1}{N_i} \right] - \sum_{\substack{i,j=1 \\ i \neq j}}^L w_i w_j \left(\frac{1}{N_i}\sum_{k=1}^{N_i} b^{(P_i)}_k\right)\left(\frac{1}{N_j}\sum_{k=1}^{N_j} b^{(P_j)}_k\right)\ ,
\end{align}
where we measure the operator $P_i$ a total of $N_i$ times and the product operator $P_{ij}=P_iP_j$ (for $i\neq j$) a total of $N_{ij}$ times.
\end{propositionA}
\begin{proof}
    The expected value of $\overline{V}$ is 
    \begin{align}
        \mathbb{E}\left[\overline{V}\right]&=\sum_{i=1}^{L} w_i^2 + \sum_{\substack{i,j=1 \\ i\neq j}}^{L} w_iw_j  \;\mathbb{E}\left[\left(\frac{1}{N_{ij}}\sum_{k=1}^{N_{ij}} b^{(P_{ij})}_k\right)\right] \nonumber \\
         &\qquad - \sum_{i=1}^{L} \frac{w_i^2 N_i}{N_i-1} \mathbb{E}\left[\left(\frac{1}{N_i}\sum_{k=1}^{N_i} b^{(P_i)}_k\right)^2-\frac{1}{N_i} \right] - \sum_{\substack{i,j=1 \\ i \neq j}}^L w_i w_j \;\mathbb{E}\left[\left(\frac{1}{N_i}\sum_{k=1}^{N_i} b^{(P_i)}_k\right)\left(\frac{1}{N_j}\sum_{k=1}^{N_j} b^{(P_j)}_k\right)\right] \ ,
    \end{align}
where we use the fact that the expected value of a constant is the same constant for the first term.
Next, we address the remaining terms separately.
\begin{enumerate}
    \item Since the term $\displaystyle\frac{1}{N_{ij}}\sum_{k=1}^{N_{ij}} b^{(P_{ij})}_k$ is unbiased estimator for $\braket{P_{ij}}$, we have
    \begin{align}
        \sum_{\substack{i,j=1 \\ i\neq j}}^{L} w_iw_j  \;\mathbb{E}\left[\left(\frac{1}{N_{ij}}\sum_{k=1}^{N_{ij}} b^{(P_{ij})}_k\right)\right]&=\sum_{\substack{i,j=1 \\ i\neq j}}^{L} w_iw_j \braket{P_{ij}} \ .
    \end{align}
    \item For the third term, we have
    \begin{align}
          \sum_{i=1}^{L} \frac{w_i^2 N_i}{N_i-1} \mathbb{E}\left[\left(\frac{1}{N_i}\sum_{k=1}^{N_i} b^{(P_i)}_k\right)^2-\frac{1}{N_i} \right]
          &=  \sum_{i=1}^{L} \frac{w_i^2 N_i}{N_i-1}  \left\{\mathbb{E}\left[\left(\frac{1}{N_i}\sum_{k=1}^{N_i} b^{(P_i)}_k\right)^2 \right] -\frac{1}{N_i}\right\}\\
         &=\sum_{i=1}^{L} \frac{w_i^2 N_i}{N_i-1}  \left\{\langle P_i \rangle^2 + \frac{1 - \langle P_i \rangle^2}{N_i} -\frac{1}{N_i}\right\}=\sum_{i=1}^{L} w_i^2 \langle P_i \rangle^2 \ ,
    \end{align}
    where we use Eq.~\eqref{eq: 2nd moment of X_i} of Lem.~\ref{lemma: third and fourth moments of X_i} in the second line.
    \item For the last term, since the samples for different indices $(i\neq j)$ are i.i.d., we have
    \begin{align}
        \sum_{\substack{i,j=1 \\ i \neq j}}^L w_i w_j \;\mathbb{E}\left[\left(\frac{1}{N_i}\sum_{k=1}^{N_i} b^{(P_i)}_k\right)\left(\frac{1}{N_j}\sum_{k=1}^{N_j} b^{(P_j)}_k\right)\right]&=\sum_{\substack{i,j=1 \\ i \neq j}}^L w_i w_j \;\mathbb{E}\left[\left(\frac{1}{N_i}\sum_{k=1}^{N_i} b^{(P_i)}_k\right)\right]\mathbb{E}\left[\left(\frac{1}{N_j}\sum_{k=1}^{N_j} b^{(P_j)}_k\right)\right]\\
        &=\sum_{\substack{i,j=1 \\ i \neq j}}^L w_i w_j  \braket{P_i}\braket{P_j} \ .
    \end{align}
\end{enumerate}
Collecting all the terms, the expected value of $\overline{V}$ becomes
\begin{align}
    \mathbb{E}\left[\overline{V}\right]&=\sum_{i=1}^{L} w_i^2 
+ \sum_{\substack{i,j=1 \\ i\neq j}}^{L} w_iw_j\,\langle P_iP_j\rangle 
- \sum_{i=1}^{L} w_i^2\,\langle P_i \rangle^2 
- \sum_{\substack{i,j=1 \\ i \neq j}}^L w_iw_j\,\langle P_i\rangle\,\langle P_j\rangle\\
&=\sum_{i=1}^{L} w_i^2 
+ \sum_{\substack{i,j=1 \\ i\neq j}}^{L} w_iw_j\,\langle P_{ij}\rangle 
- \left(\sum_{i=1}^{L} w_i\,\langle P_i\rangle\right)^2 
\end{align}
Since $ \mathbb{E}\left[\overline{V}\right]=V$ by using Eq.~\eqref{eq: variance operator}, $\overline{V}$ is indeed an unbiased estimator.
\end{proof}

\subsection{Total Variance of the Unbiased Estimator of the Operator Variance}

Next, with the motivation to assess the query complexity in DB-QSP, we compute the uncertainty of the unbiased estimator of the operator variance. 
Here, we consider variance as the uncertainty metric, which is given by:
\begin{align}\label{eq: var of the unbiased estimator}
    \text{Var}[\overline{V}] = \text{Var}\left[ \overline{\braket{\hat{O}^2}} - \overline{\braket{\hat{O}}^2}\right] =  \text{Var}\left[\overline{\braket{\hat{O}^2}}\right] + \text{Var}\left[\overline{\braket{\hat{O}}^2}\right] - 2 \; \text{Cov} \left[\overline{\braket{\hat{O}^2}},\overline{\braket{\hat{O}}^2}\right],
\end{align}
where we use the identity $\text{Var}[A+B]=\text{Var}[A]+\text{Var}[B]+2\text{Cov}(A,B)$.
The estimator for $\overline{\braket{\hat{O}^2}}$ and $\overline{\braket{\hat{O}}^2}$ are determined by the measurement on the Pauli operators $P_{ij}$ and $P_k$ respectively.
Assuming that the measurements on $P_{ij}$ and $P_k$ are independent, their covariance is zero, i.e., $\text{Cov} \left[\overline{\braket{\hat{O}^2}},\overline{\braket{\hat{O}}^2}\right]=0$.
Therefore, the remaining task is to evaluate $\text{Var}\left[\overline{\braket{\hat{O}^2}}\right]$ and $\text{Var}\left[\overline{\braket{\hat{O}}^2}\right]$.
We address these two terms in Lem.~\ref{lemma: first variance term} and Lem.~\ref{lemma: second variance term}.

\begin{lemmaA}\label{lemma: first variance term}
Suppose we have an observable $\hat{O}$ which is of the form $\hat{O}=\sum_{i=1}^{L} w_i P_i$, where $P_i\in\{I,X,Y,Z\}^{\otimes n}$ denotes the Pauli operators for $n$ qubits and $w_i$ represents its corresponding weights. 
Assuming that measurements performed for all operators are i.i.d., then the uncertainty (variance) of the estimation of $\hat{O}^2$ can be expressed as
\begin{align}\label{eq:lem7_statement}
\text{Var}\left[\overline{\langle \hat{O}^2 \rangle}\right] 
 = \sum_{\substack{i,j=1 \\ i\neq j}}^{L} \frac{w_i^2w_j^2}{N_{ij}} \left(1-\braket{P_{ij}}^2 \right) \ ,
\end{align}
where we define $\braket{P_{ij}}=\braket{P_{i}}\braket{P_{j}}$ and $N_{ij}$ is the number of sample used to estimate $\braket{P_{ij}}$.
\end{lemmaA}
\begin{proof}
 We start with the expression
\begin{align}
\text{Var}\left[\overline{\langle \hat{O}^2 \rangle}\right] 
&= \text{Var}\left[\sum_{i=1}^{L} w_i^2 + \sum_{\substack{i,j=1 \\ i\neq j}}^{L} w_iw_j \left(\frac{1}{N_{ij}}\sum_{k=1}^{N_{ij}} b^{(P_{ij})}_k\right)\right].
\end{align}

Since the first term $\sum_{i=1}^{L} w_i^2$ is a constant, its variance is zero. Therefore, we have
\begin{align}
\text{Var}\left[\overline{\langle \hat{O}^2 \rangle}\right] 
&= \text{Var}\left[\sum_{\substack{i,j=1 \\ i\neq j}}^{L} w_iw_j \left(\frac{1}{N_{ij}}\sum_{k=1}^{N_{ij}} b^{(P_{ij})}_k\right)\right].
\end{align}

Assuming that the contributions from different pairs of Pauli operators $(i,j)$ are independent, it reduces to
\begin{align}
\text{Var}\left[\overline{\langle \hat{O}^2 \rangle}\right] 
&= \sum_{\substack{i,j=1 \\ i\neq j}}^{L} \text{Var}\left[\frac{w_i w_j}{N_{ij}} \sum_{k=1}^{N_{ij}} b^{(P_{ij})}\right] = \sum_{\substack{i,j=1 \\ i\neq j}}^{L} \frac{w_i^2w_j^2}{N_{ij}^2} \text{Var}\left[\sum_{i=1}^{N_{ij}} b_{i}^{(P_{ij})}\right] \ ,
\end{align}
where we use the property $\text{Var}[aX] = a^2\,\text{Var}[X]$ (for any nonnegative constant $a$) in the last line.
Furthermore, assuming that the $b_{i}^{(P_{ij})}$ are i.i.d., it can be further simplified to
\begin{align}
\text{Var}\left[\overline{\langle \hat{O}^2 \rangle}\right] 
&=\sum_{\substack{i,j=1 \\ i\neq j}}^{L} \frac{w_i^2w_j^2}{N_{ij}^2}  \left( N_{ij}\,\text{Var}\left[b_{i}^{(P_{ij})} \right]\right) \ . \label{eq: variance 1st term}
\end{align}
Next, recall that each $ b_i $ satisfies $ b_i^2 = 1 $, and hence we obtain the following relation
\begin{align}\label{eq: variance identity}
    \text{Var}\left[b_{i}^{(P_{ij})} \right] = \braket{b_{i}^{(P_{ij})}\times b_{i}^{(P_{ij})} } - \braket{b_{i}^{(P_{ij})}}^2 = 1- \braket{b_{i}^{(P_{ij})}}^2 = 1- \braket{P_{ij}}^2 \ .
\end{align}
Combining Eq.~\eqref{eq: variance 1st term} and Eq.~\eqref{eq: variance identity} yields Eq.~\eqref{eq:lem7_statement}.
\end{proof}

Before proceeding to Lem.~\ref{lemma: second variance term}, let us first show a technical Lem.~\ref{lemma: third and fourth moments of X_i}, which will be useful in the proof of Lem.~\ref{lemma: second variance term}.
\begin{lemmaA}\label{lemma: third and fourth moments of X_i}
    Given a single Pauli operator of $n$ qubits $P_i\in\{I,X,Y,Z\}^{\otimes n}$, we estimate the expectation value of $P_i$ as
\begin{align}
   X_i= \overline{\langle P_i\rangle} = \frac{1}{N_i}\sum_{k=1}^{N_i} b^{(P_i)}_k \ .
\end{align}
where $b^{(P_i)}_k$ denotes the outcome of $k$-th measurement for the Pauli operator $P_i$.
Assuming the measurements are i.i.d., we obtain
\begin{enumerate}
\item the first moment of $X_i$ as $\mathbb{E}\left[ X_i\right]=  \overline{\langle P_i\rangle}$.
\item the second moment of $X_i$ as
\begin{align}\label{eq: 2nd moment of X_i}
    \mathbb{E}\left[ X_i^2\right]=   \langle P_i \rangle^2 + \frac{1 - \langle P_i \rangle^2}{N_i} \ .
\end{align}
    \item  the third moment of $X_i$ as
    \begin{align}
         \mathbb{E}\left[ X_i^3\right] 
         &=\langle P_i\rangle^3\left( 1 - \frac{3}{N_i} + \frac{2}{N_i^2}\right)
+\langle P_i\rangle\left(\frac{3}{N_i} - \frac{2}{N_i^2}\right)  \ .
        \end{align}
    \item   the fourth moment of $X_i$ as     
    \begin{align}
             \mathbb{E}\left[ X_i^4\right] 
             &=\langle P_i\rangle^4+\frac{6\langle P_i\rangle^2\left(1-\langle P_i\rangle^2\right)}{N_i}
            +\frac{\left(11\langle P_i\rangle^2-3\right)\left(\langle P_i\rangle^2-1\right)}{N_i^2}+\frac{2(3\langle P_i\rangle^2-1)(1-\langle P_i\rangle^2)}{N_i^3} \ .
        \end{align}
        where $\langle P_i\rangle$ denotes the true expectation value of $P_i$.
\end{enumerate}
\end{lemmaA}
\begin{proof}
 To simplify the notation, we use $b_i$ to represent $b^{(P_i)}_i$ throughout this proof.
    \begin{enumerate}
    \item \underline{\textbf{\textit{First moment of $X_i$:}}}
    
    Since the expectation value of a scalar is just the scalar itself, we obtain $\mathbb{E}\left[ X_i\right]=  \mathbb{E}\left[ \overline{\langle P_i\rangle}\right]= \overline{\langle P_i\rangle}$.
    \item \underline{\textbf{\textit{Second moment of $X_i$:}}}

Taking the expectation value of it yields
\begin{align}
   \mathbb{E}\!\left[X_i^2 \right] =   \mathbb{E}\!\left[\left( \frac{1}{N_i}\sum_{i=1}^{N_i} b_i \right)^2\right]  = \frac{1}{N_i^2}\left( \sum_{i=1}^{N_i} \mathbb{E}[b_i^2] + \sum_{\substack{i,j=1 \\ i\neq j}}^{N_i} \mathbb{E}[b_i b_j] \right) \ .
\end{align}
Since each $ b_i $ satisfies $ b_i^2 = 1 $, we have $ \mathbb{E}[b_i^2] = 1 $ for all $ i $.
Furthermore, assuming the samples are i.i.d., we obtain $ \mathbb{E}[b_i b_j] = \mathbb{E}[b_i] \, \mathbb{E}[b_j] $ for $ i\neq j $. 
Thus, it becomes
\begin{align}
    \mathbb{E}\!\left[ X_i^2 \right]= \frac{1}{N_i^2}\left( \sum_{i=1}^{N_i} 1 + \sum_{\substack{i,j=1 \\ i\neq j}}^{N_i} \mathbb{E}[b_i]^2 \right) 
    = \frac{1}{N_i^2}\left( \sum_{i=1}^{N_i} 1 + \sum_{\substack{i,j=1 \\ i\neq j}}^{N_i} \langle P_i\rangle^2 \right)   &= \frac{1}{N_i^2}\left( N_i + N_i(N_i-1) \langle P_i \rangle^2 \right) \\
     &= \langle P_i \rangle^2 + \frac{1 - \langle P_i \rangle^2}{N_i} \ .
\end{align}
where we recall the definition of the true expectation value  $ \mathbb{E}[b_i]=\langle P_i \rangle$ in the second equality.
        \item \underline{\textbf{\textit{Third moment of $X_i$:}}} 
        
        Similarly, for the third moment, we split the summation into multiple parts, i.e.  we classify the 3--tuple $(i,j,k)$ according to the ``equivalence class'' of the three indices.
        Thus, we have
        \begin{align}
            \mathbb{E}\left[ X_i^3\right] &= \frac{1}{N_i^3} \left(\sum_{i=1}^{N_i} \mathbb{E}[b_i^3] + \sum_{\substack{i,j=1 \\ i\neq j}}^{N_i} \mathbb{E}[b_i^2 b_j]+\sum_{\substack{i,j,k=1 \\ i\neq j\neq k}}^{N_i} \mathbb{E}[b_i b_jb_k] \right)\\
            &=\frac{1}{N_i^3} \left(\sum_{i=1}^{N_i} \mathbb{E}[b_i] + \sum_{\substack{i,j=1 \\ i\neq j}}^{N_i} \mathbb{E}[b_j]+\sum_{\substack{i,j,k=1 \\ i\neq j\neq k}}^{N_i} \mathbb{E}[b_i] \mathbb{E}[b_j]\mathbb{E}[b_k] \right)\\
            &= \frac{1}{N_i^3} \left(\sum_{i=1}^{N_i} \braket{P_i} + \sum_{\substack{i,j=1 \\ i\neq j}}^{N_i} \braket{P_i}+\sum_{\substack{i,j,k=1 \\ i\neq j\neq k}}^{N_i} \braket{P_i} \braket{P_i}\braket{P_i} \right) \ ,
        \end{align}
        where we again use the identity $b_i^2=1$ and the i.i.d. assumption in the second line.
        By counting the possible configurations of each summation, we arrive at
        \begin{align}
         \mathbb{E}\left[ X_i^3\right] &=\frac{1}{N_i^3} \left(N_i \braket{P_i} + 3N_i(N_i-1)\braket{P_i}+N_i(N_i-1)(N_i-2)\braket{P_i}^3 \right)\\
         &=\langle P_i\rangle^3\left( 1 - \frac{3}{N_i} + \frac{2}{N_i^2}\right)
+\langle P_i\rangle\left(\frac{3}{N_i} - \frac{2}{N_i^2}\right) \ .
        \end{align}
        \item \underline{\textbf{\textit{Fourth moment of $X_i$:}}}
        
        Lastly, for the fourth moment $\mathbb{E}\left[ X_i^4\right]$, we again split the summation into multiple parts in the last line, i.e.,  we classify the 4--tuple $(i,j,k,l)$ according to the ``equivalence class'' of the four indices.
        \begin{align}
            \mathbb{E}\left[ X_i^4\right] &= \frac{1}{N_i^4} \left(\sum_{i=1}^{N_i} \mathbb{E}[b_i^4] + \sum_{\substack{i,j=1 \\ i\neq j}}^{N_i} \mathbb{E}[b_i^3 b_j]+\sum_{\substack{i,j=1 \\ i\neq j}}^{N_i} \mathbb{E}[b_i^2 b_j^2]+\sum_{\substack{i,j,k=1 \\ i\neq j\neq k}}^{N_i} \mathbb{E}[b_i^2 b_jb_k] +\sum_{\substack{i,j ,k,l=1\\ i \neq j \neq k\neq l}}^{N_i} \mathbb{E}\left[ b_ib_jb_kb_l\right]\right)\\
            &=\frac{1}{N_i^4}\left(\sum_{i=1}^{N_i} 1 + \sum_{\substack{i,j=1 \\ i\neq j}}^{N_i} \mathbb{E}[b_i]\mathbb{E}[ b_j]+\sum_{\substack{i,j=1 \\ i\neq j}}^{N_i} 1+\sum_{\substack{i,j,k=1 \\ i\neq j\neq k}}^{N_i} \mathbb{E}[b_j]\mathbb{E}[b_k] +\sum_{\substack{i,j ,k,l=1\\ i \neq j \neq k\neq l}}^{N_i} \mathbb{E}\left[ b_i]\mathbb{E}[b_j]\mathbb{E}[b_k]\mathbb{E}[b_l\right]\right)\\
            &= \frac{1}{N_i^4}\left(\sum_{i=1}^{N_i} 1 + \sum_{\substack{i,j=1 \\ i\neq j}}^{N_i} \braket{P_i}^2+\sum_{\substack{i,j=1 \\ i\neq j}}^{N_i} 1+\sum_{\substack{i,j,k=1 \\ i\neq j\neq k}}^{N_i} \braket{P_i}^2+\sum_{\substack{i,j ,k,l=1\\ i \neq j \neq k\neq l}}^{N_i} \braket{P_i}^4\right) \ .
        \end{align}
        where we also employ the identity $b_i^2=1$ and the i.i.d. assumption in the second line.
        By accounting for all possible arrangements in each summation, we derive 
        \begin{align}
             \mathbb{E}\left[ X_i^4\right] &=\frac{1}{N_i^4} (N_i+ 4N_i(N_i-1)\braket{P_i}^2+3N_i(N_i-1) \nonumber\\
             &\qquad \qquad+ 6N_i(N_i-1)(N_i-2)\braket{P_i}^2+N_i(N_i-1)(N_i-2)(N_i-3)\braket{P_i}^4 ) \\
             &=\langle P_i\rangle^4+\frac{6\langle P_i\rangle^2\left(1-\langle P_i\rangle^2\right)}{N_i}
            +\frac{\left(11\langle P_i\rangle^2-3\right)\left(\langle P_i\rangle^2-1\right)}{N_i^2}+\frac{2(3\langle P_i\rangle^2-1)(1-\langle P_i\rangle^2)}{N_i^3} \ .
        \end{align}
    \end{enumerate}

\end{proof}

Now, we are ready to present Lem.~\ref{lemma: second variance term}, which is the second term of Eq.~\eqref{eq: var of the unbiased estimator}.

\begin{lemmaA}\label{lemma: second variance term}
Suppose we have an observable $\hat{O}$ which can be decomposed to $\hat{O}=\sum_{i=1}^{L} w_i P_i$, where $P_i\in\{I,X,Y,Z\}^{\otimes n}$ denotes the Pauli operators for $n$ qubits and $w_i$ represents its corresponding weights. 
Assuming that measurements performed for all operators are i.i.d., then the uncertainty (variance) of the square of the estimation of $\hat{O}$ can be expressed as
 \begin{align}
     \text{Var}\left[\overline{\braket{\hat{O}}^2}\right]&= \sum_{i=1}^{L}   \left[\frac{w_i^4\left(1-\langle P_i\rangle^2\right)}{(N_i-1)^2} \times\left(4\langle P_i\rangle^2N_i+ 2(1-\langle P_i\rangle^2) \right)\right] \nonumber\\
     &\quad+\sum_{\substack{i,j=1 \\ i \neq j}}^L \left[w_i^2 w_j^2 \left(\frac{\langle P_i\rangle^2\left(1-\langle P_j\rangle^2\right)}{N_j}
+\frac{\langle P_j\rangle^2\left(1-\langle P_i\rangle^2\right)}{N_i}
+\frac{\left(1-\langle P_i\rangle^2\right)
\left(1-\langle P_j\rangle^2\right)}{N_iN_j}\right) \right]
\nonumber\\
     &\quad+4 \sum_{\substack{i,j=1 \\ i \neq j}}^L \left(\frac{w_i^3w_j}{N_i}\braket{P_i}\braket{P_j}(1-\braket{P_i}^2)\right) \ ,
 \end{align}
where $N_{i}$ is the number of sample used to estimate $\braket{P_{i}}$ for each $1\leq i \leq L $.
\end{lemmaA}
\begin{proof}
We start with the expression
\begin{align}
    &\text{Var}\left[\overline{\braket{\hat{O}}^2}\right] \nonumber\\ 
    &=  \text{Var}\left[ \sum_{i=1}^{L} \frac{w_i^2 N_i}{N_i-1} \left[\left(\frac{1}{N_i}\sum_{k=1}^{N_i} b^{(P_i)}_k\right)^2-\frac{1}{N_i} \right] + \sum_{\substack{i,j=1 \\ i \neq j}}^L w_i w_j \left(\frac{1}{N_i}\sum_{k=1}^{N_i} b^{(P_i)}_k\right)\left(\frac{1}{N_j}\sum_{k=1}^{N_j} b^{(P_j)}_k\right)\right] \\
    &= \text{Var}\left[ \sum_{i=1}^{L} \frac{w_i^2 N_i}{N_i-1} \left[\left(\frac{1}{N_i}\sum_{k=1}^{N_i} b^{(P_i)}_k\right)^2-\frac{1}{N_i} \right]\right] +\text{Var}\left[\sum_{\substack{i,j=1 \\ i \neq j}}^L w_i w_j \left(\frac{1}{N_i}\sum_{k=1}^{N_i} b^{(P_i)}_k\right)\left(\frac{1}{N_j}\sum_{k=1}^{N_j} b^{(P_j)}_k\right) \right] \nonumber\\
    &\quad \quad+ 2 \;\text{Cov}\left[\sum_{i=1}^{L} \frac{w_i^2 N_i}{N_i-1} \left[\left(\frac{1}{N_i}\sum_{k=1}^{N_i} b^{(P_i)}_k\right)^2-\frac{1}{N_i} \right], \sum_{\substack{i,j=1 \\ i \neq j}}^L w_i w_j \left(\frac{1}{N_i}\sum_{k=1}^{N_i} b^{(P_i)}_k\right)\left(\frac{1}{N_j}\sum_{k=1}^{N_j} b^{(P_j)}_k\right)\right] \ ,
\end{align}
where we use the identity $\text{Var}[A+B]=\text{Var}[A]+\text{Var}[B]+2\;\text{Cov}(A,B)$ (for any variable $A,B$).
Before we proceed to evaluate these three terms, let us define the shorthand notation $X_i=\frac{1}{N_i}\sum_{k=1}^{N_i} b^{(P_i)}_k$.
The final expression of these three terms are:
\begin{enumerate}
    \item For the first term, we have
    \begin{align}
         \text{Var}\left[ \sum_{i=1}^{L} \frac{w_i^2 N_i}{N_i-1} \left(X_i^2-\frac{1}{N_i} \right)\right]  =\sum_{i=1}^{L}   \text{Var}\left[ \frac{w_i^2 N_i}{N_i-1} \left(X_i^2-\frac{1}{N_i} \right)\right] 
         &= \sum_{i=1}^{L}   \text{Var}\left[ \frac{w_i^2 N_i}{N_i-1} X_i^2 \right] \\
         &=  \sum_{i=1}^{L}  \frac{w_i^4 N_i^2}{(N_i-1)^2} \text{Var}\left[  X_i^2 \right] \ , \label{eq: variance intermediate}
    \end{align}
    where we assume that $P_i$ and $P_j$ are independent measurement for $i\neq j$ in the first equality and the property $\text{Var} [X+c]=\text{Var} [X]$ (for arbitrary constant $c$) in the second equality.
    Next, by definition we have $\text{Var}\left[  X_i^2 \right] = \mathbb{E}\left[ X_i^4\right]-\mathbb{E}\left[ X_i^2\right]^2$, where
    using Lem.~\ref{lemma: third and fourth moments of X_i} further gives
    \begin{align}
          \text{Var}\left[  X_i^2 \right] &= \langle P_i\rangle^4+\frac{6\langle P_i\rangle^2\left(1-\langle P_i\rangle^2\right)}{N_i}
            +\frac{\left(11\langle P_i\rangle^2-3\right)\left(\langle P_i\rangle^2-1\right)}{N_i^2}+\frac{2(3\langle P_i\rangle^2-1)(1-\langle P_i\rangle^2)}{N_i^3} \\
            & \quad  - \left( \langle P_i \rangle^2 + \frac{1 - \langle P_i \rangle^2}{N_i}\right)^2 \\
&=\frac{4\langle P_i\rangle^2\left(1-\langle P_i\rangle^2\right)}{N_i}
            +\frac{2\left(5\langle P_i\rangle^2-1\right)\left(\langle P_i\rangle^2-1\right)}{N_i^2}+\frac{2(3\langle P_i\rangle^2-1)(1-\langle P_i\rangle^2)}{N_i^3} \ .
    \end{align}
    Using this result, Eq.~\eqref{eq: variance intermediate} simplifies to
    \begin{align}
             &\text{Var}\left[ \sum_{i=1}^{L} \frac{w_i^2 N_i}{N_i-1} \left(X_i^2-\frac{1}{N_i} \right)\right]  
         \nonumber\\ &=  \sum_{i=1}^{L}   \frac{w_i^4 N_i^2}{(N_i-1)^2} \times \nonumber\\
         &\quad\left(\frac{4\langle P_i\rangle^2\left(1-\langle P_i\rangle^2\right)}{N_i}
            +\frac{2\left(5\langle P_i\rangle^2-1\right)\left(\langle P_i\rangle^2-1\right)}{N_i^2}+\frac{2(3\langle P_i\rangle^2-1)(1-\langle P_i\rangle^2)}{N_i^3}\right) \\
            &=\sum_{i=1}^{L}   \frac{w_i^4(1-\langle P_i\rangle^2)}{(N_i-1)^2} \left[4N_i\langle P_i\rangle^2-2\left(5\langle P_i\rangle^2-1\right)+\frac{2(3\langle P_i\rangle^2-1)}{N_i}\right]\\
            &=\sum_{i=1}^{L}   \frac{w_i^4(1-\langle P_i\rangle^2)}{N_i(N_i-1)^2} \left[4N_i^2\langle P_i\rangle^2-10N_i\langle P_i\rangle^2 +2N_i+6\langle P_i\rangle^2-2\right]\\
                &=\sum_{i=1}^{L}   \frac{w_i^4(1-\langle P_i\rangle^2)}{N_i(N_i-1)^2} \left[2\langle P_i\rangle^2\left(2N_i-3\right)\left(N_i-1\right)  +2(N_i-1)\right]\\ 
                &= \sum_{i=1}^{L}   \left[\frac{2w_i^4 }{N_i(N_i-1)} \left(1 + 2(N_i-2)\braket{P_i}^2 -(2N_i-3)\braket{P_i}^4 \right) \right]
                \label{eq: covariance 1st term final} \ .
    \end{align}

    \item For the second term, we have
    \begin{align}\label{eq:variance intermediate 2nd term}
        \text{Var}\left[\sum_{\substack{i,j=1 \\ i \neq j}}^L w_i w_j X_iX_j \right] = \text{Var}\left[\sum_{\substack{i<j }}^L 2 w_i w_j X_iX_j \right]= \sum_{\substack{i<j }}^L 4 w_i^2 w_j^2\text{Var}\left[ X_iX_j \right] \ ,
    \end{align}
    where we use the property $\text{Var}[aX] = a^2\,\text{Var}[X]$.
    Next, by definition of variance, we obtain
    \begin{align}
        \text{Var}\left[ X_iX_j \right]&=\mathbb{E}\left[X_i^2X_j^2\right]-\mathbb{E}\left[X_iX_j\right]^2 =\mathbb{E}\left[X_i^2\right]\mathbb{E}\left[ X_j^2\right] - \mathbb{E}\left[X_i\right]^2\mathbb{E}\left[ X_j\right]^2 \ .
    \end{align}
    Using Lem.~\ref{lemma: third and fourth moments of X_i}, we have
    \begin{align}
        \text{Var}\left[ X_iX_j \right]&= \left(\langle P_i \rangle^2 + \frac{1 - \langle P_i \rangle^2}{N_i}\right)\left(\langle P_j \rangle^2 + \frac{1 - \langle P_j \rangle^2}{N_i}\right)- \braket{P_i}^2\braket{P_j}^2\\
        &= 
\frac{\langle P_i\rangle^2\left(1-\langle P_j\rangle^2\right)}{N_j}
+\frac{\langle P_j\rangle^2\left(1-\langle P_i\rangle^2\right)}{N_i}
+\frac{\left(1-\langle P_i\rangle^2\right)
\left(1-\langle P_j\rangle^2\right)}{N_iN_j} \ .
    \end{align}
    Consequently, Eq.~\eqref{eq:variance intermediate 2nd term} is 
    \begin{align}
        &\text{Var}\left[\sum_{\substack{i,j=1 \\ i \neq j}}^L w_i w_j X_iX_j \right] \nonumber\\ &=4\sum_{\substack{i<j}}^L \left[\frac{w_i^2 w_j^2}{N_iN_j} \left( 
        \left(1-\langle P_i\rangle^2\right)
\left(1-\langle P_j\rangle^2\right)
+ N_i\langle P_i\rangle^2\left(1-\langle P_j\rangle^2\right)
+ N_j\langle P_j\rangle^2\left(1-\langle P_i\rangle^2\right)
\right) \right]
\label{eq: covariance 2nd term final} \ .
\end{align}

    \item  For the third term, it is
    \begin{align}
        &\quad\text{Cov}\left[\sum_{i=1}^{L} \frac{w_i^2 N_i}{N_i-1} \left[X_i^2-\frac{1}{N_i} \right], \sum_{\substack{i,j=1 \\ i \neq j}}^L w_i w_j X_iX_j\right] \nonumber\\
        &=\sum_{\substack{i,j=1 \\ i \neq j}}^L \left(\text{Cov}\left[\frac{w_i^2 N_i}{N_i-1} \left[X_i^2-\frac{1}{N_i}\right],  w_i w_j X_iX_j\right]
        +\text{Cov}\left[\frac{w_i^2 N_i}{N_i-1} \left[X_i^2-\frac{1}{N_i}\right],  w_i w_j X_jX_i\right]\right) \ ,
    \end{align}
where we use the bilinear property of the covariance.
By symmetry, the two covariance contributions are equal and hence we have
    \begin{align}
       \text{Cov}\left[\sum_{i=1}^{L} \frac{w_i^2 N_i}{N_i-1} \left[X_i^2-\frac{1}{N_i} \right], \sum_{\substack{i,j=1 \\ i \neq j}}^L w_i w_j X_iX_j\right] 
        &=2 \sum_{\substack{i,j=1 \\ i \neq j}}^L \left(\text{Cov}\left[\frac{w_i^2 N_i}{N_i-1} \left[X_i^2-\frac{1}{N_i}\right],  w_i w_j X_iX_j\right]
      \right) \label{eq: covariance double sum}\ ,
    \end{align}
        Since $\text{Cov}[A,B]=\mathbb{E}[AB]-\mathbb{E}[A]\;\mathbb{E}[B]$, we have
    \begin{align}
        &\text{Cov}\left[\frac{w_i^2 N_i}{N_i-1} \left[X_i^2-\frac{1}{N_i}\right],  w_i w_j X_iX_j\right]\nonumber\\
        &= \mathbb{E}\left[\left(\frac{w_i^2 N_i}{N_i-1} \left[X_i^2-\frac{1}{N_i}\right]\right)\left( w_i w_j X_iX_j\right)\right]-\mathbb{E}\left[\frac{w_i^2 N_i}{N_i-1} \left[X_i^2-\frac{1}{N_i}\right]\right]
        \mathbb{E}\left[ w_i w_j X_iX_j\right]\\
        &=\frac{w_i^3w_j N_i}{N_i-1}\;\mathbb{E}[X_j] \times\left\{\mathbb{E}[X_i^3]-\mathbb{E}[X_i^2]\;\mathbb{E}[X_i]\right\} \label{eq: covariance intermediate} \ ,
    \end{align}
       Using Lem.~\ref{lemma: third and fourth moments of X_i}, Eq.~\eqref{eq: covariance intermediate} reduces to
      \begin{align}
        &\text{Cov}\left[\frac{w_i^2 N_i}{N_i-1} \left[X_i^2-\frac{1}{N_i}\right],  w_i w_j X_iX_j\right]\nonumber\\
        &=\frac{w_i^3w_j N_i}{N_i-1}\braket{P_j} \times\left\{\langle P_i\rangle^3\left( 1 - \frac{3}{N_i} + \frac{2}{N_i^2}\right)
+\langle P_i\rangle\left(\frac{3}{N_i} - \frac{2}{N_i^2}\right)-\left( \langle P_i \rangle^2 + \frac{1 - \langle P_i \rangle^2}{N_i}\right)\;\braket{P_i}\right\} \\
&=\frac{w_i^3w_j N_i}{N_i-1}\braket{P_j} \times\left\{\ \frac{\langle P_i \rangle^3 (2-2N_i)+2\langle P_i \rangle(N_i-1)}{N_i^2}\right\} =\frac{2w_i^3w_j}{N_i}\braket{P_j}\braket{P_i}(1-\braket{P_i}^2) \ .
    \end{align}
    Therefore, Eq.~\eqref{eq: covariance double sum} is given by
      \begin{align}
       &\text{Cov}\left[\sum_{i=1}^{L} \frac{w_i^2 N_i}{N_i-1} \left[X_i^2-\frac{1}{N_i} \right], \sum_{\substack{i,j=1 \\ i \neq j}}^L w_i w_j X_iX_j\right] 
        =4 \sum_{\substack{i,j=1 \\ i \neq j}}^L \left(\frac{w_i^3w_j}{N_i}\braket{P_j}\braket{P_i}(1-\braket{P_i}^2)\right) \label{eq: covariance 3rd term final} \ .
    \end{align}
\end{enumerate}

 Collecting Eq.~\eqref{eq: covariance 1st term final}, Eq.~\eqref{eq: covariance 2nd term final} and Eq.~\eqref{eq: covariance 3rd term final}, we arrive at the final expression:
 \begin{align}
     \text{Var}\left[\overline{\braket{\hat{O}}^2}\right]&= 2\sum_{i=1}^{L}   \left[\frac{w_i^4 }{N_i(N_i-1)} \left(1 + 2(N_i-2)\braket{P_i}^2 -(2N_i-3)\braket{P_i}^4 \right) \right]\nonumber\\
     &\quad+4\sum_{\substack{i<j}}^L \left[\frac{w_i^2 w_j^2}{N_iN_j} \left( 
        \left(1-\langle P_i\rangle^2\right)
\left(1-\langle P_j\rangle^2\right)
+ N_i\langle P_i\rangle^2\left(1-\langle P_j\rangle^2\right)
+ N_j\langle P_j\rangle^2\left(1-\langle P_i\rangle^2\right)
\right) \right]
\nonumber\\
     &\quad+4 \sum_{\substack{i< j}}^L \left(\frac{w_i^3w_j}{N_i}\braket{P_i}\braket{P_j}(1-\braket{P_i}^2)\right) \ .
 \end{align}
\end{proof}

\begin{theoremA}[Uncertainty of the estimated variance of an observable]
    Suppose we have an observable $\hat{O}$ which is of the form $\hat{O}=\sum_{i=1}^{L} w_i P_i$, where $P_i\in\{I,X,Y,Z\}^{\otimes n}$ denotes the Pauli operators for $n$ qubits and $w_i$ represents its corresponding weights. 
Assuming that measurements performed for all operators including $P_{ij}=P_{i}P_{j}$ are i.i.d., then the uncertainty (variance) of the estimated variance of $\hat{O}$ can be expressed as
\begin{align}
     \text{Var}[\overline{V}]&=\sum_{\substack{i,j=1 \\ i\neq j}}^{L} \frac{w_i^2w_j^2}{N_{ij}} \left(1-\braket{P_{ij}}^2 \right)
     +
     2\sum_{i=1}^{L}   \left[\frac{w_i^4 }{N_i(N_i-1)} \left(1 + 2(N_i-2)\braket{P_i}^2 -(2N_i-3)\braket{P_i}^4 \right) \right]\nonumber\\
     &\quad+4\sum_{\substack{i<j}}^L \left[\frac{w_i^2 w_j^2}{N_iN_j} \left( 
        \left(1-\langle P_i\rangle^2\right)
\left(1-\langle P_j\rangle^2\right)
+ N_i\langle P_i\rangle^2\left(1-\langle P_j\rangle^2\right)
+ N_j\langle P_j\rangle^2\left(1-\langle P_i\rangle^2\right)
\right) \right]
\nonumber\\
     &\quad+4 \sum_{\substack{i< j}}^L \left(\frac{w_i^3w_j}{N_i}\braket{P_i}\braket{P_j}(1-\braket{P_i}^2)\right) \ .
\end{align}
\end{theoremA}
\begin{proof}
Since the set of Pauli operators $\{P_i\}$ and  $P_{ij}=P_iP_j$ are i.i.d. and mutually independent, we obtain $\text{Var}[\overline{V}] = \text{Var}\left[ \overline{\braket{\hat{O}^2}}\right] -  \text{Var}\left[\overline{\braket{\hat{O}}^2}\right]$, where the first and second terms are given by Lem.~\ref{lemma: first variance term} and Lem.~\ref{lemma: second variance term} respectively:
\begin{align}
     \text{Var}[\overline{V}]&=\sum_{\substack{i,j=1 \\ i\neq j}}^{L} \frac{w_i^2w_j^2}{N_{ij}} \left(1-\braket{P_{ij}}^2 \right)
     +
     2\sum_{i=1}^{L}   \left[\frac{w_i^4 }{N_i(N_i-1)} \left(1 + 2(N_i-2)\braket{P_i}^2 -(2N_i-3)\braket{P_i}^4 \right) \right]\nonumber\\
     &\quad+4\sum_{\substack{i<j}}^L \left[\frac{w_i^2 w_j^2}{N_iN_j} \left( 
        \left(1-\langle P_i\rangle^2\right)
\left(1-\langle P_j\rangle^2\right)
+ N_i\langle P_i\rangle^2\left(1-\langle P_j\rangle^2\right)
+ N_j\langle P_j\rangle^2\left(1-\langle P_i\rangle^2\right)
\right) \right]
\nonumber\\
     &\quad+4 \sum_{\substack{i< j}}^L \left(\frac{w_i^3w_j}{N_i}\braket{P_i}\braket{P_j}(1-\braket{P_i}^2)\right) \ .
\end{align}
\end{proof}

\subsection{Alternative Unbiased Method of Estimating Operator Variance}
Here, we provide another way of computing an unbiased estimator of the variance operator.
\begin{lemmaA}\label{lemma: alternative variance estimator}
Suppose we have an observable $\hat{O}$ of the form $\hat{O}=\sum_{i=1}^{L} w_i P_i$, where $P_i\in\{I,X,Y,Z\}^{\otimes n}$ denotes the Pauli operators for $n$ qubits and $w_i$ represents its corresponding weights.
Assuming measurements performed for all operators including $P_{ij}=P_{i}P_{j}$ are i.i.d.,
then the unbiased estimator of the variance operator can be alternatively written as 
\begin{align}
    \overline{V}
    &=\sum_{i=1}^{L} w_i^2 + \sum_{\substack{i,j=1 \\ i\neq j}}^{L} w_iw_j \left(\frac{1}{N_{ij}}\sum_{k=1}^{N_{ij}} b^{(P_{ij})}_k\right) - \sum_{i,j=1}^{L} w_i w_j \left(\frac{1}{N_{(i\otimes j)}}\sum_{i=1}^{N_{(i\otimes j)}}b_{i}^{(P_{i\otimes j})}\right) \ ,
\end{align}
where we perform $N_{ij}$ times measurement on the operator $P_{ij}=P_iP_j$ and $N_{(i\otimes j)}$ times joint measurement on $P_i\otimes P_j$.
\end{lemmaA}
\begin{proof}
    Recall that the variance for an observable $\hat{O}$ is given by
\begin{align}
    V = \braket{\hat{O}^2} - \braket{\hat{O}}^2  &= \sum_{i=1}^{L} w_i^2 \braket{I} + \sum_{\substack{i,j=1 \\ i \neq j}}^L w_i w_j \braket{P_{ij}} - \left(\sum_{i=1}^{L} w_i \braket{P_i}\right)^2 \ .
\end{align}
To estimate the second term $\braket{\hat{O}}^2$, we now perform joint measurements on two copies of quantum states. 
For each independent measurement $\{P_i\otimes P_j\}$, we collects the results of measured bit string $\{b_{i}^{(P_{i\otimes j})}\}$.
Thus, the unbiased estimator of the product $\braket{P_i}\braket{P_j}$ is given by
\begin{align}
    \overline{\braket{P_i}\braket{P_j}}=\frac{1}{N_{(i\otimes j)}}\sum_{i=1}^{N_{(i\otimes j)}}b_{i}^{(P_{i\otimes j})} \ ,
\end{align}
where $N_{(i\otimes j)}$ denotes the number of samples for the measurement $(P_i\otimes P_j)$.
Hence, the term $\braket{\hat{O}}^2$ can be estimated as 
\begin{align}
    \overline{\braket{\hat{O}}^2}=\sum_{i,j=1}^{L} w_i w_j \braket{P_i}\braket{P_j}=\sum_{i,j=1}^{L} w_i w_j \left(\frac{1}{N_{(i\otimes j)}}\sum_{i=1}^{N_{(i\otimes j)}}b_{i}^{(P_{i\otimes j})}\right) \ .
\end{align}
Thus, the unbiased estimator of the variance is 
\begin{align}
    \overline{V}
    &=\sum_{i=1}^{L} w_i^2 + \sum_{\substack{i,j=1 \\ i\neq j}}^{L} w_iw_j \left(\frac{1}{N_{ij}}\sum_{k=1}^{N_{ij}} b^{(P_{ij})}_k\right) - \sum_{i,j=1}^{L} w_i w_j \left(\frac{1}{N_{(i\otimes j)}}\sum_{i=1}^{N_{(i\otimes j)}}b_{i}^{(P_{i\otimes j})}\right) \ .
\end{align}
\end{proof}
Next, we derive the uncertainty of this estimated variance based on this alternative measurement protocol.
\begin{lemmaA}\label{lemma: uncertainty of alternative variance}
Let $\hat{O}$ be an observable of the form $\hat{O}=\sum_{i=1}^{L} w_i P_i$, where $P_i\in\{I,X,Y,Z\}^{\otimes n}$ denotes the Pauli operators for $n$ qubits and $w_i$ represents its corresponding weights.
Assuming the measurements $\{P_{i\otimes j}\}$ and $\{P_{ij}=P_iP_j\}$ are i.i.d. and they are mutually independent to each other,
then the uncertainty (variance) of the estimated variance $\overline{V}$ of the observable is
\begin{align}\label{eq: alternative variance final}
     \text{Var}[\overline{V}] &= \sum_{\substack{i,j=1 \\ i\neq j}}^{L}  \frac{w_i^2w_j^2 \left(1- \braket{P_{ij}}^2 \right)}{N_{ij}}  + \sum_{i,j=1}^{L}  \frac{w_i^2w_j^2 \left(1- \braket{P_{i}}^2 \braket{P_{j}}^2\right)}{N_{(i\otimes j)}} \ ,
\end{align}
where we perform $N_{ij}$ times measurement on the operator $P_{ij}=P_iP_j$ and $N_{(i\otimes j)}$ times joint measurement on $P_i\otimes P_j$.
\end{lemmaA}
\begin{proof}
    First, since the measurements $\{P_{i\otimes j}\}$ and $\{P_{ij}\}$ are i.i.d. and mutually independent, the variance of the operator is 
\begin{align}
    \overline{V}
    &=\sum_{i=1}^{L} w_i^2 + \sum_{\substack{i,j=1 \\ i\neq j}}^{L} w_iw_j \left(\frac{1}{N_{ij}}\sum_{k=1}^{N_{ij}} b^{(P_{ij})}_k\right) - \sum_{i,j=1}^{L} w_i w_j \left(\frac{1}{N_{(i\otimes j)}}\sum_{i=1}^{N_{(i\otimes j)}}b_{i}^{(P_{i\otimes j})}\right) \ .
\end{align}
where we use Lem.~\ref{lemma: alternative variance estimator}.
Expanding the uncertainty (variance) of this estimation gives
\begin{align}
    \text{Var}[\overline{V}] \
    &= \text{Var}\left[\sum_{i=1}^{L} w_i^2 + \sum_{\substack{i,j=1 \\ i\neq j}}^{L} w_iw_j \left(\frac{1}{N_{ij}}\sum_{k=1}^{N_{ij}} b^{(P_{ij})}_k\right) - \sum_{i,j=1}^{L} w_i w_j \left(\frac{1}{N_{(i\otimes j)}}\sum_{i=1}^{N_{(i\otimes j)}}b_{i}^{(P_{i\otimes j})}\right) \right]\\
    &= \text{Var}\left[ \sum_{\substack{i,j=1 \\ i\neq j}}^{L} w_iw_j \left(\frac{1}{N_{ij}}\sum_{k=1}^{N_{ij}} b^{(P_{ij})}_k\right)\right] + \text{Var}\left[ \sum_{i,j=1}^{L} w_i w_j \left(\frac{1}{N_{(i\otimes j)}}\sum_{i=1}^{N_{(i\otimes j)}}b_{i}^{(P_{i\otimes j})}\right) \right] \ ,
\end{align}
where we use the fact that the first term $\sum_{i=1}^{L} w_i^2$ has zero variance in the last line.
Next, as $\text{Var}[aX] = a^2\,\text{Var}[X]$ for any scalar factor $a$, it can be simplified to
\begin{align}
    \text{Var}[\overline{V}] &= \sum_{\substack{i,j=1 \\ i\neq j}}^{L}  \frac{w_i^2w_j^2}{N_{ij}^2} \text{Var}\left[\sum_{k=1}^{N_{ij}} b^{(P_{ij})}_k\right] + \sum_{i,j=1}^{L}  \frac{w_i^2w_j^2}{N_{(i\otimes j)}^2} \text{Var}\left[\sum_{i=1}^{N_{(i\otimes j)}}b_{i}^{(P_{i\otimes j})}\right] \label{eq: alternative variance intermediate}\ .
\end{align}
Since the bit string $b_i$ satisfies $b_i^2=1$, we have the following identity:
\begin{align}
    \text{Var}\left[\sum_{k=1}^{N_{ij}} b^{(P_{ij})}_k\right]
    =\sum_{k=1}^{N_{ij}}\text{Var}\left[ b^{(P_{ij})}_k\right]= \sum_{k=1}^{N_{ij}} \left(1- \braket{P_{ij}}^2 \right) =N_{ij}\left(1- \braket{P_{ij}}^2 \right)\ .
\end{align}
Similarly, for the joint measurements, it becomes
\begin{align}
    \text{Var}\left[\sum_{i=1}^{N_{(i\otimes j)}}b_{i}^{(P_{i\otimes j})}\right]
    =\sum_{i=1}^{N_{(i\otimes j)}}\text{Var}\left[b_{i}^{(P_{i\otimes j})}\right]= \sum_{i=1}^{N_{(i\otimes j)}} \left(1- \braket{P_{i}}^2 \braket{P_{j}}^2\right)=N_{(i\otimes j)} \left(1- \braket{P_{i}}^2 \braket{P_{j}}^2\right) \ .
\end{align}
Therefore, Eq.~\eqref{eq: alternative variance intermediate} becomes
\begin{align}
     \text{Var}[\overline{V}] &= \sum_{\substack{i,j=1 \\ i\neq j}}^{L}  \frac{w_i^2w_j^2 \left(1- \braket{P_{ij}}^2 \right)}{N_{ij}}  + \sum_{i,j=1}^{L}  \frac{w_i^2w_j^2 \left(1- \braket{P_{i}}^2 \braket{P_{j}}^2\right)}{N_{(i\otimes j)}} \ .
\end{align}
\end{proof}
Using this result, we now present a proposition that tells us how many samples we need to achieve precision $\epsilon$ when estimating the variance of the observable $\hat{O}$.
\begin{propositionA}
Suppose we have an observable $\hat{O}$ of the form $\hat{O}=\sum_{i=1}^{L} w_i P_i$, where $P_i\in\{I,X,Y,Z\}^{\otimes n}$ denotes the Pauli operators for $n$ qubits and $w_i$ represents its corresponding weights.
Let us denote $\overline{V}$ as the estimated variance of the observable.
Assume that the measurements $\{P_{i\otimes j}\}$ and $\{P_{ij}\}$ are i.i.d. and mutually independent, the number of samples required to achieve a target precision $\epsilon$ for $\overline{V}$ scales as the fourth power of the Hamiltonian’s L1 norm and quadratically in the inverse of $\epsilon$.

\end{propositionA}
\begin{proof}
First,  note that the total number of the measurements are given by
\begin{align}
    N = \sum_{\substack{i,j=1 \\ i\neq j}}^{L}  N_{ij}+ \sum_{i,j=1}^{L} N_{(i\otimes j)} \ .
\end{align}
In this setting, the optimal allocation of measurement shots can be determined via Lagrange multipliers.
Our objective is to minimize the total number of shots while ensuring that the uncertainty of $ \overline{V}$ remains below the desired precision $\epsilon^2$. 
We follow the approach outlined in Ref.~\cite{rubin2018application}, which provides the optimal allocation of measurement shots for estimating each term of the Hamiltonian.
Note that Ref.~\cite{rubin2018application} demonstrates that the number of measurements required to achieve $\epsilon$-precision is given by $\mathcal{O}(|w|^2/\epsilon^2)$ with $|w|=\sum_{i}^{L}|w_{i}|$.
First, the corresponding Lagrangian $\mathcal{L}$ can be expressed as
\begin{align}\label{eq: measurement no Lagrangian}
    \mathcal{L} = \sum_{\substack{i,j=1 \\ i\neq j}}^{L}  N_{ij}+ \sum_{i,j=1}^{L} N_{(i\otimes j)}+\lambda \left(\text{Var}\left[\overline{V_{N}}\right] - \epsilon^2\right).
\end{align}
According to Lem.~\ref{lemma: uncertainty of alternative variance}, we have
\begin{align}
     \text{Var}[\overline{V}] &= \sum_{\substack{i,j=1 \\ i\neq j}}^{L}  \frac{w_i^2w_j^2 \left(1- \braket{P_{ij}}^2 \right)}{N_{ij}}  + \sum_{i,j=1}^{L}  \frac{w_i^2w_j^2 \left(1- \braket{P_{i}}^2 \braket{P_{j}}^2\right)}{N_{(i\otimes j)}} \ .
\end{align}
Using this result, we can proceed to evaluate $\mathcal{L}$.
By taking the derivative of $\mathcal{L}$, Eq.~\eqref{eq: measurement no Lagrangian} yields
\begin{align}
    \frac{\partial \mathcal{L}}{\partial N_{ij}} &= 1 - \lambda \frac{w_i^2 w_j^2 \left(1 - \braket{P_{(i,j)}}^2\right)}{N_{ij}^2} \ ,\qquad \quad
    \frac{\partial \mathcal{L}}{\partial N_{(i \otimes j)}} = 1 - \lambda \frac{w_i^2 w_j^2 \left(1 - \braket{P_{i}}^2\braket{P_{j}}^2\right)}{N_{(i\otimes j)}^2} \ .
\end{align}
To obtain zero derivatives, we require
\begin{align}
    N_{ij} &= |w_i|\,|w_j|\sqrt{\lambda \left(1 - \braket{P_{(i,j)}}^2\right)}\,  \ , \qquad \quad
    N_{(i \otimes j)} = |w_i|\,|w_j|\sqrt{\lambda \left(1 - \braket{P_{i}}^2\braket{P_{j}}^2\right)} \label{eq: lagrange no 2}\ . 
\end{align}
Recall that we set the target precision to be $\epsilon$, i.e. we would like to achieve $\text{Var}\left[\overline{V_{N}}\right] = \epsilon^2$ and hence Eq.~\eqref{eq: alternative variance final} yields
\begin{align}
    \epsilon^2 &=\sum_{\substack{i,j=1 \\ i\neq j}}^{L}  \frac{w_i^2w_j^2 \left(1- \braket{P_{ij}}^2 \right)}{N_{ij}}  + \sum_{i,j=1}^{L}  \frac{w_i^2w_j^2 \left(1- \braket{P_{i}}^2 \braket{P_{j}}^2\right)}{N_{(i\otimes j)}}\\
    &=\frac{1}{\sqrt{\lambda}} \left(\sum_{\substack{i,j=1 \\ i\neq j}}^{L} |w_i|\,|w_j| \sqrt{\left(1 - \braket{P_{ij}}^2\right)} + \sum_{i,j=1}^L |w_i|\,|w_j| \sqrt{\left(1 - \braket{P_{i}}^2\braket{P_{j}}^2\right)}\right) \ .
\end{align}
where we substitute back Eq.~\eqref{eq: lagrange no 2} to obtain last line.
Thus, we have
\begin{align}
    \sqrt{\lambda} = \frac{1}{\epsilon^2} \left(\sum_{\substack{i,j=1 \\ i\neq j}}^{L} |w_i|\,|w_j| \sqrt{\left(1 - \braket{P_{(i,j)}}^2\right)} + \sum_{i,j=1}^L |w_i|\,|w_j| \sqrt{\left(1 - \braket{P_{i}}^2\braket{P_{j}}^2\right)}\right) \ .
\end{align}
Finally, the optimal number of total measurements are
\begin{align}
\begin{split}
     N =\sum_{\substack{i,j=1 \\ i\neq j}}^{L}  N_{ij}+ \sum_{i,j=1}^{L} N_{(i\otimes j)} 
     &= \sqrt{\lambda}\left(\sum_{\substack{i,j=1 \\ i\neq j}}^{L}  |w_i|\,|w_j|\sqrt{ \left(1 - \braket{P_{(i,j)}}^2\right)}+ \sum_{i,j=1}^{L} |w_i|\,|w_j|\sqrt{\left(1 - \braket{P_{i}}^2\braket{P_{j}}^2\right)}\right)\\
     &=  \frac{1}{\epsilon^2} \left(\sum_{\substack{i,j=1 \\ i\neq j}}^{L} |w_i|\,|w_j| \sqrt{\left(1 - \braket{P_{(i,j)}}^2\right)} + \sum_{i,j=1}^L |w_i|\,|w_j| \sqrt{\left(1 - \braket{P_{i}}^2\braket{P_{j}}^2\right)}\right)^2 \ .
\end{split}
\end{align}
Since $b_i \in \{\pm1\}$, $\braket{P_i}\leq 1$ for all $i$.
The total number of measurements is then upper bounded by
\begin{align}
    N \leq  \frac{1}{\epsilon^2} \left(\sum_{\substack{i,j=1 \\ i\neq j}}^{L} |w_i|\,|w_j| + \sum_{i,j=1}^L |w_i|\,|w_j| \right)^2 
    &\leq \frac{1}{\epsilon^2} \left(2 \sum_{i,j=1}^L |w_i|\,|w_j| \right)^2 \\
    &= \frac{4}{\epsilon^2} \left( \sum_{i=1}^L |w_i|\right)^2 \left(\sum_{j=1}^L |w_j| \right)^2 = \frac{4}{\epsilon^2} \left( \sum_{i=1}^L |w_i|\right)^4 \ .
\end{align}
So, the proposition's statement is justified.
\end{proof}

 \end{document}